\PassOptionsToPackage{svgnames}{xcolor}
\documentclass[11pt]{article}

\usepackage{amsmath,amssymb,amsfonts,amsthm,setspace,tikz,natbib,dsfont,wrapfig,lipsum, array,colortbl,mathrsfs,pifont,wasysym}

\allowdisplaybreaks

\usetikzlibrary{patterns}

\usepackage{thmtools,thm-restate}

\usepackage[a4paper]{geometry}
\usepackage[colorlinks=true,urlcolor=DarkBlue,linkcolor=black,citecolor=DarkBlue]{hyperref}

\usetikzlibrary{arrows,shapes,positioning}
\usetikzlibrary{decorations.markings}
\tikzstyle arrowstyle=[scale=1]
\tikzstyle directed=[postaction={decorate,decoration={markings,
    mark=at position .65 with {\arrow[arrowstyle]{stealth}}}}]
\tikzstyle reverse directed=[postaction={decorate,decoration={markings,
    mark=at position .65 with {\arrowreversed[arrowstyle]{stealth};}}}]
	
\usepackage{fancyhdr,ifthen}
\fancypagestyle{appendix}{
\fancyhead{}
\fancyhead[R]{\ifthenelse{\isodd{\value{page}}}{\thepage}{}}
\fancyhead[L]{\ifthenelse{\isodd{\value{page}}}{\textsc{}}{\thepage}}

\fancyfoot{}
}

\newtheoremstyle{mytheoremstyle} 
    {\topsep}                    
    {\topsep}                    
    {\itshape}                   
    {}                           
    {\sc}                   
    {.}                          
    {.5em}                       
    {}  

\theoremstyle{mytheoremstyle}

\newtheorem{definition}{Definition}

\newtheorem{lemma}{Lemma}

\newtheorem{corollary}{Corollary}

\usepackage{titlesec}

\titleformat*{\section}{\large\scshape\centering}
\titleformat*{\subsection}{\scshape\centering}
\titleformat*{\subsubsection}{\itshape}
\titleformat*{\paragraph}{\large\bfseries\centering}
\titleformat*{\subparagraph}{\large\bfseries\centering}

\title{Heterogeneously Perceived Incentives in Dynamic Environments:\\[0.3ex]
Rationalization, Robustness and Unique Selections\thanks{\today. Contact: \href{mailto:evan.piermont@rhul.ac.uk}{\tt evan.piermont@rhul.ac.uk} (Piermont), \href{mailto:pzuazogarin@hse.ru}{\tt p.zuazogarin@hse.ru} (Zuazo-Garin)}
}

\author{
\begin{tabular}{c}
Evan Piermont\\
\small Royal Holloway--University of London
\end{tabular}
\hspace{0.8cm}
\begin{tabular}{c}
Peio Zuazo-Garin\\
\small HSE University--ICEF
\end{tabular}
}

\date{}

\begin{document}

\maketitle

\begin{abstract}
\footnotesize
In dynamic settings each economic agent's choices can be revealing of her private information. This elicitation via the rationalization of observable behavior depends each agent's perception of which payoff-relevant contingencies other agents persistently deem as impossible. We formalize the potential heterogeneity of these perceptions as disagreements at higher-orders about the set of payoff states of a dynamic game. We find that apparently negligible disagreements greatly affect how agents interpret information and assess the optimality of subsequent behavior: When knowledge of the state space is only `almost common', strategic uncertainty may be greater when choices are rationalized than when they are not--- \textit{forward} and \textit{backward} induction predictions, respectively, and while backward induction predictions are robust to small disagreements about the state space, forward induction predictions are not. We also prove that forward induction predictions \textit{always} admit unique selections \`{a} la \cite{weinstein-07} (also for spaces \textit{not} satisfying richness) and backward induction predictions do not.

\vspace{0.2cm}

\end{abstract}


\section{Introduction}
\label{section:introduction}

On October 26th, 1597, 13 vessels commanded by Joseon Kingdom admiral Yi Sun-sin stood in front of a Japanese fleet of over 120 warships in the Myeongnyang Strait, on the southwestern coast of the current Republic of Korea.
 Aware of the disparity of forces, Yi planned the location of the battle carefully and Myeongnyang was chosen due to its narrowness, which would preclude his small fleet to be vastly outnumbered once attacked, and its unique tidal conditions, with currents that strongly flow first to the North and abruptly switch southward in three-hour intervals. With both fleets at Myeongnyang, the favorable current and their formidable numerical superiority led the Japanese to attack. Yi's army stood in the Northern end of the strait, where the northward flow was already calmer. Together with the narrowness of the strait, this enabled the Joseon fleet to repeal the attack until the direction of the currents switched. The unexpected reversion left the Japanese warships unable to maneuver, drifting backward and colliding among themselves, and the initial violent impulse of the new current in the Northern end helped the Joseon vessels strike their shaky enemy and sink 30 warships, eventually leading to the retreat of the Japanese fleet and the collapse of their campaign.

The tidal configuration of the strait, that is, the \textit{state} (of, literally, nature), was a key determinant of the outcome of the conflict. Joseon's commanders were aware of it, Japan's were not; Joseon's army based their strategic decisions on their knowledge of the state, Japan's army based their strategic decisions persistently ignoring the state: They did not conceive the possibility of such an disadvantageous contingency, not even \textit{after observing} that their enemy, who they knew to be familiar with the terrain and led by a commander of extraordinary skill, decided to fight in so extremely unbalanced conditions. The history relies thus on two elements every economist is familiar with. On the one hand, (nonstrategic) contingencies surrounding the interaction can affect outcomes, making economic agents eager to anticipate them; on the other, agents can update their initial assessments about these contingencies by eliciting others' private information via the rationalization of observed behavior. For example, the gains from bargaining depend on the valuation parties assign to each agreement, and the initial sequence of offers is often aimed at outguessing the valuation of the other party and hiding one's own valuation; the optimality of the price a financial trader offers depends on the risk-costs that govern other traders' supply and demand distributions, and the trader can make inferences about these costs by making sense of others' trading patterns.\footnote{See \cite{kennan-01} and \cite{sannikov-16}, respectively.} Like at Myeongnyang, the success of rationalization relies on correctly assessing the possible contingencies others have in mind. The possible heterogeneity is incommensurable, but standard economic modeling \textit{systematically} assumes that, despite different initial probabilistic assessments of each contingency, agents \textit{commonly agree} on the set of all the possible contingencies, so that any state, even if initially ignored, can be considered upon observed behavior.

This paper introduces a general, tractable theory of \textit{persistent disagreements} about the possible payoff-relevant, nonstrategic contingencies involved in a dynamic interaction. In it, we show that these disagreements, no matter how apparently negligible, can have a drastic impact on the way information is inferred from observed behavior and, in consequence, can revert our insights about some key features of standard economic predictions. Our findings cast doubts on the methodological validity of equilibrium refinements based on \textit{forward induction} arguments (i.e., the idea that observed behavior is informative of private information and future plans): Arbitrarily small discrepancies regarding the set of possible contingencies can lead to an erratic interpretation of behavior that results in `anything goes'. Predictions based on forward induction are extremely sensitive to modeling details, to the extent that arbitrarily small misspecifications about what agents agree on can lead to 
unexpected predictions ignored in the benchmark model. This lack of robustness vanishes if, instead, predictions are based on backward induction; that is, if agents are assumed to render observed erratic behavior as a mistake and not as a product of unexpected private information and a signal of future intentions. Finally, we present an contagion argument \`{a} la global games that allows for creating unique selections for every forward induction outcome.

More specifically, our first contribution is methodological and consists in formalizing the idea that the set of payoff states of a dynamic game might not be commonly known. This is attained by endowing each player with a \textit{subjective payoff structure}, which is a hierarchical construction whose first component is the set of states considered by the player, whose second component is the set of states the player understands her opponents to consider, and so on.

Unlike initial belief-hierarchies about the set of states, these subjective payoff structures are assumed to remain invariant in the player's mind as the game unfolds, and place limits to how the player's initial belief-hierarchy can be updated in response to observations throughout the game. We refer to a pair consisting of a subjective payoff structure and an initial belief-hierarchy as a \textit{subjective model}, and to the profiles of subjective models in which players share common knowledge of a set of states, as a \textit{standard model}.\footnote{Thus, in a sense, we reinterpret the set of states as a \textit{persistent belief} instead of an objective, material object. This way, a model captures the dichotomy between the beliefs that can be updated, the initial ones, and the ones that will never be abandoned. In this respect, our definition of a subjective payoff structure goes hand by hand with the idea of $\Delta$-restrictions \`{a} la \cite{battigalli-03}, with the novelty that these restrictions may be not commonly known.} In this sense, and in line with the `Wilson doctrine', our approach consists in tearing apart common knowledge assumptions and increasing the level of subjectivity in the formalization of a game:\footnote{As recalled by \cite{chung-07}, \cite{wilson-87} writes: ``Game theory has a great advantage in explicitly analyzing the consequences of trading rules that presumably are really common knowledge; it is deficient to the extent it assumes other features to be common knowledge, such as one agent's probability assessment about another's preferences or information. [...] I foresee the progress of game theory as depending on successive reduction in the base of common knowledge required to conduct useful analyses of practical problems. Only by repeated weakening of common knowledge assumptions will the theory approximate reality.''} While we maintain common knowledge of the \textit{extensive form} of the game (i.e., the rules: Players, turns, available actions, information feedback, etc.), the \textit{payoff structure} (i.e, the information regarding preferences) is completely subjective, possibly representing delusion, and envisioned as a collection of element of a universal construction which is trivially commonly known, namely, the set of all possible utility functions. We would like to be emphatic about how natural this separation of rules and stakes is: While the possible outcomes of an economic interaction can be objectively described, there is nothing intrinsically objective in the notion of \textit{utility}; the allocation of resources that results from each combination of reported preferences can be commonly known, the market shares that correspond to each different combination of pricing strategies can be commonly known, etc. But each individual agent's preference over material outcomes is a personal trait of hers. In consequence, it seems reasonable to remove the set of states, or possible utilities, from the objective (in the sense of `commonly believed') part of the game.\footnote{Of course, the subjectivization of payoff-functions is standard practice since \cite{harsanyi-67}; the formalization in this paper goes beyond that: The state space being subjectively perceived implies that the ways in which beliefs can be updated is subjective as well.} By doing so, our approach allows for great heterogeneity regarding players' interactive perception, to the extent that the standard case of a commonly known set of states turns out to be a \textit{knife-edge} situation, and important established facts and understandings concluded from this kind of modeling, extremely fragile.

The first of such understandings is the insight that the outcomes of forward induction reasoning, usually formalized by \textit{extensive form rationalizability} (\citealp{pearce-84}, \citealp{battigalli-97}), refine those of backward induction. This is known to be the case for games with complete information, both for the case of perfect information (\citealp{battigalli-97}, \citealp{heifetz-15} and \citealp{perea-18}) and imperfect information (\citealp{chen-13}, \citealp{perea-18b} and \citealp{catonini-20}). Whether an analogous result holds for the case of incomplete information remains an open question (backward induction being formalized here as an interim version of \citeauthor{penta-11}'s \citeyearpar{penta-11, penta-15} \textit{backward rationalizability}); however, our counterexamples establish that the answer is \textit{negative} when the set of payoff-states of the dynamic game is not commonly known, no matter how small players' discrepancies are.\footnote{While we did not attempt to prove this generalization, we conjecture that the refinement result should hold for dynamic games with commonly known sets of payoff-states. In our counterexamples for the case of not commonly known state spaces, it is precisely the lack of common knowledge what seems the key driving force.} The key idea that triggers the observation is that when the state space is not commonly known, it can sometimes be possible for a player to choose rationally an action the other player cannot rationalize (the action may be undominated given the state space the first player considers, but not in the state space the first player thinks the second player considers). Thus, the second player may entertain any arbitrary conjecture about the future behavior of the first one, and this is something the first player can rationally exploit. The observation has some intrinsic interest. While the literature on strategic communication teaches that agents can have incentives to act is if they had less information than they do,\footnote{\cite{dye-85}, for instance.} in standard game-theoretic models it cannot be that a player rationally pretends not to be rational: If she did so, her opponents would be able to find ways to rationalize her behavior. Subjective payoff structures
circumvent this impossibility in a natural manner.\footnote{The strategy of \textit{playing dumb} is a commonplace in popular culture; as \cite{suetonious-07} writes about Emperor Claudius: ``He did not even keep quiet about his own stupidity, but in certain brief speeches he declared that \textit{he had purposely feigned it} under [Caligula], because otherwise he could not have escaped alive and attained his present station. But he convinced no one, and within a short time a book was published, the title of which was `The Elevation of Fools and its thesis, that no one feigned folly.'' The italics are ours.}

The second understanding that collapses is the robustness of rationalizability to small misspecifications of players' perception of payoff structures. \cite{dekel-07} show that in static settings the predictions of \textit{interim} (\textit{correlated}) \textit{rationalizability} are \textit{robust} to small misspecifications of players' higher-order beliefs about the payoff-states: Predictions excluded in the benchmark model do not to arise in slightly misspecified models (i.e., the solution concept is upper hemicontinuous on players' belief-hierarchies). We partially extend this insight by proving in Proposition \ref{proposition:UHCtypes} that the correspondence that describes extensive form rationalizability is upper hemicontinuous on initial belief-hierarchies.\footnote{While we take for granted that this result is `folk knowledge', to the best of our knowledge this is the first available proof of the result.} However, the conclusion that extensive form rationalizability is robust in a general sense would be erroneous. Our examples show that the solution concept is not robust to misspecifications of \textit{models}. It is possible to perturb a dynamic game with a commonly known set of payoff-states and unique extensive form rationalizable outcome so that, for every arbitrarily close by model, extensive form rationalizable outcomes are multiple. These observations shed important light on two aspects of the role of the state space in rationalization:

\begin{itemize}

\item First, \cite{penta-12} and \cite{chen-12} show that in dynamic games with sufficiently \textit{rich} set of payoff-states (i.e., state spaces in which for every strategy of every player there exists some payoff-state for which the strategy is strictly dominant), no refinement of interim rationalizability is robust to perturbations of initial beliefs.\footnote{\citeauthor{penta-12}'s \citeyearpar{penta-12} result is slightly stronger. By allowing for payoff-states about which players may have persistent private information, his result pertains not plain vanilla interim rationalizability, but interim \textit{sequential} rationalizability, a solution concept that refines interim rationalizability by employing sequential rationality as a notion of optimal behavior, instead of plain rationality. For expositional purposes, throughout the paper we refer solely to `interim rationalizability' and bundle \citeauthor{penta-12}'s \citeyearpar{penta-12} and \citeauthor{chen-12} \citeyearpar{chen-12} work together: The moral of both papers is that no refinement of the predictions given by solution concepts characterizing the behavioral implications of rationality and (only) \textit{initial} common belief in rationality is robust.} Our Proposition \ref{proposition:UHCtypes} shows that extensive form rationalizability is robust in this sense. Since extensive form rationalizability refines interim rationalizability, these results together imply that in games where the state space is rich, interim rationalizability and extensive form rationalizability \textit{coincide} (Corollary \ref{corollary:ISR}). The intuition is obvious. If for every strategy there is some state that makes that strategy dominant then every observed behavior can be rationalized; thus, there is nothing rationalization helps to falsify. In consequence, while richness assumptions are innocuous in static games (specifying an interim belief makes every state ignored by the belief strategically irrelevant), our results highlight their problematic nature in dynamic environments.

\item Second, the failure of upper hemicontinuity of extensive form rationalizability on subjective payoff structures reveals the high sensitivity that rationalization processes display w.r.t~common knowledge assumptions. On the one hand, as soon as the knife-edge supposition that players commonly agree about the state space is abandoned, the rationalization a player makes might be erroneous, or a player may fail to rationalize some observed behavior which was indeed rational.
 In consequence, despite the conceptual appeal of forward induction, the conditions for it to be applied successfully seem to crucially rely on knife-edge assumptions that are tremendously demanding, and hard to meet in practice.

\end{itemize}
We end our analysis of robustness to misspecifications of models by noting that there is a standard solution concept which is robust. Proposition \ref{proposition:UHCall} shows that the individual strategies given by backward induction, formalized here as an interim version of \citeauthor{penta-11}'s \citeyearpar{penta-11,penta-15} \textit{backward rationalizability}, are upper hemicontinuous in the whole space of subjective models. The key difference with extensive form rationalizability lies in the fact that, in the latter, the 
restrictions on beliefs are placed only on the histories that are reachable given the subjective payoff structure. Since this set of histories is not, in general, lower hemicontinuous on the subjective payoff structure, upper hemicontinuity of extensive form rationalizability collapses (that restrictions are satisfied along a sequence cannot guarantee that they are also satisfied in the limit). 
On the contrary, backward induction reasoning places restrictions on beliefs at \textit{every} history, and thus, in a way, independent of the subjective payoff structure. We end by noting that, since backward rationalizability also refines interim rationalizability, in dynamic games satisfying richness backward rationalizability, extensive form rationalizability and interim rationalizability all coincide (Corollary \ref{corollary:BR}), what further reinforces our critique of richness assumptions in dynamic games.

Finally, we present a partial structure theorem for extensive form rationalizability in the line of the seminal result by \cite{weinstein-07}. Theorem \ref{theorem:uniqueselection} shows that for every standard model, every \textit{outcome} consistent with extensive form rationalizability is the unique one also consistent with extensive form rationalizability for a sequence of profiles of perturbed subjective models. Notably, the theorem does not require the state space of the benchmark standard model to satisfy richness: Even if it does not, it can be approximated by profiles of subjective payoff structures that do satisfy richness at increasingly higher-orders (a player's space may not be rich, but she may think that her opponents' are, or may think that her opponents think that their opponents' are, etc.; indeed, we show in Proposition \ref{proposition:genericrichness} that this notion of \textit{higher-order} richness is generic in the space of all subjective payoff structures). Thus, our structure theorem applies to virtually all finite dynamic games with incomplete information, unlike those by \cite{penta-12} and \cite{chen-12}, that are restricted to games whose state space satisfies richness (and in which, as seen above, rationalization has no bite). However, while its proof is technically challenging, the conceptual consequences of Theorem \ref{theorem:uniqueselection} remain unclear. On the one hand, unlike \cite{weinstein-07}, \cite{penta-12} and \cite{chen-12} we cannot claim that the predictions of the solution concept we study are \textit{generically} unique: Our counterexamples document that a standard model can be approximated by a sequence of profiles of perturbed subjective models in which multiplicity of extensive form rationalizable outcomes is robust. On the other hand, it is not clear which relevant solution concept, if any, is a refinement of extensive form rationalizability, so that it is not clear how the theorem could be applied to deny the robustness of such hypothetical solution concept.

Nonetheless, the result suggests two possible lines of research. First, it leaves the characterization of the sharpest predictions robust to misspecifications of models as an open question. While the predictions of backwards rationalizability are robust, we document via a counterexample that they do not always admit unique selections. We conjecture then that, if a robust solution concept admitting unique selections exists, it must consist in some hybrid reasoning procedure combining elements of both backward and extensive-form rationalizability. Second, the proof of the theorem relies on a contagion argument that elaborates the classical one in global games and involves an explicit interplay between initial assessments, persistent information and the interpretation of observed behavior. It remains then to be explored whether such an interplay, slightly more complex than the standard one in dynamic global games, provides a rationale for the arousal of endogenous coordination.

%
%
%
%

\vspace{0.5cm}

The rest of the paper is structured as follows. Section \ref{section:preliminaries} formalizes the basic elements for the analysis of dynamic games with incomplete information. Section \ref{section:robsutnessmodel} presents a tool to formalize players' disagreements of arbitrary high-order about the set of payoff-states of the game, and extends the definition of extensive form rationalizability and backward rationalizability to account for this feature. Section \ref{section:results} presents the main findings of the paper. Section \ref{section:conclusions} ends with some literature review and new directions for research. All the proofs are relegated to the appendix.


\section{Preliminaries}
\label{section:preliminaries}

The formalization of dynamic games in Section \ref{subsection:dynamicgames} is standard (see \citealp{penta-12}), except for a small addition in Section \ref{subsubsection:canonicalpayoffstructures} concerning the representation of the utility functions and the information thereof players entertain. Section \ref{subsection:typeshierarchies} reviews the formalization of initial types (or belief-hierarchies) and is completely standard (see \citealp{harsanyi-67}, or \citealp{mertens-85}).


\subsection{Dynamic games}
\label{subsection:dynamicgames}

A \textit{dynamic game} consists in a list $\langle\Gamma,\Upsilon\rangle$ where $(i)$ $\Gamma$, an \textit{extensive form}, specifies the set of players, admissible sequences of choices, turns and past choices observed by each player at each turn, and $(ii)$ $\Upsilon$, a \textit{payoff structure}, formalizes players' possible preferences over the outcomes of the game and the information players may have about each others' preferences.


\subsubsection{Extensive forms}
\label{subsubsection:extensive forms}

An \textit{extensive form} consists in a list $\Gamma=\left\langle I, (A_{i})_{i\in I}, H, Z\right\rangle$ where $I$ is s finite set of \textit{players}, and:

\begin{itemize}

\item For each player $i$, $A_{i}$ is a finite set of \textit{actions}. A \textit{history} represents the unfolding of the game and consists in a finite sequence of possibly simultaneous choices, i.e., on a finite sequence of elements from $\{h^{0}\}\cup A$, where $A:=\bigcup_{J\subseteq I}A_{J}$ and $A_{J}:=\prod_{i\in J}A_{i}$ for any $J\subseteq I$. We say that history $h'$ \textit{follows} history $h$, denoted by $h\preceq h'$, if $h'$ obtains from adding finitely many possibly simultaneous choices to $h$.\footnote{That is, when there exists some $(a^{n})_{n\leq N}\subseteq A$ such that $h'=(h;(a^{n})_{n\leq N})$.}

\item $H$ and $Z$ are finite and disjoint sets of histories such that $(H\cup Z,\preceq)$ is a rooted and oriented tree with terminal nodes $ Z$. Symbol $h^{0}$ denotes the ex ante stage of the game, i.e., the root of the tree, and histories in $H$ and $Z$ are referred to as \textit{partial} and \textit{terminal}, respectively. For any player $i$ and partial history $h$, let $A_{i}\left(h\right)$ denote the set of actions \textit{available} to $i$ at $h$. Player $i$ is \textit{active} at $h$ if $A_{i}(h)$ is nonempty; let $H_{i}$ denote the set of these histories. We assume that: $(i)$ a player is never the only active one twice in a row, and $(ii)$ whenever a player is active, at least two actions are available to her.

\end{itemize}
In this context, the set of player $i$'s \textit{strategies} is $S_{i}:=\prod_{h\in H_{i}}A_{i}(h)$ and, as usual, the set of strategy \textit{profiles} is denoted by $S:=\prod_{i\in I}S_{i}$ and the set of player $i$'s opponents strategies, by $S_{-i}:=\prod_{j\neq i}S_{j}$. Obviously, for each partial history $h$ each strategy $s$ induces a unique terminal history, $z(s|h)$. Finally, let $S_{i}(h)$ and $S_{-i}(h)$ denote, respectively, the set of player $i$'s strategies and the set of $i$'s opponents' strategies that reach partial history $h$, and $H_{i}(s_{i})$, the set of player $i$'s histories that can be reached when she chooses strategy $s_{i}$.\footnote{To be precise, $S_{i}(h)=\{s_{i}\in S_{i}|h\preceq z(s_{-i};s_{i}|h^{0})\textup{  for some  }s_{-i}\in S_{-i}\}$ and $S_{-i}(h)=\prod_{j\neq i}S_{j}(h)$ on the one hand, and $H_{i}(s_{i}):=\{h\in H_{i}|s_{i}\in S_{i}(h)\}$ on the other.}


\subsubsection{Payoff structures}
\label{subsubsection:payoffstructures}

A \textit{standard} \textit{payoff structure} (for extensive form $\Gamma$) consists of a list $\Upsilon:=(\Theta_{0},(\Theta_{i},u_{i})_{i\in I})$, where:

\begin{itemize}

\item $\Theta_{0}$ is a compact and metrizable set of \textit{states of nature} and, for each player $i$, $\Theta_{i}$ is a compact and metrizable set of \textit{payoff types}. We denote the set of \textit{payoff states} by $\Theta:=\Theta_{0}\times\prod_{i\in I}\Theta_{i}$ and, for each player $i$, the set of profiles of player $i$'s opponents' payoff types by $\Theta_{-i}:=\prod_{j\neq i}\Theta_{j}$.

\item For each player $i$, $u_{i}:Z\times\Theta\rightarrow\mathds{R}$ is player $i$'s continuous \textit{utility function}.\footnote{Throughout the paper finite sets are endowed with the discrete topology, $\mathds{R}$ is endowed with the Euclidean topology, product sets are endowed with the product topology and spaces consisting of closed sets, with the Hausdorff metric. For a given topological space $X$, its corresponding set of measures over the Borel $\sigma$-algebra is endowed with the weak$\ast$ topology. Unless stated otherwise, measurability always refers to Borel measurability.}

\end{itemize}
The set of all standard payoff structures (for extensive form $\Gamma$) is denoted by $\mathscr{P}$.\footnote{A more general formalization of standard payoff structures allows for the set of payoff states not being a Cartesian product. We note that the whole theory developed in the paper extends to that case by introducing the necessary obvious changes.} Thus, within a dynamic game with incomplete information $\langle\Gamma,\Upsilon\rangle$, each payoff state $\theta\in\Theta$ fully characterizes the profile of players' utility functions: $(u_{i}(\,\cdot\,,\theta))_{i\in I}$. Besides, each player $i$ is assumed to \textit{know} her payoff type $\theta_{i}$ and (possibly) face uncertainty about the rest of components of $\theta$. Hence, besides the description of preferences, each payoff state $\theta$ also determines the information players have about the payoff states: In every moment of the game, a player $i$ with payoff type $\theta_{i}$ \textit{knows} that the true payoff state is some element in $\{\theta_{i}\}\times\Theta_{0}\times\Theta_{-i}$; in consequence, $\{\theta_{i}\}\times\Theta_{0}\times\Theta_{-i}$ is, precisely, the set of payoff states that player $i$ can eventually deem as possible as the game unfolds.


\subsubsection{Canonical representation and convergence of payoff structures}
\label{subsubsection:canonicalpayoffstructures}

The questions we explore throughout the paper are related to the strategic impact of \textit{small} relaxation of players' common knowledge of an information structure $\Upsilon$. We develop this point in detail in Section \ref{subsection:subjectiveness}, but before introducing it we need to first clarify what is understood as a \textit{small}, or perturbation, in a standard payoff structure. The fact that the sets of payoff states can vary vastly across different standard payoff structures and be of radically different nature presents some inconvenience for an analyst interested in assessing how \textit{close} two different standard payoff structures are form each other. To circumvent this issue we introduce a canonical representation that allows for envisioning \textit{every} standard payoff structure as part of an
object directly derived from the extensive form:

\begin{definition}[Canonical representation of payoff structures]
\label{definition:canonical}
Let $\Gamma$ be an extensive form and $\Upsilon=(\Theta_{0},(\Theta_{i},u_{i})_{i\in I})$, a payoff structure. Then, the \textup{canonical representation} of $\Upsilon$ is the list $\mathcal{C}(\Upsilon):=\mathcal{C}_{0}(\Upsilon)\times\prod_{i\in I}\mathcal{C}_{i}(\Upsilon)$, where:

\begin{enumerate}

\item[$(i)$] $\mathcal{C}_{0}(\Upsilon)$ is the set of possible profiles of utility functions:
\[
\mathcal{C}_{0}(\Upsilon):=\{(u_{i}(\,\cdot\,,\theta))_{i\in I}|\,\theta\in \Theta\}.
\]

\item[$(ii)$] For each player $i$, $\mathcal{C}_{i}(\Upsilon)$ is the set of profiles of utility functions that player $i$ can consider possible as the game unfolds:
\[
\mathcal{C}_{i}(\Upsilon):=\left\{\left\{(u_{i}(\,\cdot\,,(\theta_{0},\theta_{i},\theta_{-i})))_{i\in I}\left|(\theta_{0},\theta_{-i})\in\Theta_{0}\times\Theta_{-i}\right.\right\}\left|\,\theta_{i}\in\Theta_{i}\right.\right\}.
\]
\end{enumerate}
\end{definition}

The interpretation of each element $\upsilon=(\upsilon_{0},(\upsilon_{i})_{i\in I})\in\mathcal{C}(\Upsilon)$ is similar to that of each payoff state $\theta$: A description of each player $i$'s preferences and information. First, $\upsilon_{0}$ consists in a profile specifying a utility function $(\upsilon_{0})_{i}:Z\rightarrow\mathds{R}$. Second, $\upsilon_{i}$ can be interpreted as the set of profiles that can be conceived by a player $i$ already endowed with a payoff type as the game unfolds---anything outside of $\upsilon_{i}$ will never be part of $i$'s beliefs (even if necessary for rationalizing observed behavior).

Now, the main reason for introducing Definition \ref{definition:canonical} is that it allows for envisioning payoff structures as informational constraints over a canonical object directly derived from the extensive form, namely, the space of profiles of all conceivable utility functions, $\prod_{i\in I}\mathds{R}^{Z}$: Simply notice that, given $\Upsilon$, $\mathcal{C}_{0}(\Upsilon)$ is a compact subset of $\prod_{i\in I}\mathds{R}^{Z}$ and that for each $i\in I$, each $\upsilon_{i}\in\mathcal{C}_{i}(\Upsilon)$ is a compact subset of $\mathcal{C}_{0}(\Upsilon)$---so that $\mathcal{C}_{i}(\Upsilon)$ is a compact subset of compact subsets of $\mathcal{C}_{0}(\Upsilon)$. 
The notion of convergence for canonical representations of standard payoff is very natural: Say that $(\mathcal{C}(\Upsilon^{n}))_{n\in\mathds{N}}$ converges to $\mathcal{C}(\Upsilon)$ if:
\begin{itemize}

\item[1.] Sequence $(\mathcal{C}_{0}(\Upsilon^{n}))_{n\in\mathds{N}}$ converges to $\mathcal{C}_{0}(\Upsilon)$ in the Hausdorff metric for compact subsets of $\prod_{i\in I}\mathds{R}^{Z}$.

\item[2.] For each player $i$, sequence $(\mathcal{C}_{i}(\Upsilon^{n}))_{n\in\mathds{N}}$ converges to $\mathcal{C}_{i}(\Upsilon)$ in the Hausdorff metric for compact subsets of compact subsets of $\prod_{i\in I}\mathds{R}^{Z}$

\end{itemize}
The use of canonical representations greatly simplifies (at a conceptual level) our analysis of robustness in Section \ref{section:results}; perturbing payoff structures will consist on simply perturbing profiles of sets instead of perturbing both sets and utility functions whose domain includes these perturbed sets:

\begin{definition}[Convergence of payoff structures]
\label{definition:convergence}
Let $\Gamma$ be an extensive form. Then, we consider that a sequence of standard payoff structures $(\Upsilon^{n})_{n\in\mathds{N}}$ \textup{converges} to a standard payoff structure $\Upsilon$ if $(\mathcal{C}(\Upsilon^{n}))_{n\in\mathds{N}}$ converges to $\mathcal{C}(\Upsilon)$.
\end{definition}


\subsection{Initial beliefs}
\label{subsection:typeshierarchies}

At the beginning of dynamic game $\langle\Gamma,\Upsilon\rangle$ each player $i$ is endowed with a payoff type and an initial belief hierarchy about the state of nature and their opponents' payoff types, $\Theta_{0}\times\Theta_{-i}$. The payoff type is assumed to remain invariant as the game progresses, but the initial belief hierarchy may be updated in response to observed behavior. More formally, for each player $i$ we have a \textit{type} $t_{i}=(\theta_{i},\pi_{i})$, where $\theta_{i}\in\Theta_{i}$ is a payoff type and $\pi_{i}=(\pi_{i,1},\pi_{i,2},\dots,\pi_{i,k},\dots)$ is a belief hierarchy \`{a} la \cite{brandenburger-93}, where $\pi_{i,1}$ represents player $i$'s \textit{first-order} belief about $\Theta_{0}\times\Theta_{-i}$, $\pi_{i,2}$ represents player $i's$ \textit{second-order} belief about both $\Theta_{0}\times\Theta_{-i}$ and $i$'s opponents first-order beliefs, and so on. We assume that the belief hierarchy is \textit{coherent}, that is, that higher order beliefs marginalize to lower order beliefs and that coherency is common belief.\footnote{The details of the formalization are standard in the literature but, to be more precise, a belief hierarchy for player $i$ consists in a sequence $(\pi_{i,k})_{k\in\mathds{N}}$ where: $\pi_{i,1}$ is an element of $X_{i,1}:=\Delta(\Theta_{0}\times\Theta_{-i})$ and for each $k\in\mathds{N}$ subsequence $(\pi_{i,1},\dots,\pi_{i,k},\pi_{i,k+1})$ is an element of $X_{i,k+1}:=\{(\pi_{i,1}',\dots,\pi_{i,k}',\pi_{i,k+1}')\in X_{i,k}\times \Delta(\Theta_{0}\times\Theta_{-i}\times\prod_{j\neq i}X_{j,k})\,|\,\textup{marg}_{\Theta_{0}\times\Theta_{-i}\times\prod_{j\neq i}X_{j,k}}\pi_{i,k+1}'=\pi_{i,k}'\}$.}

Let $\Pi_{i}(\Theta)$ denote the set of all such belief hierarchies of player $i$ and $T_{i}(\Theta)$, the set of all types. We know from \cite{brandenburger-93} that there exists a homeomorphism $\tau_{i}:\Pi_{i}(\Theta)\rightarrow\Delta\left(\Theta_{0}\times T_{-i}(\Theta)\right)$, where $T_{-i}(\Theta)=\prod_{j\neq i}T_{j}(\Theta)$. Throughout the paper, for any player $i$ and type $t_{i}=(\theta_{i},\pi_{i})$ we write $\theta_{i}(t_{i})$ to represent type $t_{i}$'s payoff type and  $\pi_{i}(t_{i})=(\pi_{i,1}(t_{i}),\pi_{i,2}(t_{i}),\dots,\pi_{i,k}(t_{i}),\dots)$ to represent type $t_{i}$'s belief hierarchy.


\section{Heterogeneously perceived incentives}
\label{section:robsutnessmodel}

In these section we introduce the main methodological novelty of the paper. We develop a formalism that allows for players to entertain arbitrary, different perceptions, also at the higher-orders, regarding the payoff structure of the game: Player $i$ might be deluded about which the payoff structure player $j$ considers it to be, or may think that player $j$ is deluded about which the payoff structure player $k$ considers it to be, and so on. We formalize the objects that allow for this heterogeneous subjective approach to the game in Section \ref{subsection:subjectiveness} and, in Section \ref{subsection:solutionconcepts}, we extend the definitions of extensive form rationalizability and backward rationalizability to the present set-up.


\subsection[Subjective payoff structures and models]{Subjective payoff structures and subjective models}
\label{subsection:subjectiveness}


\subsubsection{Definitions}
\label{subsubsection:definitions}

Next, we introduce the main new formal element of the paper. A \textit{subjective payoff structure} is a hierarchical construction that represents how a player perceives the payoff structure of a game, how she perceives her opponents to perceive the payoff structure of the game, and so on. Thus, technically, a subjective payoff structure consists in a list in which the first component is a standard payoff structure (the one the player considers to be the \textit{right} one), the second component is a map that assigns a subjective payoff structure to each opponent (the one the player considers each opponent to consider the \textit{right} one), and so on. One consistency condition is required (see part $(iii)$ in the definition below): If player $i$ considers that $\Theta_{j}$ are the possible payoff types of player $j$, then she must consider that $j$ considers that $\Theta_{j}$ is (at least in terms of canonical representations) the set of her own possible payoff types---i.e., there is common knowledge in the event that players know their own set of payoff types.

\begin{definition}[Subjective payoff structure]
Let $\Gamma$ be an extensive form. A \textup{subjective payoff structure} for player $i$ is a sequence $d_{i}=(d_{i,k})_{k\in\mathds{N}}$ such that $d_{i,1}\in\mathscr{P}_{i}^{1}:=\mathscr{P}$ and, for every $k\geq 2$, we have $(d_{i,1},\dots,d_{i,k})\in\mathscr{P}_{i}^{k}$, where:
\begin{align*}
\mathscr{P}_{i}^{k}&:=\left\{(d_{i,1},\dots,d_{i,k})\left|
\begin{tabular}{r l}
$(i)$&$(d_{i,1},\dots,d_{i,k-1})\in\mathscr{P}_{i}^{k-1}$\\[3ex]
$(ii)$&$d_{i,k}:I\setminus\{i\}\rightarrow\displaystyle\bigcup_{j\neq i}\mathscr{P}_{j}^{k-1},\textup{  with  }d_{i,k}(j)\in\mathscr{P}_{j}^{k}\textup{  for every  }j\neq i$,\\[3ex]
$(iii)$&$\mathcal{C}_{j}(d_{i,1})=\mathcal{C}_{j}(d_{i,2}(j))\textup{  for every  }j\neq i$
\end{tabular}\right.
\right\}.
\end{align*}
Let $\mathscr{P}_{i}^{\infty}$ denote the set of all subjective payoff structures for player $i$.
\end{definition}

\vspace{0.2cm}

Throughout the paper, with some abuse of notation, we will indistinctly use $\Upsilon$ to refer to a standard payoff structure or to the subjective payoff structure that represents common knowledge of $\Upsilon$; accordingly, we will treat $\mathscr{P}$ as a subset of $\mathscr{P}_{i}^{\infty}$.\footnote{That is, for player $i$, $d_{i}^{\Upsilon}$ such that $d_{i,1}^{\Upsilon}=\Theta$, $d_{i,2}^{\Upsilon}(j)=d_{j,1}^{\Upsilon}$, $d_{i,3}^{\Upsilon}(j)=d_{j,2}^{\Upsilon}$, and so on.} For a given subjective payoff structure $d_{i}$ we denote:
\begin{itemize}

\item By $d_{j|i}:=(d_{i,k}(j))_{k\geq 2}$ and $d_{-i|i}:=(d_{j|i})_{j\neq i}$, the subjective payoff structure ascribed by $d_{i}$ to player $j\neq i$, and the list of subjective payoff structures ascribed by $d_{i}$ to every player $j\neq i$, respectively.

\item By $(d_{i,1})_{0}$, the set of states of nature that corresponds to the set of payoff states in standard payoff structure $d_{i,1}$, and by $(d_{i,1})_{j}$, for each $j\in I$, the set of player $j$'s payoff types in that same set.

\end{itemize}
Since subjective payoff structures put restrictions on which payoff states can be considered by the player, obviously, not every type $t_{i}$ is necessarily \textit{consistent} with every subjective payoff structure $d_{i}$; in order for the pair $(d_{i},t_{i})$ be consistent, the payoff type of $t_{i}$ must be in $(d_{i,1})_{i}$, the support of its first-order belief must be included in the projection on $(d_{i,1})_{0}\times (d_{i,1})_{-i}$ of $d_{i,1}$, and so on.\footnote{Formally, we require that $(1)$ $(d_{i,1})_{i}$ contains $\theta_{i}(t_{i})$; $(2)$ that $\textup{Proj}_{(d_{i,1})_{0}\times (d_{i,1})_{-i}}(d_{i,1})$ contains the support of $\pi_{i,1}(t_{i})$; $(3)$ that for every $t_{-i}$ in the support of $\tau_{i}(\pi_{i}(t_{i}))$ and every $j\neq i$, $(1)$ and $(2)$ hold w.r.t.~$d_{j|i}$; and so on.} For each profile of subjective payoff structures $d$, we let $T_{i}(d_{i})$ and $T_{-i}(d_{-i}):=\prod_{j\neq i}T_{j}(d_{j})$ denote, respectively, the set of player $i$'s types that are consistent with $d_{i}$, and the set of opponents types that are consistent with $d_{-i}$. Finally, let $T(d):=T_{i}(d_{i})\times T_{-i}(d_{-i})$. Consistent combinations of subjective payoff structures and types play a central role in our analysis, since they fully specify players' perception of the incentives of the game:

\begin{definition}[Subjective models, Standard models]
\label{definition:models}
Let $\Gamma$ be an extensive form. Then: 

\begin{itemize}

\item[$(i)$] A \textup{subjective model} for player $i$ is a pair $(d_{i},t_{i})$ where $d_{i}$ is a subjective payoff structure for player $i$ and $t_{i}$ is a type consistent with $d_{i}$. Let $\mathscr{M}_{i}^{\infty}$ denote the set of subjective models for player $i$.

\item[$(ii)$] We say that $(d_{i},t_{i})$ is \textup{standard} if $d_{i}$ represents common knowledge of a standard payoff structure. Let $\mathscr{M}_{i}^{1}$ denote the set of standard subjective models for player $i$.

\item[$(iii)$] A \textup{standard model} is a profile of standard subjective models $(d_{i},t_{i})_{i\in I}$ where the standard payoff structure $\Theta$ associated to each $d_{i}$, $d_{i,1}$, is the same for every player $i$. Let $\mathscr{M}^{1}$ denote the set of standard models.

\end{itemize}

\end{definition}

Thus, a subjective model fully specifies the player's perception of the payoff structure of the game; the subjective payoff structure describes the information about preferences the player will restrict her attention to, and the type, her initial probabilistic assessment about the payoff states. If the player perceives the set of states as commonly known, her perception is represented by a standard subjective model and, if every player agree on their perception of the set of payoff states and this fact is commonly known, players' perceptions are represented by a standard model. We present now an example aimed at illustrating these notions.


\subsubsection{Example}
\label{subsubsection:exampledirectories}

Consider the following extensive form:

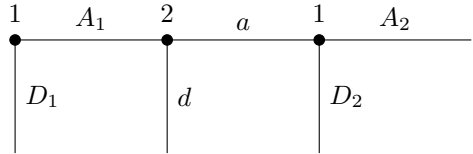
\begin{center}
\begin{tikzpicture}

\coordinate (D11) at (0,0) ;
\node[above] at (0,0.1) {\small $1$};
\draw[fill=black] (D11) circle (0.18em);  

\coordinate (D21) at (2,0) ;
\node[above] at (2,0.1) {\small $2$};
\draw[fill=black] (D21) circle (0.18em);  

\coordinate (D12) at (4,0) ;
\node[above] at (4,0.1) {\small $1$};
\draw[fill=black] (D12) circle (0.18em); 

\coordinate (O1) at (0,-1.5) ;
%
\coordinate (O2) at (2,-1.5) ;
%
\coordinate (O3) at (4,-1.5) ;
%
\coordinate (O4) at (6,0) ;

\draw[black] (D11) to node [above] {\small $A_{1}$} (D21);
\draw[black] (D11) to node [right] {\small $D_{1}$} (O1);

\draw[black] (D21) to node [above] {\small $a$} (D12);
\draw[black] (D21) to node [right] {\small $d$} (O2);

\draw[black] (D12) to node [above] {\small $A_{2}$} (O4);
\draw[black] (D12) to node [right] {\small $D_{2}$} (O3);

\end{tikzpicture}
\end{center}
The mechanics of the game are simple: There are two players ($1$ and $2$) with two available actions each (to \textit{advance} and to go \textit{down}); if, at her turn, a player chooses to go down the game ends, if, instead, she chooses to advance it is then the other player's turn. Actions are observed, player $1$ chooses first and has at most two turns, and player $2$ has at most only one turn. Now, by adding a standard payoff structure $\Theta^{(n,+)}$, we can have a standard dynamic game as the following one, in which for a given $n$, utility functions are commonly known:

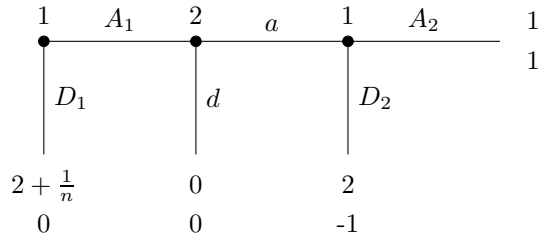
\begin{center}
\begin{tikzpicture}

\coordinate (D11) at (0,0) ;
\node[above] at (0,0.1) {\small $1$};
\draw[fill=black] (D11) circle (0.18em);  

\coordinate (D21) at (2,0) ;
\node[above] at (2,0.1) {\small $2$};
\draw[fill=black] (D21) circle (0.18em);  

\coordinate (D12) at (4,0) ;
\node[above] at (4,0.1) {\small $1$};
\draw[fill=black] (D12) circle (0.18em); 

\coordinate (O1) at (0,-1.5) ;
\node[below] at (0,-1.5) {\small $\begin{tabular}{c}$2+\frac{1}{n}$\\$0$\end{tabular}$};

\coordinate (O2) at (2,-1.5) ;
\node[below] at (2,-1.5) {\small $\begin{tabular}{c}0\\0\end{tabular}$};

\coordinate (O3) at (4,-1.5) ;
\node[below] at (4,-1.5) {\small $\begin{tabular}{c}2\\-1\end{tabular}$};

\coordinate (O4) at (6,0) ;
\node[right] at (6,0) {\small $\begin{tabular}{c}1\\1\end{tabular}$};

\draw[black] (D11) to node [above] {\small $A_{1}$} (D21);
\draw[black] (D11) to node [right] {\small $D_{1}$} (O1);

\draw[black] (D21) to node [above] {\small $a$} (D12);
\draw[black] (D21) to node [right] {\small $d$} (O2);

\draw[black] (D12) to node [above] {\small $A_{2}$} (O4);
\draw[black] (D12) to node [right] {\small $D_{2}$} (O3);

\end{tikzpicture}
\end{center}
Now, if we want to introduce the possibility of discrepancies among players, we can assume that player $2$'s subjective payoff structure is $d_{2}^{n}$, representing common knowledge of $\Theta^{(n,+)}$,\footnote{I.e., such that $d_{2,1}^{n}=\Theta^{(n,+)}$, $d_{2,2}^{n}(1)=\Theta^{(n,+)}$, $d_{2,3}^{n}(1)(2)=\Theta^{(n,+)}$, etc.} and that player $1$'s subjective structure is given by $d_{1}^{n}$, with a first component representing the following payoff structure, $\Theta^{(n,-)}$:
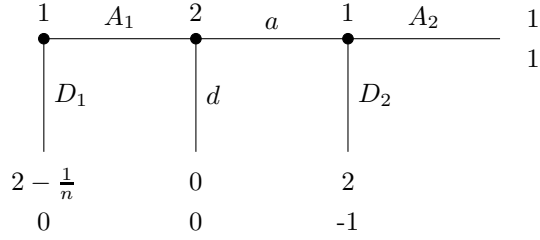
\begin{center}
\begin{tikzpicture}

\coordinate (D11) at (0,0) ;
\node[above] at (0,0.1) {\small $1$};
\draw[fill=black] (D11) circle (0.18em);  

\coordinate (D21) at (2,0) ;
\node[above] at (2,0.1) {\small $2$};
\draw[fill=black] (D21) circle (0.18em);  

\coordinate (D12) at (4,0) ;
\node[above] at (4,0.1) {\small $1$};
\draw[fill=black] (D12) circle (0.18em); 

\coordinate (O1) at (0,-1.5) ;
\node[below] at (0,-1.5) {\small $\begin{tabular}{c}$2-\frac{1}{n}$\\$0$\end{tabular}$};

\coordinate (O2) at (2,-1.5) ;
\node[below] at (2,-1.5) {\small $\begin{tabular}{c}0\\0\end{tabular}$};

\coordinate (O3) at (4,-1.5) ;
\node[below] at (4,-1.5) {\small $\begin{tabular}{c}2\\-1\end{tabular}$};

\coordinate (O4) at (6,0) ;
\node[right] at (6,0) {\small $\begin{tabular}{c}1\\1\end{tabular}$};

\draw[black] (D11) to node [above] {\small $A_{1}$} (D21);
\draw[black] (D11) to node [right] {\small $D_{1}$} (O1);

\draw[black] (D21) to node [above] {\small $a$} (D12);
\draw[black] (D21) to node [right] {\small $d$} (O2);

\draw[black] (D12) to node [above] {\small $A_{2}$} (O4);
\draw[black] (D12) to node [right] {\small $D_{2}$} (O3);

\end{tikzpicture}
\end{center}
and such that $d_{1}^{n}$ considers that player $2$ considers $\Theta^{(n,+)}$ to be commonly known.\footnote{I.e., such that $d_{1,1}^{n}=\Theta^{(n,-)}$ and $d_{2|1}^{n}=d_{2}^{n}$.} Thus, $(d_{1},d_{2})$ represent a situation in which player $2$ considers that $\Theta^{(n,+)}$ is commonly known and player $1$ knows this, but considers the payoff structure to be $\Theta^{(n,-)}$. For each subjective payoff structure, there is a unique type for the player consistent with it: For player $2$, the one that represents common belief in the payoffs given by the unique state in $\Theta^{(n,+)}$, $t_{2}^{n}$, and for player $1$, the one that represents probability 1 belief in the payoffs given by the unique state in $\Theta^{(n,-)}$ and in type $t_{2}^{n}$, $t_{1}^{n}$. Thus, $M_{1}^{n}:=(d_{1}^{n},t_{1}^{n})$ and $M_{2}^{n}:=(d_{2}^{n},t_{2}^{n})$ are the unique possible subjective models; the first one is not standard, the second one is. Thus, $(M_{1}^{n},M_{2}^{n})$ is \textit{not} a standard model.


\subsubsection{Richness and higher-order richness}
\label{subsubsection:higherorderrichness}

Richness assumptions play a central role in the literature of structure theorems and unique selections (see \citealp{weinstein-07}, \citealp{penta-12} or \citealp{chen-12}). The reason is that states in which a strategy is dominant are employed as `seeds' to initiate a contagion argument that propagates through the hierarchy of beliefs: A player might not believe in one of those states, but if she believes that her opponent believes in one of then, then the player should choose a best response to the strategy her opponent considers to be dominant. Formally, richness is defined as follows:

\begin{definition}[Richness, c.f.~\citealp{penta-12}]
Let $\Gamma$ be an extensive form and $\Upsilon$, a standard payoff structure. We say that $\Upsilon$ is \textup{rich} if for every payoff type $\theta_{i}$ and every strategy $s_{i}$ there exists some payoff type $\theta_{i}(\theta_{i},s_{i})\in\Theta_{i}$ such that $s_{i}$ is conditionally dominant for player $i$ at every payoff state $\theta'$ with $\theta_{i}'=\theta_{i}(\theta_{i},s_{i})$ and such that $(\theta_{0},\theta_{i},\theta_{-i})$ and $(\theta_{0},\theta_{i}(\theta_{i},s_{i}),\theta_{-i})$ are payoff equivalent for every $j\neq i$ for every pair $(\theta_{0},\theta_{-i})$.
%
%
%
\end{definition}

Notice that in the case of standard payoff structures, the notion is rather demanding. It requires that it is common knowledge among players that no strategy is known not to be conditionally dominant. Or, in other words, that it is commonly known among players that every strategy can be found to be conditionally dominant for some appropriate belief about the payoff states. On the contrary, subjective payoff structures allow for greatly relaxing this common knowledge assumption while keeping its bite. Players may know that richness does not hold, but they may consider that their opponents consider it to hold; or they may consider that it does not hold, that their opponents consider it not to hold, but consider that their opponents consider their opponents to consider that it holds. Thus, it is possible to approximate common knowledge of absence of richness while still having the chance to exploit an infection argument: 
I could perceive the game as not being rich, but perceive that my opponents perceive it to be rich, or only perceive that they perceive that their opponents perceive that it is rich, and so on. In order to formalize the idea that richness holds at high orders but not low ones let us introduce the following definition:

\begin{definition}[Higher-order richness]
Let $\Gamma$ be an extensive form and $d_{i}$, a subjective payoff structure for player $i$. Then we say that $d_{i}$ satisfies:
\begin{itemize}

\item[$(i)$] $1^{\textup{st}}$-\textup{order richness} if $d_{i,1}$ is rich.

\item[$(ii)$]$k^{\textup{th}}$-\textup{order richness}, for $k\geq 2$, if $d_{j|i}$ satisfies $(k-1)^{\textup{th}}$-richness for every $j\neq i$.

\item[$(iii)$] \textup{Higher-order richness} if $d_{i}$ satisfies $k^{\textit{th}}$-order richness for some $k\geq 1$.

\end{itemize}
\end{definition}

Notably, it turns out that higher-order richness is a generic property within the space of subjective payoff structures; thus, while richness is a very demanding notion for games with a commonly known state space, higher-order richness is extremely mild for games with subjective payoff structures:

\begin{restatable}[Genericity of higher-order richness]{proposition}{genericrichness}
\label{proposition:genericrichness}
Let $\Gamma$ be an extensive form. Then, the set of subjective payoff structures that satisfy higher-order richness is open and dense in $\mathscr{P}_{i}^{\infty}$.
\end{restatable}


\subsubsection{Discussion: Relaxation of all common knowledge assumptions}
\label{subsubsection:discussionCKrelax}

We would like to be emphatic about how natural this separation of rules and stakes is: While the possible outcomes of an economic interaction can be objectively described, there is nothing intrinsically objective in the notion of \textit{utility}; the allocation of resources that results from each combination of reported preferences can be commonly known, the market shares that correspond to each different combination of pricing strategies can be commonly known, etc. But each individual agent's preference over material outcomes is a personal trait of hers. In consequence, it seems reasonable to remove the set of states, or possible utilities, from the objective (in the sense commonly known) part of the game


Subjective payoff structures allow for relaxing all common knowledge assumptions about payoffs in a simple, tractable way; the only common knowledge assumptions they entail is the trivial one: The set of utility functions players reason about consists in functions that map outcomes to utilities. Furthermore, not only do subjective payoff structures allow for approximating models not satisfying any richness assumption with models that satisfy richness at high orders, but they also allow for approximating models satisfying certain knowledge assumptions (say, the no information case, as in \citealp{weinstein-07}, or \citealp{chen-12}) with models that satisfy different knowledge assumptions at high orders (say, models with persistent private information, as in \citealp{penta-12}). In this sense, subjective payoff structures allow for enhancing the subjective nature of individual strategic decisions and, following the `Wilson doctrine', develop game-theoretic analyses that dispense with unrealistic, excessively demanding common knowledge assumptions.

To the best of our knowledge, the methods for reducing common knowledge assumptions closest to our approach build on \citeauthor{battigalli-03}'s \citeyearpar{battigalli-03} $\Delta$-\textit{restrictions}. Formally, a $\Delta$-restrictions is a specification of profiles of first-order beliefs, $\Delta:=(\Delta_{i})_{i\in I}$, so that for each player $i$ we have a collection $\Delta_{i}:=(\Delta_{\theta_{i}})_{\theta_{i}\in\Theta_{i}}$ where $\Delta_{\theta_{i}}\subseteq\Delta^{H_{i}\cup\{h^{0}\}}(S_{-i}\times\Theta_{0}\times\Theta_{-i})$ for every payoff type $\theta_{i}$. The difference between $\Delta$-restrictions and profiles of subjective payoff structures is mainly twofold. On the one hand, $\Delta$-restrictions allow for placing restrictions on the evolution of strategic uncertainty, not only payoff uncertainty. On the other, $\Delta$-restrictions impose constraints over higher-order beliefs by assuming that $\Delta$ is commonly known among players. Thus, it is possible to envision subjective payoff structures as $\Delta$-restrictions that $(1)$ do not place constrains on beliefs about opponents' behavior, and $(2)$ include specifications of constraints on higher-order beliefs that do not necessarily entail nontrivial common knowledge assumptions.

The restrictions embodied by subjective payoff structures also provide a reduced from method of representing unawareness in games. The uncertainty about the payoff state could be thought of as arising from the resolution of some unmodeled contingencies regarding physical aspects of the world external to the choices made by the players---a state of nature $\theta_0 \in \prod_{i\in I}\mathds{R}^{Z}$ being shorthand for the particular resolution of this uncertainty that induces such payoffs to each player. Thus, although a player may easily envision the entirety of the abstract space $\prod_{i\in I}\mathds{R}^{Z}$, in actuality he considers on those states relating to some contingency he is aware of.

Now, when a player is aware of a contingency but initially deems it probability zero, he may revise his beliefs by reasoning about his opponents behavior. This is not so when the player is unaware of the contingency, since observing his opponents actions will not (in general) enlighten him to new contingencies. Thus, the essential attribute of our formalism---persistent, heterogeneous belief about the set of payoff states---can arise as a natural consequence of asymmetric awareness.

Of course, an agent's perception of other agents' awareness must be expressible in terms of things he is himself aware of. This is the motivation of various extant restrictions in the awareness literature---closure under subformula in \cite{fagin1988belief}, the \emph{confinidness} property in \cite{heifetz-06}, property F1 in \cite{piermont2019algebraic}. The translation to our set up would be: if $i$ is aware that $j$ is aware of $\theta$ then $i$ must be himself aware of $\theta$. This conceptual restriction can be made formal via the edict: $(d_{i,2}(j))_0 \subseteq (d_{i,1})_0$. The set of payoff states that $i$ considers $j$ to consider is a subset of those he himself is aware of.



\subsection{Solution concepts}
\label{subsection:solutionconcepts}


\subsubsection{Conjectures and sequential rationality}
\label{subsubsection:solutionpreliminaries}


The evolution throughout the game of each player's subjective beliefs is presented by conjectures. Formally, for each player $i$ endowed with a subjective model $(d_{i},t_{i})$, a \textit{conjecture} consists in a list of beliefs $\mu_{i}=(\mu_{i}(h))_{h\in H_{i}\cup\{h^{0}\}}$ satisfying that:
\begin{itemize}

\item[$(i)$] For every history $h$, either initial or in which player $i$ is active, $\mu_{i}(h)$ is a probability measure on $S_{-i} \times (d_{i,1})_{0}\times T_{-i}(d_{-i|i})$ that assigns probability $1$ to $S_{-i}\times (d_{i,1})_{0}\times T_{-i}(d_{-i|i})$.

\item[$(ii)$] Whenever possible, beliefs are updated following conditional probability.

\item[$(iii)$] The initial beliefs represented by $\mu_{i}$ are consisted with those represented by $t_{i}$:
\[
\textup{marg}_{(d_{i,1})_{0}\times T_{-i}(d_{-i|i})}\mu_{i}(h^{0})=\tau_{i}(\pi_{i}(t_{i})).
\]

\end{itemize}
We denote the set of conjectures consistent with subjective model $(d_{i},t_{i})$ by $\textup{C}(d_{i},t_{i})$.

Given a payoff type $\theta_{i}$ and a history $h\in H_{i}\cup\{h^{0}\}$, each conjecture $\mu_{i}$ naturally induces a \textit{conditional expected payoff} for each strategy $s_{i}$:
\[
U_{i}\left(\mu_{i},s_{i}\left|\theta_{i},h\right.\right):=\int_{S_{-i}\times\Theta_{0}\times\Theta_{-i}}u_{i}\left(z\left(s_{-i};s_{i}\left|h\right.\right),(\theta_{0},\theta_{i},\theta_{-i})\right)\textup{d}(\textup{marg}_{S_{-i}\times\Theta_{0}\times\Theta_{-i}}\mu_{i}(h)).
\]
Finally, given a payoff type $\theta_{i}$ and a conjecture $\mu_{i}$, the notion of sequential rationality is then captured via the set of (sequential) \textit{best responses}:\footnote{Notice that, if we denote the set of all conjecture of player $i$ as $\Delta^{H_{i}\cup\{h^{0}\}}(S_{-i}\times(d_{i,1})_{0}\times T_{-i}(d_{-i|i}))$, then $r_{i}:\Theta_{i}\times\Delta^{H_{i}\cup\{h^{0}\}}(S_{-i}\times (d_{i,1})_{0}\times T_{-i}(d_{-i|i}))\rightrightarrows S_{i}$ is upper hemicontinuous and nonempty-valued.}
\[
r_{i}\left(\theta_{i},\mu_{i}\right):=\left\{s_{i}\in S_{i}\left|s_{i}\in\bigcap_{h\in H_{i}(s_{i})}\textup{arg}\underset{s'_{i}\in S_{i}}{\textup{max}}\,U_{i}\left(\mu_{i},s'_{i}\left|\theta_{i},h\right.\right)\right.\right\}.
\]


\subsubsection{Extensive form rationalizability}
\label{subsubsection:solutionconcept}

Extensive form rationalizability formalizes the notion of \textit{forward induction}; that is, the idea that while forming conjectures about opponents' behavior, players take observed past choices into account. Specifically, as clarified by the epistemic analysis by \cite{battigalli-02}, players \textit{rationalize} observed behavior to the `highest possible degree' of strategic sophistication the behavior is consistent with: Rationality; or rationality and belief whenever possible in opponents' rationality; or rationality, belief whenever possible in opponents' rationality and belief whenever possible in opponents' belief whenever possible in their opponents' rationality; etc. Thus, the main conceptual feature of the solution concept is that observed behavior is employed to infer other players' private information and that conjectures about behavior are based on this inference. To be more specific, if a player entertains certain beliefs about the payoff type of an opponent, upon observing a choice that would have been irrational for that payoff type, if possible, she should update her beliefs about the payoff type and believe that the true type is one for which the observed behavior is rational.

Formally, the original notion due to \cite{pearce-84} and \cite{battigalli-97} is generalized, in our context, as a correspondence that maps each player $i$'s subjective model $(d_{i},t_{i})$ to a certain subset of strategies $\textup{F}_{i}(d_{i},t_{i})\subseteq S_{i}$:

\begin{definition}[Extensive form rationalizability, c.f.~\citealp{pearce-84}, \citealp{battigalli-97}]
\label{definition:solutionconcept}
Let $\Gamma$ be an extensive form. Player $i$'s set of \textup{(interim)} \textup{extensive form rationalizable} strategies for subjective model $(d_{i},t_{i})$ is $\textup{F}_{i}(d_{i},t_{i}):=\bigcap_{k\geq 0}\textup{F}_{i,k}(d_{i},t_{i})$ where $\textup{F}_{i,0}(d_{i},t_{i})=S_{i}$ and $\textup{C}_{i,0}^{F}(d_{i},t_{i})=\textup{C}_{i}(d_{i},t_{i})$, and for each $k\geq 0$,
\begin{align*}
\textup{F}_{i,k+1}(d_{i},t_{i})&:=\left\{
s_{i}\in S_{i}\left|\exists\mu_{i}\in\textup{C}_{i,k}^{F}(d_{i},t_{i})\textup{  s.t.  }s_{i}\in r_{i}(\theta_{i}(t_{i}),\mu_{i})
\right.
\right\},\\[3ex]
\textup{C}_{i,k+1}^{F}(d_{i},t_{i})&:=\left\{
\mu_{i}\in \textup{C}_{i,k}^{F}(d_{i},t_{i})\left|
\begin{tabular}{l}
$\textup{$\forall h\in H_{i}\cup\{h^{0}\}$ s.t.}$\\[2ex]
$S_{-i}(h)\times T_{-i}(d_{-i|i})\cap\textup{Graph}\left(\textup{F}_{-i,k}(d_{-i|i},\,\cdot\,)\right)\neq\emptyset$,\\[1.5ex]
$\mu_{i}(h)[(d_{i,1})_{0}\times M]=1\textup{  for some }M\subseteq\textup{Graph}\left(F_{-i,k}(d_{-i|i},\,\cdot\,)\right)$
\end{tabular}
\right.
\right\}.
\end{align*}
\end{definition}


\subsubsection{Backward rationalizability}
\label{subsubsection:backwardrationalizability}

\textit{Backward rationalizability} (see \citealp{penta-11,penta-15}) extends the standard notion of backward induction to games with incomplete information. The reasoning process behind differs from that behind extensive form rationalizability seen above: Here, upon observing unexpected behavior, players do not necessarily try to find a rationale for said behavior (even at the cost of dropping their belief in `common belief in rationality'); instead, they can interpret observed behavior as a mistake and maintain their belief that `common belief in rationality' holds.\footnote{It is not conceptually accurate to talk about common belief in rationality in a dynamic game, not at every history this belief can be held; however, we keep this terminology in the present discussion for expositional simplicity.} Thus, the difference between the extensive form rationalizability and backward rationalizability is one of epistemic priority (see \citealp{catonini-19}): While the former dictates to update higher-order beliefs about payoff states in order to perform a rationalization to the highest possible degree, the latter permits to maintain the higher-order beliefs about payoff states, believe that unexpected behavior is uninformative and keep the faith on continuation behavior being consistent with rationality and common belief thereof.

Under subjective payoff structures, \citeauthor{penta-11}'s \citeyearpar{penta-11,penta-15} original notion is generalized as a correspondence that maps each player $i$'s subjective model $(d_{i},t_{i})$ to a certain subset of strategies $\textup{B}_{i}(d_{i},t_{i})\subseteq S_{i}$. It is easy to check that the following Definition \ref{definition:BR} below captures the idea, formalized in epistemic language by \cite{perea-14} and \cite{battigalli-18},\footnote{Unlike \cite{penta-11,penta-15}, and in order to make the comparison with extensive form rationalizability more immediate, we drop the requirement that a player has beliefs about her own strategy. Nonetheless, it is immediate that the version we employ captures the idea of `continuation behavior consistent with common belief in rationality' as formalized by \textit{rationality and common belief in future rationality} in \cite{perea-14} or \textit{common full belief in optimal planning and in belief in continuation consistency} in \cite{battigalli-18}.} that players believe that rationality will hold at every possible higher-order in the continuation game, even after moves inconsistent with it. First, define the set of profiles of player $i$'s opponents' strategies that are equivalent to $s_{-i}$ in the continuation game following history $h$ as:
\[
[s_{-i}]_{h}:=\left\{s_{-i}'\in S_{-i}\left|\forall j\neq i,\forall h\in H_{j}\textup{  s.t.  }h\prec h', s'_{j}(h)=s_{j}(h)\right.\right\}.
\]
Then, backward rationalizability is defined in our context as:

\begin{definition}[Backward rationalizability, c.f.~\citealp{penta-11}]
\label{definition:BR}
Let $\Gamma$ be an extensive form. Player $i$'s set of \textup{(interim)} \textup{backward} \textup{rationalizable} strategies for subjective model $(d_{i},t_{i})$ is $\textup{B}_{i}(d_{i},t_{i}):=\bigcap_{k\in\mathds{N}}\textup{B}_{i,k}(d_{i},t_{i})$ where $\textup{B}_{i,0}(d_{i},t_{i}):=S_{i}$ and $\textup{C}_{i,0}^{B}(d_{i},t_{i}):=\textup{C}_{i}(d_{i},t_{i})$, and for each $k\geq 0$,
\begin{align*}
\textup{B}_{i,k+1}(d_{i},t_{i})&:=\left\{s_{i}\in S_{i}\left|\exists \mu_{i}\in \textup{C}_{i,k}^{B}(d_{i},t_{i})\textup{  s.t.  }s_{i}\in r_{i}(\theta_{i}(t_{i}),\mu_{i})
\right.\right\},\\[3ex]
\textup{C}_{i,k+1}^{B}(d_{i},t_{i})&:=\left\{\mu_{i}\in \textup{C}_{k}^{B}(d_{i},t_{i})\left|
\begin{tabular}{l}
$\textup{$\forall h\in H_{i}\cup\{h^{0}\}$,}$\\[2ex]
$(\theta_{0},s_{-i},t_{-i})\in\textup{supp }\mu_{i}(h)\Rightarrow [s_{-i}]_{h}\cap\textup{B}_{-i,k}(d_{-i|i},t_{-i})\neq\emptyset$
\end{tabular}
\right.\right\}.
\end{align*}
\end{definition}


\section{Results}
\label{section:results}

The application of the theoretical framework developed above will shed light on the impact on strategic behavior of players' discrepancies about the state space. Section \ref{subsection:norefinementresult} shows that extensive form rationalizable outcomes are \textit{not always} a refinement of those by backward induction, and that they are \textit{not} robust to misspecifications of subjective models. Section \ref{subsection:robustness} shows that extensive form rationalizability delivers predictions that are robust to misspecifications of initial beliefs (Proposition \ref{proposition:UHCtypes}) and that backward rationalizability delivers predictions that are robust to misspecification of subjective models (Proposition \ref{proposition:UHCall}); a direct consequence of these results is that, in dynamic games where the set of payoff states is rich, extensive form rationalizability, backward rationalizability and \textit{interim rationalizability} all coincide (Corollaries \ref{corollary:ISR} and \ref{corollary:BR}).\footnote{By interim rationalizability we refer to \citeauthor{dekel-07}'s \citeyearpar{dekel-07} \textit{interim correlated rationalizability} applied to the normal-form representation of the dynamic game. In \cite{penta-12} interim rationalizability is defined using \textit{sequential} rationality instead of \textit{ex ante} rationality, and the solution concept is dubbed \textit{interim sequential rationalizability}. In the games considered by \cite{chen-12}, where players have no private information, interim sequential rationalizability and interim rationalizability coincide (see \cite{penta-12}, p.~648).} Finally, in Section \ref{subsection:uniqueselections} we prove that, by perturbing standard models, it is possible to create unique selection arguments \`{a} la \cite{weinstein-07} so that, for any standard model and any extensive form rationalizable outcome, the outcome is the unique prediction along the misspecified subjective models (Theorem \ref{theorem:uniqueselection}), and we end by presenting an example that illustrates why such a unique selection argument is impossible for backward rationalizability.


\subsection{Rationalization as a refinement criterion}
\label{subsection:norefinementresult}


\subsubsection{Main counterexample}
\label{subsubsection:maincounterexample}

Consider again the extensive form studied in Section \ref{subsubsection:exampledirectories}. Let us endow it with a standard payoff structure $\Theta=\{\theta\}$ such that the resulting dynamic game is: 

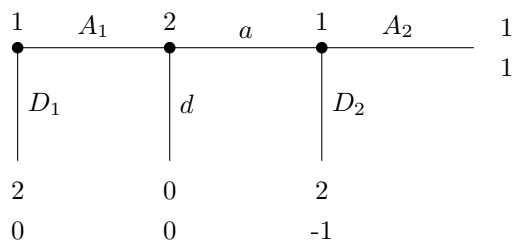
\begin{center}

\begin{tikzpicture}

\coordinate (D11) at (0,0) ;
\node[above] at (0,0.1) {\small $1$};
\draw[fill=black] (D11) circle (0.18em);  

\coordinate (D21) at (2,0) ;
\node[above] at (2,0.1) {\small $2$};
\draw[fill=black] (D21) circle (0.18em);  

\coordinate (D12) at (4,0) ;
\node[above] at (4,0.1) {\small $1$};
\draw[fill=black] (D12) circle (0.18em); 

\coordinate (O1) at (0,-1.5) ;
\node[below] at (0,-1.5) {\small $\begin{tabular}{c}$2$\\0\end{tabular}$};

\coordinate (O2) at (2,-1.5) ;
\node[below] at (2,-1.5) {\small $\begin{tabular}{c}0\\0\end{tabular}$};

\coordinate (O3) at (4,-1.5) ;
\node[below] at (4,-1.5) {\small $\begin{tabular}{c}2\\-1\end{tabular}$};

\coordinate (O4) at (6,0) ;
\node[right] at (6,0) {\small $\begin{tabular}{c}1\\1\end{tabular}$};

\draw[black] (D11) to node [above] {\small $A_{1}$} (D21);
\draw[black] (D11) to node [right] {\small $D_{1}$} (O1);

\draw[black] (D21) to node [above] {\small $a$} (D12);
\draw[black] (D21) to node [right] {\small $d$} (O2);

\draw[black] (D12) to node [above] {\small $A_{2}$} (O4);
\draw[black] (D12) to node [right] {\small $D_{2}$} (O3);

\end{tikzpicture}

\end{center}

Now, for each player $i=1,2$ let $d_{i}^{0}$ represent the standard subjective payoff structure associated to $\Theta$. Clearly, there is a unique type for player $i$ consistent with $\Theta$, namely the one representing common belief in $\theta$, which we denote by $t_{i}^{0}$. Thus, there is a unique profile of subjective models, $(M_{1},M_{2})$ where $M_{i}=(d_{i}^{0},t_{i}^{0})$ for both $i=1,2$, which is, indeed, a standard model. The strategic analysis of this game yields the following conclusions:
\begin{itemize}

\item The unique prediction consistent with extensive form rationalizability is $D_{1}$ (outcomewise). This is easy to see. First, player $1$ will never choose $A_{2}$ at her second turn. 
Second, 
if player $2$ finds herself at $A_{1}$, while she can be surprised that $1$ took the risk of advancing instead of guaranteeing a payoff of $2$, she 
need not abandon the belief that 1 is rational. Player $1$ may have advanced, rationally, by believing that $2$ does not expect her
(i.e., 2 does not expect 1)
to be rational in her second turn, and therefore believing that $2$ will advance as well. Thus, since $2$ maintains the belief that $1$ is rational, she will expect $1$ to choose $D_{2}$ in her 
first 
second turn, and will therefore choose $d$. Being able to forecast this, player $1$ will choose $D_{1}$ in her first turn.

\item The unique prediction consistent with backward rationalizability is $D_{1}$ as well. This is immediate.

\end{itemize}

Thus, both solution concepts uniquely predict $D_{1}$ is played. 
 Suppose now that we perturb the game using a sequence consisting of the subjective models defined in Section \ref{subsubsection:exampledirectories}, $(M_{1}^{n},M_{2}^{n})_{n\in\mathds{N}}$. For each $n\in\mathds{N}$, the situation can be represented as follows:
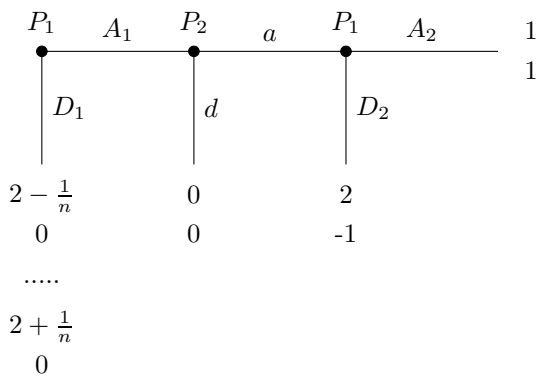
\begin{center}
\begin{tikzpicture}

\coordinate (D11) at (0,0) ;
\node[above] at (0,0.1) {\small $P_{1}$};
\draw[fill=black] (D11) circle (0.18em);  

\coordinate (D21) at (2,0) ;
\node[above] at (2,0.1) {\small $P_{2}$};
\draw[fill=black] (D21) circle (0.18em);  

\coordinate (D12) at (4,0) ;
\node[above] at (4,0.1) {\small $P_{1}$};
\draw[fill=black] (D12) circle (0.18em); 

\coordinate (O1) at (0,-1.5) ;
\node[below] at (0,-1.5) {\small $\begin{tabular}{c}$2-\frac{1}{n}$\\0\\[0.5ex]·····\\[0.5ex]$2+\frac{1}{n}$\\0\end{tabular}$};

\coordinate (O2) at (2,-1.5) ;
\node[below] at (2,-1.5) {\small $\begin{tabular}{c}0\\0\end{tabular}$};

\coordinate (O3) at (4,-1.5) ;
\node[below] at (4,-1.5) {\small $\begin{tabular}{c}2\\-1\end{tabular}$};

\coordinate (O4) at (6,0) ;
\node[right] at (6,0) {\small $\begin{tabular}{c}1\\1\end{tabular}$};

\draw[black] (D11) to node [above] {\small $A_{1}$} (D21);
\draw[black] (D11) to node [right] {\small $D_{1}$} (O1);

\draw[black] (D21) to node [above] {\small $a$} (D12);
\draw[black] (D21) to node [right] {\small $d$} (O2);

\draw[black] (D12) to node [above] {\small $A_{2}$} (O4);
\draw[black] (D12) to node [right] {\small $D_{2}$} (O3);

\end{tikzpicture}

\end{center}
Remember that here players commonly agree on the payoffs that correspond to terminal histories $(A_{1},d)$, $(A_{1},a,D_{2})$ and $(A_{1},d,A_{2})$. However, in player $2$'s mind they also commonly agree about the payoffs corresponding to $D_{1}$ being $(2+1/n,0)$, but this is wrong: In player $1$'s mind, the payoff that corresponds to $D_{1}$ is $(2-1/n,0)$, but she knows that player $2$ 
believes in the commonly agreement on $(2+1/n,0)$. Obviously, the larger $n$, the closer we are to the game with commonly known payoff structure $\Theta$. This is the sense in which we say that $(M_{1}^{n},M_{2}^{n})_{n\in\mathds{N}}$ converges to $(M_{1},M_{2})$. For each given $n$, the corresponding strategic analysis yields the following conclusions:
\begin{itemize}

\item The predictions consistent with extensive form rationalizability are multiple now: $D_{1}$, $(A_{1},d)$ and $(A_{1},a,D_{2})$. This is easy to see. Upon observing $A_{1}$, which in her mind is strictly dominated, player $2$ may drop the belief that player $1$ is rational, and may have any arbitrary belief about 1's future behavior. She might believe that player $1$ will be completely irrational and choose $A_{1}$ in her second turn, or she might believe player $1$ will be more sensible next time and choose $D_{2}$. In consequence, both $a$ and $d$ are rationalizable choices for player $2$. As a result, $(A_{1},D_{2})$, $(D_{1},A_{2})$ and $(D_{1},D_{2})$ are all rationalizable strategies for player $1$, depending on what she expects player $2$ to choose.

\item The unique prediction consistent with backward rationalizability is $D_{1}$ again. Also again, this is immediate.

\end{itemize}

The example shows that extensive form rationalizability refining backward rationalizability crucially depends on players commonly agreeing on the set of possible payoff states. As arbitrarily small disagreements arise, both $(A_{1},d)$ and $(A_{1},a,D_{2})$ become consistent with extensive form rationalizability; however, \textit{none} of them is consistent with backward rationalizability. Importantly, the example delivers an additional lesson: Extensive form rationalizable outcomes fail to be upper hemicontinuous, even in standard models; all $D_{1}$, $(A_{1},d)$ and $(A_{1},a,D_{2})$ are consistent with extensive form rationalizability along the sequence, but only $D_{1}$ is so in the limit.


\subsubsection{Conclusions}
\label{subsubsection:maincounterexampleconclusions}

Section \ref{subsubsection:maincounterexample} illustrates the critical dependence of rationalization processes to apparently negligible modeling details of the state space. As a result, the intuition that forward induction (i.e., extensive form rationalizability) always provides refinements of backward induction (i.e., backward rationalizability) is revealed as both extremely fragile and dependent on agents perfectly agreeing on the possible payoff-relevant nonstrategic contingencies of the environment. That is, arbitrarily small discrepancies among players can lead to the collapse of usual refinement criteria based on the use of observed behavior in order to form conjectures about future behavior.

The reasons for the conceptual discontinuity are evident. There is an option that is rational for player $1$ that, because of an arbitrarily small discrepancy on the perception of the payoff structure, cannot be rationalized by player $2$. Hence, player $1$ can try to \textit{play dumb} by choosing to advance and obtain, depending on the reaction of $2$, a better payoff than the one backward rationalizability allows for. 
In contrast to the standard environment with a commonly known payoff structure---where it is \textit{always impossible} that a choice rendered as irrational by her opponents is potentially beneficial for a player---arbitrarily small disagreements about the payoff structure ca permit an agent to profit from seemingly irrational choices. 
In technical terms, the observation is the result of the interplay between two factors:

\begin{enumerate}

\item The correspondence that maps subjective payoff structures to histories that can be reached by extensive form rationalizable strategies of the opponents is not lower hemicontinuous. This is an obvious fact, if only because it is well-known that the best response correspondence is not lower hemicontinuous. It is also easily visible in the example. In the limit, initial action $A_{1}$ is rational for player 1 and thus, player $2$'s history $A_{1}$ is reachable by strategies in $\textup{F}_{1,1}(M_{1})$. On the contrary, for every $n\in\mathds{N}$, history $A_{1}$ is not reachable by any strategy in $\textup{F}_{1,1}(M_{1}^{n})$. To put it explicitly: A history that can be reached in the limit may not be reached along the perturbation.

\item Extensive form rationalizability puts certain restrictions on beliefs only at some histories: requiring that player $2$ believes that player $1$ is rational at every history that could have been reached if this was the case, but no constraint is placed in the remaining histories. This, together with the lack of lower hemicontinuity mentioned above leads to the collapse of upper hemicontinuity. To see it, let $(\mu_{2}^{n})_{n\in\mathds{N}}$ be a sequence of player $2$'s conjectures such that, for each $n\in\mathds{N}$, $\mu_{2}^{n}$ assigns, at every history $h\in H_{2}\cup\{h^{0}\}$ reachable by some strategy in $\textup{F}_{1,1}(M_{1}^{n})$, probability $1$ to player $1$ choosing a strategy in $\textup{F}_{1,1}(M_{1})$. Clearly, this constraint only applies in $h=h^{0}$. Hence, even if the sequence is convergent, there is no way to guarantee that its limit, $\mu_{2}$, will assign, at history $A_{1}$, probability 1 to player $1$ choosing a strategy in $\textup{F}_{1,1}(M_{1})$: The elements in the sequence do not satisfy any particular requirement at history $A_{1}$. In consequence, it is impossible to ensure that a sequence of conjectures that justifies the inclusion of strategy $s_{2}$ in each $\textup{F}_{2,2}(M_{2}^{n})$ converges (or has a subsequence that does so) to a conjecture that justifies the inclusion of $s_{2}$ in $\textup{F}_{2,2}(M_{2})$.

\end{enumerate}

Besides the above, the example delivers other two additional important takeaways. First, that extensive form rationalizability is not robust to misspecifications of models (in particular, of subjective payoff structures): Arbitrarily small discrepancies on the set of payoff states can lead to predictions that where ignored in the benchmark model ($(A_{1},d)$ and $(A_{1},a,D_{1})$ in the example). Second, that multiplicity of the extensive form rationalizable outcomes can be robust: It is easy to see that, in the example, for every $n\in\mathds{N}$, outcomes $D_{1}$, $(A_{1},d)$ and $(A_{1},a,D_{1})$ are all consistent with extensive form rationalizability for every model $M'$ close enough to $(M_{1}^{n},M_{2}^{n})$. These two observations, further discussed in Sections \ref{subsection:robustness} and \ref{subsection:uniqueselections} respectively, contradict some established facts of static games. On the one hand, since in static games interim beliefs are not updated, all the uncertainty regarding utility functions can be encapsulated by the type, and we know from \cite{dekel-07} that interim rationalizability (the static version of extensive form rationalizability) is robust to misspecifications of types. On the other hand, the seminal analysis of \cite{weinstein-07} reveals that multiplicity of interim rationalizability is nowhere robust, there exists no type with a neighborhood within which interim rationalizability is always multiple.

\subsection{Robustness to misspecifications}
\label{subsection:robustness}


\subsubsection{Extensive form rationalizability}
\label{subsubsection:ERrobustness}

The example above shows that extensive form rationalizability is not robust: Arbitrarily small misspecifications of subjects' perceptions of the payoff structure can lead to outcomes ignored by the benchmark model.\footnote{Specifically, while the original standard model only admits $D_{1}$, predictions $(A_{1},a,D_{2})$, $(A_{1},d)$ and $D_{1}$ all are consistent with the perturbed models.} Notice however, that the models employed in the example, both the benchmark one and the perturbations, are such the corresponding payoff structures only admit one type. Thus, the lack of robustness of extensive form rationalizability follows from its lack of upper hemicontinuity on subjective payoffs structures. In this section we show that extensive form rationalizability does pass some partial robustness test: For a fixed subjective payoff structure, it is upper hemicontinuous on initial types.

\begin{restatable}[Robustness of extensive form rationalizability to misspecifications of initial beliefs]{proposition}{UHCtypes}
\label{proposition:UHCtypes}
Let $\Gamma$ be an extensive form. Then for any player $i$ and any subjective payoff structure $d_{i}$, the correspondence $\textup{F}_{i}(d_{i},\cdot\,):T_{i}(d_{i})\rightrightarrows S_{i}$ is upper hemicontinuous.
\end{restatable}

The conceptual relevance of the result is twofold. First, standard economic modeling assumes that the state space of the game is commonly known and thus, imperturbable. Proposition \ref{proposition:UHCtypes} guarantees that, in these settings, extensive form rationalizable predictions are robust (also at individual level, i.e., in terms of strategies, not only outcomes). Second, as we argue next, the preposition warns against the perils and unfitting consequences of richness assumption in dynamic games. We know from \cite{penta-12} and \citeauthor{chen-12}'s \citeyearpar{chen-12} work that under richness, no refinement of interim rationalizability is robust. Now, Proposition \ref{proposition:UHCtypes} establishes that extensive form rationalizability is upper hemicontinuous in the settings studied by \cite{penta-12} and \cite{chen-12}, and we know that extensive form rationalizability \textit{is} a refinement of interim rationalizability. Hence, the conclusion, formally stated in Corollary \ref{corollary:ISR}, is immediate: Under richness, interim rationalizability (both ex ante and sequential) and extensive form rationalizability \textit{coincide}.\footnote{The argument is so straightforward that we omit any proof} Thus, richness assumptions seem unsuitable for dynamic games; by killing all the bite of rationalization, they render the use of dynamic modeling a mostly redundant tool. Let us end the discussion by formally presenting the equivalence:\footnote{\cite{battigalli-07} provide a similar insight in a weaker version of the following result (Proposition 5 in p.177): For rich $\Theta$ and state of nature $\theta$ extensive form rationalizability and interim rationalizability coincide for those types consistent with initial common belief of $\theta$.}

\begin{corollary}[Triviality of rationalization under richness I]
\label{corollary:ISR}
Let $\langle\Gamma,\Upsilon\rangle$ be a dynamic game. Then, if $\Upsilon$ is rich, extensive form rationalizability and interim rationalizability coincide.
\end{corollary}


\subsubsection{Backward rationalizability}
\label{subsubsection:BRrobustness}

The example in Section \ref{section:robsutnessmodel} shows that the failure of upper hemicontinuity of extensive form rationalizability on subjective payoff structures is due to the interplay of: $(1)$ the set of histories reachable by extensive-from rationalizability not being lower hemicontinuous on subjective payoff structures, and $(2)$ the rationalization requirement intrinsic to extensive form rationalizability depending on the subjective payoff structure.\footnote{See bullet points $1.$ and $2.$ in Section \ref{subsubsection:maincounterexampleconclusions}.} Fact $(1)$ implies that beliefs satisfying $(2)$ along a sequence may not imply that the limit belief satisfies $(2)$, because in the limit, the number of histories at which rationalization should me made can have exploded. Now, the reasoning process that explains backward rationalizability imposes the same constraint at \textit{every} history,\footnote{Namely, that continuation play will follow according to what backward rationalizability of the previous order would have dictated.} and thus, if it holds for beliefs in a sequence, it is always satisfied as well by the limit belief. As the following proposition shows, this intuition backward rationalizability is immune to the mechanism that causes extensive form not to be robust, proves to be correct:

\begin{restatable}[Robustness of backward rationalizability to misspecifications of models]{proposition}{UHCall}
\label{proposition:UHCall}
Let $\Gamma$ be an extensive form. Then for any player $i$, the correspondence $\textup{B}_{i}:\mathscr{M}_{i}^{\infty}\rightrightarrows S_{i}$ is upper hemicontinuous.
\end{restatable}

It is immediate then that the result below obtains as a corollary. Clearly, this reinforces the critique on the assumption of richness in dynamic modeling.

\begin{corollary}[Triviality of rationalization under richness II]
\label{corollary:BR}
Let $\langle\Gamma,\Upsilon\rangle$ be a dynamic game. Then, if $\Upsilon$ is rich, backward rationalizability, extensive form rationalizability and interim rationalizability coincide.
\end{corollary}

Once backward rationalizable predictions are confirmed to be robust, a question arises naturally: Does backward rationalizability characterize the \textit{strongest} predictions that are robust to misspecifications of models? That is, is it true that no (nontrivial) refinement of backward rationalizability is robust? The standard argument for results of this kind in static games is that by \cite{weinstein-07}. In those games every interim rationalizable prediction is uniquely selected by some perturbation of the type. In consequence, no refinement of interim rationalizability is robust.\footnote{Suppose it is and take the prediction ignored by the hypothetical refinement. In the perturbation, this prediction should be the unique prediction made by the hypothetical refinement. If the latter was upper hemicontinuous, it would also include it in the limit.} The example in Section \ref{subsubsection:counterexample} below shows that not every prediction of backward rationalizability admits a perturbation that uniquely selects it. In consequence, our analysis leaves two open questions: $(i)$ does backward rationalizability characterize the strongest predictions that are robust to misspecifications of models (which, if true, would require a proof different from the usual approach)? And $(ii)$ in case the answer to the previous question is negative, which is the solution concept that does characterize the strongest predictions robust to misspecifications of models?

\subsection{Unique selections}
\label{subsection:uniqueselections}


\subsubsection{Theorem for extensive form rationalizability}
\label{subsubsection:structurerationalizability}

The following theorem shows that, for any subjective model and any outcome consistent with extensive form rationalizability, it is possible to create a perturbation of the model that allows for the chosen outcome to be uniquely selected along the perturbed sequence---uniquely selected outcomewise, not in terms of strategies. Notice that the theorem does not require that the state space associated to the standard model is rich: It is possible to perturb the state space using profiles of subjective payoff structures whose associated subjective models satisfy higher-order richness, at increasingly higher order the more the state space is approximated. This allows for the usual infection argument to be easily constructed. In consequence, the theorem establishes that dynamic games not satisfying richness also admit unique selections:

\begin{restatable}[Structure theorem for extensive form rationalizability]{theorem}{uniqueselection}
\label{theorem:uniqueselection}
Let $\langle\Gamma,\Upsilon\rangle$ be a dynamic game. Then, for any profile of finite types $t\in T(\Theta)$ and any strategy profile $s\in\textup{F}(\Upsilon,t)$ there exists a sequence of profiles of models $(d^{n},t^{n})_{n\in\mathds{N}}$ such that:
\begin{itemize}

\item[$(i)$] $(d_{i},t_{i})_{n\in\mathds{N}}$ converges to $(d_{i},t_{i})$.

\item[$(ii)$] For every $n\in\mathds{N}$ and every strategy profile $s^{n}\in\textup{F}(d^{n},t^{n})$, $z(s^{n}|h^{0})=z(s|h^{0})$.

\end{itemize}

\end{restatable}

The interest of the theorem lies on illustrating that extensive form rationalizability admits unique selections in settings in which rationalization adds bite, not only on those in which extensive-from rationalizability is necessarily equivalent to interim rationalizability (because richness of the state space is assumed). On top of that, while the perturbations required do entail higher-order richness assumptions, we know from Proposition \ref{proposition:genericrichness} that this requirement is generically satisfied in the set of all subjective payoff structures. Thus, not only the consequence, but also the very demanding nature of standard richness assumptions is alleviated.

However, the conceptual implications beyond the previous observations are not clear. On the one hand, (perfect Bayesian) equilibrium outcomes do not, in general, refine those of extensive form rationalizability, so Theorem \ref{theorem:uniqueselection} does not allow for concluding that equilibrium predictions are not robust---as \citeauthor{weinstein-07}'s \citeyearpar{weinstein-07} structure theorem allows for in the static case. On the other hand, $(i)$ only standard models are perturbed in the theorem (not arbitrary profiles of subjective models), and $(ii)$ extensive form rationalizability is not robust to perturbations of subjective models. In consequence it cannot be concluded either that uniqueness of extensive form rationalizable predictions is a generic phenomenon in the space of profiles of subjective models. Indeed, the contrary is illustrated in the example in Section \ref{subsection:norefinementresult}: It is possible that a standard model is perturbed so that, along the sequence, multiplicity of extensive form rationalizable outcomes is robust. Our last example below shows that unique selection arguments do not exist, at least as a rule, for backward rationalizability.


\subsubsection{Counterexample for backward rationalizability}
\label{subsubsection:counterexample}

The following example illustrates that, while backward rationalizability is robust to misspecifications of models---see Proposition \ref{proposition:UHCall} above, the questions of whether it is the \textit{strongest} robust solution concept remains open: It is not a conclusion that follows from a standard unique selection argument. As exemplified below, there exist situations in which it is not possible to perturb standard models so that backward rationalizable outcomes are uniquely selected. In consequence, even if true, there is no obvious way to claim that no refinement of backward rationalizability delivers robust predictions.

Consider the following dynamic game with standard payoff structure consisting of states $\theta_{1}$ and $\theta_{2}$:

\begin{center}
\begin{tikzpicture}

\coordinate (D11) at (0,0) ;
\node[above] at (0,0.1) {\small $1$};
\draw[fill=black] (D11) circle (0.18em);  

\coordinate (D21) at (2,0) ;
\node[above] at (2,0.1) {\small $2$};
\draw[fill=black] (D21) circle (0.18em);  

\coordinate (D12) at (4,0) ;
\node[above] at (4,0.1) {\small $1$};
\draw[fill=black] (D12) circle (0.18em); 

\coordinate (O1) at (0,-1.5) ;
\node[below] at (0,-1.5) {\small $\begin{tabular}{c}$4$\\$4$\end{tabular}$};

\coordinate (O2) at (2,-1.5) ;
\node[below] at (2,-1.5) {\small $\begin{tabular}{c}3\\3\end{tabular}$};

\coordinate (O3) at (4,-1.5) ;
\node[below] at (4,-1.5) {\small $\begin{tabular}{c}0\\0\end{tabular}$};

\coordinate (O4) at (6,0) ;
\node[right] at (6,0) {\small $\begin{tabular}{c}2\\2\end{tabular}$};

\draw[black] (D11) to node [above] {\small $A_{1}$} (D21);
\draw[black] (D11) to node [right] {\small $D_{1}$} (O1);

\draw[black] (D21) to node [above] {\small $a$} (D12);
\draw[black] (D21) to node [right] {\small $d$} (O2);

\draw[black] (D12) to node [above] {\small $A_{2}$} (O4);
\draw[black] (D12) to node [right] {\small $D_{2}$} (O3);

\node at (3,-3) {\small State $\theta_{1}$};

\end{tikzpicture}
\begin{tikzpicture}

\coordinate (D11) at (0,0) ;
\node[above] at (0,0.1) {\small $1$};
\draw[fill=black] (D11) circle (0.18em);  

\coordinate (D21) at (2,0) ;
\node[above] at (2,0.1) {\small $2$};
\draw[fill=black] (D21) circle (0.18em);  

\coordinate (D12) at (4,0) ;
\node[above] at (4,0.1) {\small $1$};
\draw[fill=black] (D12) circle (0.18em); 

\coordinate (O1) at (0,-1.5) ;
\node[below] at (0,-1.5) {\small $\begin{tabular}{c}$0$\\$0$\end{tabular}$};

\coordinate (O2) at (2,-1.5) ;
\node[below] at (2,-1.5) {\small $\begin{tabular}{c}3\\0\end{tabular}$};

\coordinate (O3) at (4,-1.5) ;
\node[below] at (4,-1.5) {\small $\begin{tabular}{c}1\\1\end{tabular}$};

\coordinate (O4) at (6,0) ;
\node[right] at (6,0) {\small $\begin{tabular}{c}2\\2\end{tabular}$};

\draw[black] (D11) to node [above] {\small $A_{1}$} (D21);
\draw[black] (D11) to node [right] {\small $D_{1}$} (O1);

\draw[black] (D21) to node [above] {\small $a$} (D12);
\draw[black] (D21) to node [right] {\small $d$} (O2);

\draw[black] (D12) to node [above] {\small $A_{2}$} (O4);
\draw[black] (D12) to node [right] {\small $D_{2}$} (O3);

\node at (3,-3) {\small State $\theta_{2}$};

\end{tikzpicture}
\end{center}
The description of the payoff structure is completed by assuming that player $1$ always knows the payoff state, that the second player never knows it, and that these two are common knowledge. The set of backward rationalizable strategies of player $2$ for arbitrary type $t_{2}$ is:
\begin{itemize}

\item[(A)] $\{a\}$, if $t_{2}$ assigns positive probability to $\theta_{2}$. At state $\theta_{2}$ strategy $(A_{1},A_{2})$ is strictly dominant for player 1 (who, remember, knows that the state is $\theta_{1}$). Thus, if player $2$ observes that player $1$ advances in her first round, she updates her beliefs by excluding the possibility of $\theta_{1}$ being the true state. In consequence, $a$ becomes her only rational choice.

\item[(B)] $\{a,d\}$, if $t_{2}$ assigns null probability to $\theta_{2}$. In this case it is unexpected for player $2$ to observe player 1 advance, $D_{1}$ would have been strictly dominant if the state had been $\theta_{1}$. Thus, player $2$ needs to perform an update of her beliefs from scratch, what allows for the following two possibilities: Was choosing $A_{1}$ a mistake and, effectively, the state is $\theta_{1}$? Or where player $2$'s beliefs wrong and the state is $\theta_{2}$, what justifies player $1$'s decision to advance? Both updating alternatives (and the corresponding mixed beliefs) are admissible according to backward rationalizable reasoning, and the corresponding only rational choices are $d$ and $a$, respectively (and one or either of them, if the updated is a mixed belief).

\end{itemize}
Pick now standard model $M:=((\Theta,t_{1}),(\Theta,t_{2}))$, where $(i)$ $\Theta:=\{\theta_{1},\theta_{2}\}$ and $(ii)$ $t_{1}$ and $t_{2}$ represent common belief in states $\theta_{2}$ and $\theta_{1}$, respectively. We claim now that no perturbation of $M$ can lead to $d$ being uniquely selected for player $2$ and thus, to outcome $(A_{1},d)$ (which is backward rationalizable for $M$) being uniquely selected. To see it notice first that if a standard subjective model is close enough to $\Theta$ then it is the union of two disjoint sets of states, $\Theta^{1}$ and $\Theta^{2}$, such that it is strictly dominant for player $1$ to choose $D_{1}$ at every $\theta_{1}'\in\Theta^{1}$ and to choose $(A_{1},A_{2})$ at every $\theta_{2}'\in\Theta^{2}$. Accordingly, for any perturbation of $(\Theta,t_{2})$, the same argument as in (A) and (B) leads to the following conclusions: (A') if in the perturbed model $\Theta^{2}$ gets positive probability then $a$ is player 2's unique backward rationalizable strategy, and (B') if in the perturbed model $\Theta^{2}$ gets zero probability then both $a$ and $d$ are backward rationalizable for player $2$. As a result, there is no perturbation of $M$ in which $a$ is \textit{not} backward rationalizable for player $2$.

\section{Related literature}
\label{section:conclusions}



Our approach bears some similarity with the notion of \textit{unawareness} (see \citealp{fagin1988belief}; \citealp{modica1994awareness}), and in particular with the state-space or semantic approach to modeling awareness (see \citealp{dekel-98}; \citealp{heifetz-06, heifetz-08}; \citealp{piermont2019algebraic}), in which agent's are endowed with a coarse understanding of the true space of uncertainty. In particular, these approaches often model uncertainty via a partial order of increasingly expressive state-spaces: agent's residing in more expressive state-spaces can only reason about events in lower state-spaces. 

Agents who are \emph{introspectively} unaware, however, might reason that there exist contingencies they are unaware of, without, of course, knowing exactly what such contingencies entail (see \citealp{halpern-09}; \citealp{halpern2019partial}). Subjective payoff structures can capture both naive and introspective unawareness by changing restrictions on the relation between $d_i$ and $d_{j|i}$: when $d_i$ is required to contain $d_{j|i}$ then the agent $i$ is naively unaware, as he does not consider it possible that $j$ considers a contingency he does not.

Note, most game-theoretic analyses of unawareness---\cite{heifetz-14}, \cite{perea-18c} and \cite{guarino-20}, for instance---have focused on unawareness w.r.t.~actions or strategies, not payoff states. 

Lack of common knowledge of the type structure has also been studied by \cite{ziegler-19}, and alternative misspecifications of models, by \cite{esponda-16} and, in the context of macroeconomic models, \cite{hansen-01} and \cite{cho-17}.

The study of disagreements about payoff structures
follows the literature on how small changes in beliefs and information at high orders affects strategic behavior. 
\citeauthor{rubinstein-89}'s \citeyearpar{rubinstein-89} \textit{email game} documents that behavior under common knowledge or under \textit{almost} common knowledge can vary drastically. Later, \cite{weinstein-07} show that these discontinuities of behavior are not an isolated phenomenon, but rather, a pervasive feature of games with incomplete information. \cite{penta-12} and \cite{chen-12} extend this observation to dynamic games. \cite{ely-11} and \cite{ruiz-g-18} characterize the types in which these discontinuities arise in static and dynamic settings, respectively. \cite{penta-21} study the strategic impact of the discontinuities corresponding to higher-order uncertainty about the observability of choices, not preferences. Within the literature of robust mechanism design, \cite{oury-12} and \cite{chen-20} study which social choice functions are implementable is a way robust to small misspecifications of higher-order beliefs. To this respect, the new notion of continuity in this paper and our results on the continuity of different solution concepts suggest novel questions and techniques for the study of implementability in dynamic mechanism design.

The counterintuitive nature of behavior being so discontinuous on information has sprung a literature that argues that these discontinuities are an artifact of very specific formalization or unrealistic assumptions on behavior. 
An approach consists in varying the notions of `similarity' of belief hierarchies or `approximation', as studied by \cite{dekel-06}, \citeauthor{chen-10} \citeyearpar{chen-10,chen-17} or \cite{morris-16}. Another approach consists in showing that these discontinuities vanish with the introduction of bounded rationality and, specifically, under arbitrarily small departures from the benchmark of rationality and common belief thereof, as in \cite{strzalecki-14}, \cite{heifetz-18a}, \cite{germano-20}, \cite{murayama-20} or \cite{jimenez-gomez-19}. On the contrary, the literature on global games, starting from \cite{carlsson-93}, has embraced the discontinuities of choices as an intrinsic feature of strategic behavior and leveraged on them to explain diverse economic phenomena such as currency crises (\citealp{morris-98}), bank runs (\citeauthor{angeletos-06}, \citeyear{,angeletos-06, angeletos-07}), conflict (\citealp{baliga-12}), overvaluation in financial markets (\citealp{han-17}) or disclosure policies for stress tests (\citealp{inostroza-18}). Perturbations of state spaces, together with the unique selection argument in Theorem \ref{theorem:uniqueselection}, allow for exploring variations of dynamic global games in which: $(i)$ information about the fundamentals is exogenously obtained solely at the beginning of the interaction, and new one is inferred \textit{only} through observed behavior, and $(ii)$ agents can entertain different perceptions of which payoff-relevant contingencies can eventually take place.


\section*{Acknowledgments and disclaimers}
\label{section:acknowledgements}
\addcontentsline{toc}{section}{Acknowledgments and disclaimers}

A very preliminary version of this project was occasionally presented under the title ``Rationalization and robustness in dynamic games with incomplete information.'' Thanks are due to Pierpaolo Battigalli, Emiliano Catonini, Pierfrancesco Guarino, Jarom\'{i}r Kov\'{a}\v{r}\'{i}k, Avi Lichtig, Antonio Penta and Gabriel Ziegler for insightful comments and valuable feedback, and to audiences at HSE University-International College of Economics and Finance and Universit\`{a} Bocconi. Zuazo-Garin acknowledges financial support from the Spanish Ministry of Economy and Competitiveness, from the Department of Education, Language Policy and Culture of the Basque Government (grants ECO2012-31346 and POS-2016-2-0003 and IT568-13, respectively), from the ERC Programme (ERC Grant 579424) and from the Russian Academic Excellence Project `5-100'. Zuazo-Garin also expresses his gratitude to Northwestern University and Kellogg School of Management (the Department of Economics, MEDS and the CMS-EMS, particularly) for immense hospitality. The authors are responsible for all errors.


\newpage
\appendix
\pagestyle{appendix}

\section[Appendix A. Proofs I: Propositions]{Proofs I: Propositions}
\label{section:proofpropositions}


\subsection{Proposition \ref{proposition:genericrichness}}
\label{subsection:proofgenericrichness}

\genericrichness*
\begin{proof}
For denseness, take any arbitrary standard payoff structure satisfying richness $\Upsilon$. Then, for any subjective payoff structure $d_{i}$ and any $n\in\mathds{N}$ define $d_{i}^{n}(d_{i})$ by setting $d_{i}^{0}(d_{i}):=\Upsilon$ and, for any $n\geq 1$, $d_{i,1}^{n}(d_{i}):=d_{i,1}$ and $d_{-i|i}^{n}(d_{i}):=d_{-i}^{n-1}(d_{-i|i})$. This way, $d_{i}^{n}(d_{i})$ is a subjective payoff structure satisfying that $d_{i,k}^{n}(d_{i})=d_{i,k}$ for every $k\leq n$. Thus, clearly, $(d_{i}^{n}(d_{i}))_{n\in\mathds{N}}$ converges to $d_{i}$. Openness follows from a standard inductive argument: It is immediate that the set of subjective payoff structure satisfying 1$^{\textup{st}}$-order richness is open; given that, it is immediate as well that the set of subjective payoff structure satisfying 2$^{\textup{nd}}$-order richness is open as well, and so on.
\end{proof}


\subsection{Proposition \ref{proposition:UHCtypes}}
\label{subsection:proofUHCtypes}

\UHCtypes*
\begin{proof}
We proceed by induction and verify that the correspondence $\textup{F}_{i,k}(d_{i},\,\cdot\,):T_{i}(d_{i})\rightrightarrows S_{i}$ is upper hemicontinuous for every $k\geq 0$. The claim holds trivially for the initial case ($k=0$) so we can focus on the proof of the inductive step. Fix $k\geq 0$ such that the claim holds; we verify next that it also does for $k+1$. Fix player $i$, subjective payoff structure $d_{i}$ and sequence of types $(t_{i}^{n})_{n\in\mathds{N}}$ converging to some type $t_{i}$, and pick strategy $s_{i}$ such that $s_{i}\in\textup{F}_{i}(d_{i},t_{i}^{n})$ for every $n\in\mathds{N}$. For each $n\in\mathds{N}$, pick conjecture $\mu_{i}^{n}$ that justifies the inclusion of $s_{i}$ in $\textup{F}_{i,k+1}(d_{i},t_{i}^{n})$. Let $(\mu_{i}^{n_{m}})_{m\in\mathds{N}}$ be a convergent subsequence of $(\mu_{i}^n)_{n\in\mathds{N}}$ with limit $\mu_{i}$. Notice then that for any $\ell\leq k$ and any history $h$ reachable by some profile of opponents' strategies in $\textup{F}_{-i,\ell}(d_{-i|i},\,\cdot\,)$ we have that:
\begin{align*}
\mu_{i}(h)\left[(d_{i,1}(i))_{0}\times\textup{Graph}\left(\textup{F}_{-i,\ell}(d_{-i|i},\,\cdot\,)\right)\right]\geq&\\[1ex]
\geq\underset{m\rightarrow\infty}{\textup{limsup}}&\,\mu_{i}(h)^{n_{m}}\left[(d_{i,1})_{0}\times\textup{Graph}\left(\textup{F}_{-i,\ell}(d_{-i|i},\,\cdot\,)\right)\right]=1,
\end{align*}
because we know from the induction hypothesis that $\textup{F}_{-i,\ell}(d_{-i|i},\cdot\,)$ has closed graph. This, together with the continuity of marginalization and the upper hemicontinuity of the best response correspondence allows for concluding that $\mu_{i}$ is a conjecture that justifies the inclusion of $s_{i}$ in $\textup{F}_{i,k+1}(d_{i},t_{i})$.
\end{proof}


\subsection{Proposition \ref{proposition:UHCall}}
\label{subsection:proofUHCall}


\subsubsection{Auxiliary results}

\begin{lemma}
\label{lemma:subjectivemodels}
Let $\Delta^{C}(\Omega)$ be a set of conditional probability systems with (compact and metrizable) space of uncertainty $\Omega$ and family of conditioning events $C$. Consider sequence $(\Omega^{n})_{n\in\mathds{N}}$ of closed subsets of $\Omega$ and and closed $\Omega^{0}\subseteq \Omega$ such that: $(i)$ $(\Omega^{n})_{n\in\mathds{N}}$ converges to $\Omega^{0}$ in the Hausdorff metric, and $(ii)$ for any $n\in\mathds{N}\cup\{0\}$, $\Omega^{n}$ has non-empty intersection with every conditioning event $c\in C$. Then, the following two hold:

\begin{itemize}

\item[$(i)$] For any sequence of conditional probability systems $(\mu^{n})_{n\in\mathds{N}}$ with limit $\mu$ and such that $\mu^{n}\in\Delta^{C}(\Omega^{n})$ for every $n\in\mathds{N}$ it holds that $\mu\in\Delta^{C}(\Omega^{0})$.

\item[$(ii)$] Sequence $(\Delta^{C}(\Omega^{n}))_{n\in\mathds{N}}$ converges to $\Delta^{C}(\Omega^{0})$ in the Hausdorff metric.

\end{itemize}

\end{lemma}
\begin{proof}
For claim $(i)$ we only need to check that $\textup{supp }\mu(c)\subseteq c\cap\Omega^{0}$ for any $c\in C$.\footnote{That updating according conditional probability is respected in the limit is a well-established fact.} Fix arbitrary $c\in C$ and pick $x\in\textup{supp }\mu(c)$. Since the support correspondence is lower hemicontinuous,\footnote{See Theorem 17.14 in \cite{aliprantis-07}, p.~563.} we know that there exists a subsequence $(\mu^{n_{m}})_{m\in\mathds{N}}$ and a sequence $(x^{m})_{m\in\mathds{N}}$ with limit $x$ such that $x^{m}\in\textup{supp }\mu^{n_{m}}(c)\subseteq c\cap \Omega^{n_{m}}$ for any $m\in\mathds{N}$. Then, pick metric $d$ that topologizes $\Omega$, and let $d_{H}$ denote its corresponding Hausdorff metric. We know that $d_{H}(c\cap \Omega^{n_{m}},c\cap \Omega)\rightarrow 0$ and therefore, that $d(c\cap \Omega^{n_{m}},c\cap \Omega^{0})\rightarrow 0$. Since we have that $d(x^{n_{m}},x)\rightarrow 0$ and $c\cap\Omega^{0}$ is closed, we conclude that $x\in c\cap \Omega^{0}$.

For claim $(ii)$, proceed by contradiction and suppose that there exists some $\varepsilon>0$ such that for any $m\in\mathds{N}$ there exists some $n_{m}\geq n$ such that $d_{H}(\Delta^{C}(\Omega^{n_{m}}),\Delta^{C}(\Omega^{0}))>\varepsilon$. Then, for any $m\in\mathds{N}$ there exists some $\mu^{n_{m}}\in\Delta^{C}(\Omega^{n_{m}})$ such that $d(\mu^{n_{m}},\mu)>\varepsilon$ for any $\mu\in\Delta^{C}(\Omega^{0})$. Now, since $\Delta^{C}(\Omega)$ is compact, we know that there exists a convergent subsequence $(\mu^{n_{m_{r}}})_{r\in\mathds{N}}$ with limit $\mu^{0}$. We know from $(i)$ that $\mu^{0}\in\Delta^{C}(\Omega^{0})$, and thus, since $d(\mu^{n_{m_{r}}},\mu^{0})\rightarrow 0$, we reached a contradiction.
\end{proof}


\begin{corollary}
If sequence of subjective payoff structures $(d_{i}^{n})_{n\in\mathds{N}}$ converges to subjective payoff structure $d_{i}$, then, sequence of sets of opponents' types $(T_{-i}(d_{-i|i}^{n}))_{n\in\mathds{N}}$ converges to $T_{-i}(d_{-i|i})$ in the Hausdorff metric.
\end{corollary}
\begin{proof}
We proceed by induction. Let $X_{i}^{1}(d_{i}'):=\textup{Proj}_{(d'_{i,1})_{0}\times (d'_{i,1})_{-i}}(d'_{i,1})$ for any subjective payoff structure $d_{i}'$. Obviously, $X_{i}^{1}(d_{i}^{n})$ converges to $X_{i}^{1}(d_{i})$ in the Hausdorff metric, and thus, we know from Lemma \ref{lemma:subjectivemodels} that $(Z_{i}^{1}(d_{i}^{n}))_{n\in\mathds{N}}$, where $Z_{i}^{1}(d_{i}^{n}):=\Delta(X_{i}^{1}(d_{i}^{n}))$ for any $n\in\mathds{N}$, converges to $Z_{i}^{1}(d_{i}):=\Delta(X_{i}^{1}(d_{i}))$ in the Hausdorff metric. Now, set inductively $Z_{i}^{k-1}(d'_{i})=\Delta(X_{i}^{k-1}(d'_{i}))$ and $X_{i}^{k}(d''_{i}):=X_{i}^{k-1}(d''_{i})\times\prod_{j\neq i}Z_{j}^{k-1}(d'_{j|i})$ for every $k\geq 0$ and every $d'_{i}$ and $d''_{i}$, and suppose that $k\geq 1$ is such that $X_{i}^{k}(d_{i}^{n})$ converges to $X_{i}^{k}(d_{i})$ and $(Z_{j}^{k}(d_{j|i}^{n}))_{n\in\mathds{N}}$ converges to $Z_{j}^{k}(d_{j|i})$ for any $j\neq i$ (all of them in the Hausdorff metric). Then, obviously, $X_{i}^{k+1}(d_{i}^{n})$, defined in the obvious way, converges to $X_{i}^{k+1}(d_{i})$ in the Hausdorff metric, and therefore, we know from Lemma \ref{lemma:subjectivemodels} that $(Z_{i}^{k+1}(d_{i}^{n}))_{n\in\mathds{N}}$ converges to $Z_{i}^{k+1}(d_{i})$. Obviously, it follows that $(\prod_{k\in\mathds{N}}\prod_{j\neq i}Z_{j}^{k}(d_{j|i}^{n}))_{n\in\mathds{N}}$ converges to $\prod_{k\in\mathds{N}}\prod_{j\neq i}Z_{j}^{k}(d_{j|i})$ and hence, that $(T_{-i}(d_{-i|i}^{n}))_{n\in\mathds{N}}$ converges to $T_{-i}(d_{-i|i})$.
\end{proof}

\begin{corollary}
\label{corollary:subjectivemodels}
The set of player $i$'s subjective models is closed; i.e., if pair $(d_{i},t_{i})\in\mathscr{P}_{i}^{\infty}\times T_{i}$ is the limit of a sequence of subjective models $(d_{i}^{n},t_{i}^{n})_{n\in\mathds{N}}$, then $(d_{i},t_{i})$ is a subjective model too.
\end{corollary}
\begin{proof}
Just notice that we know from the previous corollary that $((d_{i,1}^{n})_{0}\times T_{-i}(d_{-i|i}^{n}))_{n\in\mathds{N}}$ converges to $(d_{i,1})_{0}\times T_{-i}(d_{-i|i})$, and thus, we know from Lemma \ref{lemma:subjectivemodels} that, if $\tau_{i}^{n}\in\Delta((d_{i,1}^{n})_{0}\times T_{-i}(d_{-i|i}^{n}))$ for every $n\in\mathds{N}$, and $(\tau_{i}^{n})_{n\in\mathds{N}}$ converges to $\tau_{i}$, then $\tau_{i}\in\Delta((d_{i,1})_{0}\times T_{-i}(d_{-i|i}))$.
\end{proof}


\subsubsection{Proof of the proposition}

\UHCall*
\begin{proof}
We will proceed by induction and verify that the correspondence $\textup{B}_{i,k}:\mathscr{M}_{i}^{\infty}\rightrightarrows S_{i}$ is upper hemicontinuous for every $k\geq 0$. The claim holds trivially for the initial case so we can focus on the proof of the inductive step. Fix $k\geq 0$ such that the claim holds; we verify next that it also does for $k+1$. Fix player $i$, convergent sequence of subjective models $(M_{i}^{n})_{n\in\mathds{N}}$ with limit $M_{i}$ and strategy $s_{i}$ such that $s_{i}\in\textup{B}_{i}(M_{i}^{n})$ for every $n\in\mathds{N}$. For each $n\in\mathds{N}$, pick conjecture $\mu_{i}^{n}$ that justifies the inclusion of $s_{i}$ in $\textup{F}_{i,k+1}(M_{i}^{n})$. Let $(\mu_{i}^{n_{m}})_{m\in\mathds{N}}$ be a convergent subsequence of $(\mu_{i}^n)_{n\in\mathds{N}}$ with limit $\mu_{i}$. For each $n\in\mathds{N}$ let $d_{i}^{n}$ denote the subjective payoff structure associated to subjective model $M_{i}^{n}$ and $d_{i}$, the one associated to $M_{i}$. Notice then that for any $\ell\leq k$ and any history $h$ we have that:
\begin{align*}
X_{i,\ell}^{n}(h):=\left\{(s_{-i},t_{-i})\in S_{-i}\times T_{-i}(d_{-i|i})\left|[s_{-i}]_{h}\cap\textup{B}_{-i,\ell}(d_{-i|i}^{n},t_{-i})\neq h^{0}\right.\right\}&=\\[1ex]
\bigcup_{s_{-i}\in S_{-i}}\left([s_{-i}]_{h}\times T_{-i}(d_{-i|i})\right)\cap\textup{Graph}&\left(\textup{B}_{-i,\ell}(d_{-i|i}^{n},\,\cdot\,)\right),
\end{align*}
and thus, we know by finiteness of $S_{-i}$ and the induction hypothesis that $X_{i,\ell}(h)$ is closed. Then, it holds that:
\begin{align*}
\forall m\in\mathds{N},\,\mu_{i}^{n_{m}}(h)\left[(d_{i,1}^{n_{m}})_{0}\times X_{i,\ell}^{n_{m}}(h)\right]=1&\Rightarrow\forall m\in\mathds{N},\,\mu_{i}^{n_{m}}(h)\left[\bigcup_{r\geq n_{m}}(d_{i,1}^{r})_{0}\times X_{i,\ell}^{r}(h)\right]=1\\[1ex]
&\Rightarrow\forall m\in\mathds{N},\,\mu_{i}(h)\left[\textup{cl}\left(\bigcup_{r\geq n_{m}}(d_{i,1}^{n_{m}})_{0}\times X_{i,\ell}^{r}(h)\right)\right]=1\\[1ex]
&\Rightarrow\mu_{i}(h)\left[\bigcap_{m\in\mathds{N}}\textup{cl}\left(\bigcup_{r\geq n_{m}}(d_{i,1}^{r})_{0}\times X_{i,\ell}^{r}(h)\right)\right]=1\\[1ex]
&\Rightarrow\mu_{i}(h)\left[(d_{i,1})_{0}\times X_{i,\ell}(h)\right]=1,
\end{align*}
where:
\[
X_{i,\ell}(h):=\left\{(s_{-i},t_{-i})\in S_{-i}\times T_{-i}(d_{-i|i})\left|[s_{-i}]_{h}\cap\textup{B}_{-i,\ell}(d_{-i|i},t_{-i})\neq h^{0}\right.\right\}.
\]
Notice that the last inclusion holds by virtue of $X_{i,\ell}^{n}(h)$ being upper hemicontinuous on $n$ because of the induction hypothesis. The above, together with the continuity of marginalization and the upper hemicontinuity of the best response correspondence allows for concluding that $\mu_{i}$ is a conjecture that justifies the inclusion of $s_{i}$ in $\textup{B}_{i,k+1}(M_{i})$.
\end{proof}

\section[Appendix B. Proofs II: Theorem]{Proofs II: Theorem}
\label{section:proofstheorem2}

\subsection{Additional notation}
\label{subsection:Bnotation}


\subsubsection{Distinguished histories}
\label{subsubsection:histories}

The histories of player $i$ that might be reached when every opponent $j\in J\subseteq I\setminus\{i\}$ plays according to some given correspondence $\textup{W}_{j}(d_{j},\,\cdot\,):T_{j}(d_{j})\rightrightarrows S_{j}$ is of special interest and formalized as follows:
\[
\textup{H}_{i}(\textup{W}_{J},d_{J}):=\left\{h\in H_{i}\cup\{h^{0}\}\left|\prod_{j\in J}\left(S_{j}(h)\times T_{j}(d_{j})\cap\textup{Graph}(\textup{W}_{j}(d_{j},\,\cdot\,)\right))\neq h^{0}\right.\right\}.
\]

Another kind of histories that play an important role throughout the proof are those in which the player updates beliefs from scratch. For each player $i$ and conjecture $\mu_{i}$, these are the histories that are considered unlikely to be reached by every history preceding them:
\[
\textup{H}_{i}(\mu_{i}):=\left\{h\in H_{i}\cup\{h^{0}\}\left|\textup{marg}_{S_{-i}}\mu_{i}(h')\left[S_{-i}(h)\right]=0\textup{  for every  }h'\prec h\right.\right\}.
\]
Notice that every conjecture $\mu_{i}$ is fully described by the beliefs $\mu_{i}(h)$ where $h\in\textup{H}_{i}(\mu_{i})$: The belief that corresponds to any other history is obtained via the chain rule.


\subsubsection{Strict rationalizability}
\label{subsubsection:strictrat}

First, for each player $i$ and strategy $s_{i}$ let $[s_{i}]$ denote the set of player $i$'s strategies that are outcome-equivalent to $s_{i}$ --those strategies $s_{i}'$ such that $s_{i}'(h)=s_{i}(h)$ for every history $h\in H_{i}(s_{i})$. Then, player $i$'s set of (\textit{extensive form}) \textit{strictly rationalizable} strategies for subjective model $(d_{i},t_{i})$ is $\textup{F}_{i}^{0}(d_{i},t_{i}):=\bigcap_{k\geq 0}\textup{F}_{i,k}^{0}(d_{i},t_{i})$ where $\textup{F}_{i,0}^{0}(d_{i},t_{i})=S_{i}$ and for each $k\geq 0$,
\begin{align*}
\textup{F}_{i,k+1}^{0}(d_{i},t_{i})&:=\left\{
s_{i}\in S_{i}\left|\exists\mu_{i}\in\textup{C}_{i,k}^{0}(d_{i},t_{i})\textup{  s.t.  }r_{i}(\theta_{i}(t_{i}),\mu_{i})=[s_{i}]
\right.
\right\},\\[3ex]
\textup{C}_{i,k+1}^{0}(d_{i},t_{i})&:=\left\{
\mu_{i}\in \textup{C}_{i,k}^{0}(d_{i},t_{i})\left|
\begin{tabular}{l}
$\textup{$\forall h\in\textup{H}_{i}(\textup{F}_{-i,k}^{0},d_{-i|i})$,}$\\[2ex]
$\mu_{i}(h)[(d_{i,1})_{0}\times M]=1\textup{  for some }M\subseteq\textup{Graph}\left(\textup{F}_{-i,k}^{0}(d_{-i|i},\,\cdot\,)\right)$
\end{tabular}
\right.
\right\}.
\end{align*}

Related to this, we also add the following two notational shorthands, convenient for lemmas \ref{lemma:secondperturbation} and \ref{lemma:thirdperturbation2} in the proof of the theorem:

\begin{itemize}

\item For each subjective payoff structure $d_{i}$ and $k\geq 0$ we denote by $[s_{i}|d_{i}]_{k}$ the set of player $i$'s strategies that are equivalent to strategy $s_{i}$ at every history consistent with the opponents' playing some profile of strategies in $\textup{F}_{-i,k}^{0}(d_{-i|i},\,\cdot\,)$, that is:
\[
[s_{i}|d_{i}]_{k}=\left\{s_{i}'\in S_{i}\left|s_{i}'(h)=s_{i}(h)\textup{  for every  }h\in\textup{H}_{i}(\textup{F}_{-i,k}^{0},d_{-i|i})\right.\right\}.
\]

\item The types consistent with subjective payoff structure $d_{i}$ that admit strategy $s_{i}$ as $\textup{F}_{i,k}^{0}(d_{i},\,\cdot\,)$ via some conjecture that assigns, initially, positive probability to \textit{every} history reachable by some profile of opponents' strategies in $\textup{F}_{-i,k-1}^{0}(d_{-i|i},\,\cdot\,)$ is formalized as:
\[
T_{i,k}(d_{i},s_{i}):=\left\{t_{i}\in T_{i}(d_{i})\left|
\begin{tabular}{r l}
\multicolumn{2}{l}{$\exists \mu_{i}\in\Delta^{H_{i}\cup\{h^{0}\}}(S_{-i}\times (d_{i,1})_{0}\times T_{-i}(d_{-i|i}))$\textup{  such that:}}\\[2ex]
$(i)$&$\mu_{i}\textup{  justifies the inclusion of $s_{i}$ in }\textup{F}_{i,k}^{0}(d_{i},t_{i})$,\\[1.5ex]
$(ii)$&$\forall h\in\textup{H}_{i}(\textup{F}_{-i,k-1}^{0},d_{-i|i}),\,\textup{marg}_{S_{-i}}\mu_{i}(h^{0})[S_{-i}(h)]>0$
\end{tabular}
\right.
\right\}.
\]

\end{itemize}


\subsection{Proof of the theorem}
\label{subsection:Btheorem}

\uniqueselection*
\begin{proof}
Fix standard payoff structure $\Upsilon$, profile of finite types $t\in T(\Theta)$ and strategy profile $s\in\textup{F}(\Upsilon,t)$. Then, we know from lemmas \ref{lemma:firstperturbation2} and \ref{lemma:secondperturbation} below that for each $k\in\mathds{N}$:

\begin{itemize}

\item[(A)] There exist some $m_{k}\in\mathds{N}$ and a sequence of standard models $(\Upsilon^{m},t^{k,m})_{m\geq m_{k}}$ converging to $(\Upsilon,t)$ such that, for any $m\geq m_{k}$ and any player $i$, $t_{i}^{k,m}\in T_{i}(\Upsilon^{m})$ is finite, $s_{i}\in\textup{F}_{i,k}^{0}(\Upsilon^{m},t_{i}^{k,m})$ and for every $j\neq i$,
\[
\textup{H}_{j}(\textup{F}_{i,k}^{0},\Upsilon^{m})=\textup{H}_{j}(\textup{F}_{i,k},\Upsilon^{m})=\textup{H}_{j}(\textup{F}_{i,k},\Upsilon).
\]

\item[(B)] For any $m\in\mathds{N}$ there exists a sequence of profiles of finite types $(t^{k,m,\ell})_{\ell\in\mathds{N}}$ converging to $t^{k,m}$ such that for each player $i$, $t_{i}^{k,m,\ell}\in T_{i,k}(\Upsilon^{m},s_{i})$ (as defined in Section \ref{subsubsection:strictrat}) for every $\ell\in\mathds{N}$.

\end{itemize}
Now, for each player $i$ and each $n\in\mathds{N}$ set $\tilde{t}_{i}^{n}:=t_{i}^{m_{n},m_{n},n}$, and notice that:

\begin{itemize}

\item $\tilde{t}_{i}^{n}$ is finite and included in $T_{i,m_{n}}(\Upsilon^{m_{n}},s_{i})$.

\item $\textup{H}_{i}(\textup{F}_{-i,m_{n}}^{0},\Upsilon^{m_{n}})=\textup{H}_{i}(\textup{F}_{-i,m_{n}},\Upsilon^{m_{n}})=\textup{H}_{i}(\textup{F}_{-i,m_{n}},\Upsilon)$.

\end{itemize}
We know then from Lemma \ref{lemma:thirdperturbation2} below that:

\begin{itemize}

\item[(C)] For any $i$ there exists some subjective model $(d_{i}^{n},t_{i}^{n})$ such that: $d_{i,m_{n}}^{n}=d_{i,m_{n}}'$ for $d_{i}=\Upsilon^{m_{n}}$, $\pi_{i,m_{n}}(t_{i}^{n})=\pi_{i,m_{n}}(\tilde{t}_{i}^{n})$ and as (for $[s_{i}|\Upsilon^{m_{n}}]_{m_{n}-1}$ defined in Section \ref{subsubsection:strictrat}),
\[
\textup{F}_{i,m_{n}+1}(d_{i}^{n},t_{i}^{n})\subseteq [s_{i}|\Upsilon^{m_{n}}]_{m_{n}-1},
\]
and hence,
\[
\textup{F}_{i}(d_{i}^{n},t_{i}^{n})\subseteq [s_{i}|\Upsilon^{m_{n}}]_{m_{n}-1}.
\]
\end{itemize}
It follows from (A) that:
\begin{align*}
[s_{i}|\Upsilon^{m_{n}}]_{m_{n}-1}&=\left\{s_{i}'\in S_{i}\left|s_{i}'(h)=s_{i}(h)\textup{  for any  }h\in\textup{H}_{i}(\textup{F}_{-i,m_{n}-1}^{0},\Upsilon^{m_{n}})\right.\right\}\\[1ex]
&=\left\{s_{i}'\in S_{i}\left|s_{i}'(h)=s_{i}(h)\textup{  for any  }h\in\textup{H}_{i}(\textup{F}_{-i,m_{n}-1},\Upsilon)\right.\right\}\\[1ex]
&\subseteq\left\{s_{i}'\in S_{i}\left|s_{i}'(h)=s_{i}(h)\textup{  for any  }h\in\textup{H}_{i}(\textup{F}_{-i},\Upsilon)\right.\right\}.
\end{align*}
Thus, $(d^{n},t^{n})_{n\in\mathds{N}}$ is a sequence of profiles of subjective models converging to $(\Upsilon,t)$ such that for any player $i$ and $n\in\mathds{N}$,
\[
\textup{F}_{i}(d_{i}^{n},t_{i}^{n})\subseteq\left\{s_{i}'\in S_{i}\left|s_{i}'(h)=s_{i}(h)\textup{  for any  }h\in\textup{H}_{i}(\textup{F}_{-i},\Upsilon)\right.\right\}.
\]
Now, fix $n\in\mathds{N}$ and $s^{n}\in\textup{F}(d^{n},t^{n})$. Let $J\subseteq I$ be the set of those players $i$ for which $h^{0}\in H_{i}$. Clearly, $h^{0}\in\textup{H}_{i}(\textup{F},\Upsilon)$ and thus, it must hold that $s_{i}^{n}(h^{0})=s_{i}(h^{0})$. Then, obviously, for every player $j$ such that $(h^{0},(s_{i}^{n}(h^{0}))_{i\in J})\in H_{j}$, it holds that $(h^{0},(s_{i}^{n}(h^{0}))_{i\in J})\in\textup{H}_{j}(\textup{F},\Upsilon)$ as well. Hence, an easy inductive argument enables to conclude that $z(s^{n}|h^{0})=z(s|h^{0})$.
\end{proof}


\subsection{Proof of Lemma \ref{lemma:firstperturbation2}}
\label{subsection:tostrict}


\subsubsection{First perturbation of payoff structures}
\label{subsubsection:perturbation1}

Let $\Upsilon$ be a standard payoff structure. We first construct a perturbation $(\Upsilon^{n})_{n\in\mathds{N}}$ that will ensure that any strategy of any player that is extensive form rationalizable given $\Upsilon$ is also strictly rationalizable along the tail of the perturbation. We proceed in three steps:
\begin{enumerate}

\item Set $\Theta_{0}^{n}:=\Theta_{0}$.

\item Fix $n\in\mathds{N}$ and define the following set for each player $i$:
\[
\Theta_{i}^{n}:=\left\{\theta_{i}^{q_{i}}(s_{i},\theta_{i})\left|s_{i}\in S_{i},\theta_{i}\in\Theta_{i}\textup{  and  }q_{i}\in([n,\infty)\cap\mathds{Q})\cup\{\infty\}\right.\right\},
\]

\item For each player $i$, define utility function:
\[
u_{i}^{n}(z,(\theta_{0},(\theta_{j}^{q_{j}}(s_{j},\theta_{j}))_{j\in I})):=\left\{
\begin{tabular}{l l}
$u_{i}(z,\theta)+\frac{1}{q_{i}}$&$\textup{if}\,\,\,q_{i}\neq\infty\textup{  and}$\\[0.5ex]
&$\textup{\color{white}if}\,\,\,z=z(s'_{-i};s_{i})\textup{  for some  }s'_{-i}\in S_{-i},$\\[1ex]
$u_{i}(z,\theta)$&$\textup{otherwise}$,
\end{tabular}
\right.
\]
for each terminal history $z$ and each $(\theta_{0},(\theta_{j}^{q_{j}}(s_{j},\bar{\theta}_{j}))_{j\in I})\in\Theta_{0}^{n}\times\prod_{i\in I}\Theta_{i}^{n}$.

\end{enumerate}
Set $\Upsilon:=(\Theta_{0}^{n},(\Theta_{i}^{n},u_{i}^{n})_{i\in I})$. Clearly, $(\Upsilon^{n})_{n\in\mathds{N}}$ is a well-defined sequence of standard payoff structures converging to $\Upsilon$. Now, let us make few comments to better motivate the need of this construction:
\begin{itemize}

\item[(A)] For every player $i$, every strategy $s_{i}$, every payoff type $\theta_{i}\in\Theta_{i}$, and every $q_{i}\neq\infty$, there exists a payoff type $\theta_{i}^{q_{i}}(s_{i},\theta_{i})\in\Theta_{i}^{n}$ that mimics the payoffs of $\theta_{i}$ for player $i$ except for a slight improvement in the payoffs corresponding to the terminal nodes consistent with $s_{i}$ (thus breaking all possible ties in expected utility in favor of $s_{i}$). This payoff types are the key elements we use in Lemma \ref{lemma:firstperturbation2} below to break ties and move from extensive form rationalizability to strict rationalizability.

\item[(B)] With some abuse of terminology it can be considered that $\Theta$ is a subset of $\Theta^{n}$. That the states of nature of the former are contained in the latter is immediate from the construction. Player $i$'s payoff types in $\Theta_{i}$ are not per se contained in those in $\Theta_{i}^{n}$; however, there is a natural identification between payoff types $\theta_{i}$ and $\theta_{i}^{\infty}(\,\cdot\,,\,\cdot\,)$. For our proof of Lemma \ref{lemma:firstperturbation1} below we will, with some abuse of notation, leverage on this fact.

\end{itemize}


\subsubsection{From rationalizability to strict rationalizability}
\label{subsubsection:tostrict}

The proof of the following lemma (and also of Lemma \ref{lemma:secondperturbation}) will depend on the concept of \textit{finite conjecture}. Given standard payoff structure $\Upsilon$ we say that conjecture $\mu_{i}$ is \textit{finite} if for every history $h\in H_{i}\cup\{h^{0}\}$ it holds that the belief hierarchy on $\Theta$ induced by $\mu_{i}(h)$, $\pi_{i}(h):=\tau_{i}(\textup{marg}_{\Theta_{0}\times T_{-i}(\Theta)}\mu_{i}(h))$, is that of a finite type. Then, we have that:

\begin{lemma}
\label{lemma:firstperturbation2}
Let $\Gamma$ be an extensive form and $\Upsilon$, a standard payoff structure. Then, for any player $i$ and any $k\geq 0$ the following two hold:

\begin{itemize}

\item[$(i)$] For any finite type $t_{i} \in T_{i}(\Theta)$, and any strategy $s_{i}\in\textup{F}_{i,k}(\Upsilon,t_{i})$ there exists a $n_{k}^{i,1}\in\mathds{N}$ and a sequence of finite types $(t_{i}^{k,n})_{n\geq n_{k}^{i,1}}$ converging to $t_{i}$ such that $t_{i}^{n}\in T_{i}(\Theta^{n})$ and $s_{i}\in\textup{F}_{i,k}^{0}(\Upsilon^{n},t_{i}^{n})$ for any $n\geq n_{k}^{i,1}$.

\item[$(ii)$] There exists some $n_{k}^{i,2}\in\mathds{N}$ such that, for any $n\geq n_{k}^{i,2}$, $\textup{F}_{i,k}^{0}(\Upsilon^{n},t_{i})\subseteq\textup{F}_{i,k}(\Upsilon^{n},t_{i})$ for every type $t_{i}\in T_{i}(\Theta^{n})$ and $\textup{H}_{j}(\textup{F}_{i,k}^{0},\Upsilon^{n})=\textup{H}_{j}(\textup{F}_{i,k},\Upsilon^{n})$ for every $j\neq i$.

\end{itemize}

\end{lemma}
\begin{proof}
We proceed by induction on $k$. The initial case $(k=0)$ hold trivially for both claims, so we can focus on the proofs of the inductive step. Suppose that $k\geq 0$ is such that the claims hold; we verify next that they also hold for $k+1$. We prove each claim separately:

\vspace{0.3cm}

\noindent\textsc{Claim $(i)$.} Fix player $i$, finite type $t_{i}$ strategy $s_{i}\in\textup{F}_{i,k+1}(\Upsilon,t_{i})$ and finite conjecture $\mu_{i}$ that justifies the inclusion of $s_{i}$ in $\textup{F}_{i,k+1}(\Upsilon,t_{i})$ (whose existence is guaranteed by Lemma \ref{lemma:basiclemma} below). Next we will base on each $\mu_{i}$ to construct a conditional probability system for $S_{-i}\times\Theta_{0}^{n}\times T_{-i}(\Theta^{n})$, $\mu_{i}^{n}$, that will ensure that $s_{i}$ is in $\textup{F}_{i,k+1}^{0}(\Upsilon^{n},t_{i}^{k+1,n})$ for some finite type $t_{i}^{k+1,n}\in T_{i}(\Theta^{n})$. To this end set first $X_{i}^{-1}:=\emptyset$ and next, for each $\ell=0,\dots,k-1$, $X_{i}^{\ell}:=\textup{H}_{i}(\textup{F}_{-i,k-\ell},\Upsilon)$. Now, fix $\ell=0,\dots,k-1$ and history $h\in X_{i}^{\ell}\setminus X_{i}^{\ell-1}$. We know from the induction hypothesis that $(a)$ there exists some $n_{k-\ell}^{i,2}\in\mathds{N}$ such that $\textup{H}_{i}(\textup{F}_{-i,k-\ell},\Upsilon^{n})=\textup{H}_{i}(\textup{F}_{-i,k-\ell}^{0},\Upsilon^{n})$ for every $n\geq n_{k-\ell}^{i,2}$ and that $(b)$ for any pair $(s_{-i},t_{-i})$ in the support of the marginal of $\mu_{i}(h)$ on $S_{-i}\times T_{-i}(\Theta)$ there exists a sequence of finite types $(t_{-i}^{k-\ell,n}(s_{-i},t_{-i}))_{n\geq n_{k-\ell}^{i,1}(s_{-i},t_{-i})}$ converging to $t_{-i}$ and satisfying the conditions in claim $(i)$ for $k-\ell$. Then, set (remember that the support of $\mu(h)$ is finite):
\[
n_{h}^{i,1}:=\textup{max}\left\{\textup{max}\left\{\left.n_{\ell}^{i,1}(s_{-i},t_{-i})\right|(s_{-i},t_{-i})\in\textup{supp }\mu(h)\right\},n_{\ell}^{i,2}\right\}
\]
and define, for each $n\geq n_{h}^{i,1}$,
\begin{align*}
&\mu_{i}^{n}(h)[E]:=\mu_{i}(h)\left[\left\{(s_{-i},\theta_{0},t_{-i})\in S_{-i}\times\Theta_{0}\times T_{-i}(\Theta)\left|
\begin{tabular}{r l}
\multicolumn{2}{l}{$(s_{-i},\theta_{0},t_{-i}^{k-\ell,n}(s_{-i},t_{-i}))\in E$}
\end{tabular}
\right.\right\}\right].
\end{align*}
Finiteness of $\mu_{i}$ guarantees that $\mu_{i}(h)$ is a well-defined probability measure on $S_{-i}\times\Theta_{0}^{n}\times T_{-i}(\Theta^{n})$. Set then $n_{k+1}^{i,1}:=\textup{max}\{n_{h}^{i,1}|h\in X_{i}^{\ell}\setminus X_{i}^{\ell-1},\ell=0,\dots,k-1\}$.

Notice now that the following three properties hold:\footnote{Remember that $\textup{H}_{i}(\mu_{i})$ was defined in Section\ref{subsubsection:histories} as the set of histories in which conjecture $\mu_{i}$ updates beliefs from scratch.}
\begin{itemize}

\item[(1)]\textit{For every pair of histories $h\in\textup{H}_{i}(\mu_{i})\cap\textup{H}_{i}(\textup{F}_{-i,k-\ell},\Upsilon)$ and $h'\in H_{i}$, if the marginal on $S_{-i}$ of $\mu_{i}(h)$ assigns positive probability to $S_{-i}(h')$, then $h'\in\textup{H}_{i}(\textup{F}_{-i,k-\ell}^{0},\Upsilon^{n})$}. To see this fix history $h\in \textup{H}_{i}(\mu_{i})\cap\textup{H}_{i}(\textup{F}_{-i,k-\ell},\Upsilon)$ and notice first that:
\begin{align*}
\mu_{i}^{n}(h)&\left[\Theta_{0}^{n}\times\textup{Graph}\left(\textup{F}_{-i,k-\ell}^{0}(\Upsilon^{n},\,\cdot\,)\right)\right]=\\[1ex]
&=\mu_{i}(h)\left[ \Theta_{0}\times\left\{(s_{-i},t_{-i})\in S_{-i}\times T_{-i}(\Theta)\left|s_{-i}\in\textup{F}_{-i,k-\ell}^{0}(\Upsilon^{n},t_{i}^{n}(s_{-i},t_{-i}))\right\}\right.\right]\\[1ex]
&=\mu_{i}(h)\left[\Theta_{0}\times\textup{Graph}\left(\textup{F}_{-i,k-\ell}(\Upsilon,\,\cdot\,)\right)\right]=1.
\end{align*}
Thus, for every history $h'\in H_{i}$ such that the marginal of $\mu_{i}^{n}(h)$ on $S_{-i}$ assigns positive probability to $S_{-i}(h')$ we necessarily have that:
\[
h\in\textup{H}_{i}(\textup{F}_{-i,k-\ell},\Upsilon^{n})
\]
and thus the claim must hold.

\item[(2)]\textit{For every history $h'\in\textup{H}_{i}(\textup{F}_{-i,k-\ell},\Upsilon)$ there exists some $h\in\textup{H}_{i}(\mu_{i})\cap\textup{H}_{i}(\textup{F}_{-i,k-\ell},\Upsilon)$ such that the marginal on $S_{-i}$ of $\mu_{i}^{n}(h)$ assigns positive probability to $S_{-i}(h')$.} To see this, fix history $h'\in\textup{H}_{i}(\textup{F}_{-i,k-\ell},\Upsilon)$. If $h'\in\textup{H}_{i}(\mu_{i})$ then there is nothing needed to be proved. If $h'\notin\textup{H}_{i}(\mu_{i})$ then we know that there exists some $h\in\textup{H}_{i}(\mu_{i})$ such that $h\prec h'$ and the marginal of $\mu_{i}(h)$ on $S_{-i}$ puts positive probability on $S_{-i}(h')$. Obviously, since $h\prec h'$ and $h'\in\textup{H}_{i}(\textup{F}_{-i,k-\ell},\Upsilon)$ we conclude that, indeed, $h\in\textup{H}_{i}(\mu_{i})\cap\textup{H}_{i}(\textup{F}_{-i,k-\ell},\Upsilon)$. The fact that the marginals on $S_{-i}$ of $\mu_{i}(h)$ and $\mu_{i}^{n}(h)$ coincide ensures that the claim is correct.

\item[(3)] \textit{For every pair of histories $h\in\textup{H}_{i}(\mu_{i})\cap\textup{H}_{i}(\textup{F}_{-i,k-\ell},\Upsilon)$ and $h'\notin\textup{H}_{i}(\textup{F}_{-i,k-\ell},\Upsilon)$ the marginal of $\mu_{i}(h)$ on assigns zero probability to $S_{-i}(h')$.} To see it remember that we know from the induction hypothesis that $\textup{H}_{i}(\textup{F}_{-i,k-\ell},\Upsilon)\subseteq\textup{H}_{i}(\textup{F}_{-i,k-\ell}^{0},\Upsilon^{n})$. Then, it follows from property (1) above that for every $h\in\textup{H}_{i}(\textup{F}_{-i,k-\ell},\Upsilon)$, the marginal of $\mu_{i}^{n}(h)$ on $S_{-i}$ assigns zero probability to $S_{-i}(h')$ for every history $h'\notin\textup{H}_{i}(\textup{F}_{-i,k-\ell},\Upsilon)$.

\end{itemize}
We already defined $\mu_{i}^{n}(h)$ for histories in $h\in\textup{H}_{i}(\mu_{i})\cap\textup{H}_{i}(\textup{F}_{-i,k-\ell},\Upsilon)$. Properties $(2)$ and $(3)$ above allows for defining $\mu_{i}^{n}(h)$ \textit{exclusively} for every history $h\in\textup{H}_{i}(\textup{F}_{-i,k-\ell},\Upsilon)\setminus\textup{H}_{i}(\textup{F}_{-i,\ell+1},\Upsilon)$ using conditional probability. Clearly, $\mu_{i}^{n}$ is a well-defined conjecture. Set now type $t_{i}^{n}:=(\theta_{i}^{n}(s_{i},\theta_{i}(t_{i})),\pi_{i}^{n})$ (notice the use of comment (A) in Section \ref{subsubsection:perturbation1}), where:
\[
\pi_{i}^{n}:=\tau_{i}^{-1}\left(\textup{marg}_{\Theta_{0}^{n}\times T_{-i}(\Theta^{n})}\mu_{i}^{n}(h^{0})\right).
\]
Then, we have that $t_{i}^{n}$ is a finite element of $T_{i}(\Theta^{n})$, and sequence $(t_{i}^{n})_{n\geq n_{k+1}^{i,1}}$ converges to $t_{i}$. Notice that, by construction, $\mu_{i}^{n}$ is an element of $\textup{C}_{i,k}^{0}(\Upsilon^{n},t_{i}^{n})$, the induction hypotheses guarantee that at each history $h\in\textup{H}_{i}(\textup{F}_{-i,k}^{0},\Upsilon^{n})$ belief $\mu_{i}^{n}(h)$ assigns probability 1 to the graph of $\textup{F}_{-i,k}^{0}(\Upsilon^{n},\,\cdot\,)$ and that, for each $\ell=1,\dots k-1$, at each history $h\in\textup{H}_{i}(\textup{F}_{-i,k-\ell}^{0},\Upsilon^{n})\setminus\textup{H}_{i}(\textup{F}_{-i,k-\ell+1}^{0},\Upsilon^{n})$ belief $\mu_{i}^{n}(h)$ assigns probability 1 to the graph of $\textup{F}_{-i,k-\ell}^{0}(\Upsilon^{n},\,\cdot\,)$. In addition, it is evident that $r_{i}(\theta_{i}(t_{i}^{n}),\mu_{i}^{n})=[s_{i}]$. Hence, we conclude that $s_{i}\in\textup{F}_{i,k+1}^{0}(\Upsilon^{n},t_{i}^{n})$.
\hfill$\bigstar$

\vspace{0.3cm}

\noindent\textsc{Claim $(ii)$.} Fix player $i$ and set $\bar{n}_{k+1}^{i,1}:=\textup{max}\{n_{k}^{j,2}|\ell=1,\dots,k\textup{  and  }j\neq i\}$ where each $n_{k}^{j,2}$ verifies part $(ii)$ of the induction hypothesis for player $j\neq i$. Then, pick type $t_{i}$, strategy $s_{i}\in\textup{F}_{i,k+1}^{0}(\Upsilon^{n},t_{i})$ and conjecture $\mu_{i}$ that justifies the inclusion of $s_{i}$ in $\textup{F}_{i,k+1}^{0}(\Upsilon^{n},t_{i})$. We know from part $(ii)$ of the induction hypothesis that $\textup{F}_{-i,k}^{0}(\Upsilon^{n},t_{-i})\subseteq\textup{F}_{-i,k}(\Upsilon,t_{-i})$ and $\textup{H}_{i}(\textup{F}_{-i,k}^{0},\Upsilon^{n})=\textup{H}_{i}(\textup{F}_{-i,k},\Upsilon^{n})$ for every $t_{-i}\in T_{-i}(\Theta^{n})$. Hence, it clearly follow that $\mu_{i}$ justifies the inclusion of $s_{i}$ in $\textup{F}_{i,k+1}(\Upsilon^{n},t_{i})$. Now, we divide the proof of the second claim in the following three-part cycle:

\begin{itemize}

\item $\textup{H}_{j}(\textup{F}_{i,k+1}^{0},\Upsilon^{n})\subseteq\textup{H}_{j}(\textup{F}_{i,k+1},\Upsilon^{n})$ for any $j\neq i$ and any $n\geq\bar{n}_{k+1}^{i,1}$. This follows immediately from part $(i)$, just proved above.

\item There exists some $\bar{n}_{k+1}^{i,2}\in\mathds{N}$ such that $\textup{H}_{j}(\textup{F}_{i,k+1},\Upsilon^{n})\subseteq\textup{H}_{j}(\textup{F}_{i,k+1},\Upsilon)$ for any player $j\neq i$ and any $n\geq\bar{m}_{k+1}$. Pick $n_{k+1}^{i}$; according to part $(iii)$ of Lemma \ref{lemma:firstperturbation1} we have that $\textup{H}_{j}(\textup{F}_{i,k+1},\Upsilon^{n})\subseteq\textup{H}_{j}(\textup{F}_{i,k+1},\Upsilon)$ for any $n\geq \bar{n}_{k+1}^{2,i}$.

\item There exists some $\bar{n}_{k+1}^{i,3}\in\mathds{N}$ such that $\textup{H}_{j}(\textup{F}_{i,k+1},\Upsilon)\subseteq\textup{H}_{j}(\textup{F}_{i,k+1}^{0},\Upsilon^{n})$ for any $j\neq i$ and any $n\geq\bar{n}_{k+1}^{i,3}$. Let $T_{i}^{0}(\Theta)$ denote the set of finite types in $T_{i}(\Theta)$. Because of upper hemicontinuity of $\textup{F}_{i,k+1}(\Upsilon,\,\cdot\,)$ we know that every $t_{i}\in  T_{i}^{0}(\Theta)$ has some open neighborhood $U_{i}(t_{i})$ such that:
\[
\bigcup_{t_{i}'\in U_{i}(t_{i})}\{t_{i}'\}\times\textup{F}_{i,k+1}(\Upsilon,t_{i}')\subseteq U_{i}(t_{i})\times\textup{F}_{i,k+1}(\Upsilon,t_{i}).
\]
Then, because of denseness of $T_{i}^{0}(d_{i})$ we have that:
\begin{align*}
\textup{Graph}\left(\textup{F}_{i,k+1}(\Upsilon,\,\cdot\,)\right)&= \bigcup_{t_{i}\in T_{i}^{0}(\Theta)}\bigcup_{t_{i}'\in U_{i}(t_{i})}\{t_{i}'\}\times\textup{F}_{i,k+1}(\Upsilon,t_{i}')\\[1ex]
&=\bigcup_{t_{i}\in T_{i}^{0}(\Theta)}U_{i}(t_{i})\times\textup{F}_{i,k+1}(\Upsilon,t_{i})
\end{align*}
Thus, $h\in\textup{H}_{j}(\textup{F}_{i,k+1},\Upsilon)$ if and only if there exists some finite $t_{i}$ such that $s_{i}\in S_{i}(h)\cap\textup{F}_{i,k+1}(\Upsilon,t_{i})$. In consequence, by finiteness of $H_{i}$ and claim $(i)$ there exists some $\bar{n}_{k+1}^{i,3}$ such that, for every $n\geq\bar{n}_{k+1}^{i,3}$, if $h\in\textup{H}_{j}(\textup{F}_{i,k+1},\Upsilon)$, then $h\in\textup{H}_{j}(\textup{F}_{i,k+1}^{0},\Upsilon^{n})$. 
\end{itemize}
Thus, by setting $n_{k+1}^{i,2}:=\textup{max}\{\bar{n}_{k+1}^{i,1},\bar{n}_{k+1}^{i,2},\bar{n}_{k+1}^{i,3}\}$ we conclude that for any $j\neq i$ and any $n\geq\bar{n}_{k+1}^{i,2}$,
\[
\textup{H}_{j}(\textup{F}_{i,k+1},\Upsilon)=\textup{H}_{j}(\textup{F}_{i,k+1},\Upsilon^{n})=\textup{H}_{j}(\textup{F}_{i,k+1}^{0},\Upsilon^{n}),
\]
and hence, the proof is complete.
\end{proof}


\subsection{Proof of Lemma \ref{lemma:secondperturbation}}
\label{subsection:expanding}

\begin{lemma}
\label{lemma:secondperturbation}
Let $\Gamma$ be an extensive form. Then, for any $k\in\mathds{N}$, any player $i$, any subjective payoff structure $d_{i}$, any finite type $t_{i}\in T_{i}(d_{i})$ and any strategy $s_{i}\in\textup{F}_{i,k}^{0}(d_{i},t_{i})$ there exists a sequence of finite types $(t_{i}^{n})_{n\in\mathds{N}}$ in $T_{i,k}(d_{i},s_{i})$ converging to $t_{i}$.
\end{lemma}
\begin{proof}
Fix $k\in\mathds{N}$, player $i$, subjective payoff structure $d_{i}$, type $t_{i}\in T_{i}(d_{i})$ and strategy $s_{i}\in\textup{F}_{i,k}^{0}(d_{i},t_{i})$, and pick finite conjecture $\mu_{i}$ that justifies the inclusion of $s_{i}$ in $\textup{F}_{i,k}^{0}(d_{i},t_{i})$ (the existence of such a conjecture is guaranteed by Lemma \ref{lemma:basiclemma} below). Based on $H_{i}(\mu_{i})$ we define next:\footnote{Remember that $\textup{H}_{i}(\mu_{i})$ was defined in Section\ref{subsubsection:histories} as the set of histories in which conjecture $\mu_{i}$ updates beliefs from scratch.}
\[
H_{i}^{0}(\mu_{i}):=H_{i}(\mu_{i})\cap\textup{H}_{i}(\textup{F}_{-i,k-1}^{0},d_{-i|i}).
\]
The interest of $H_{i}^{0}(\mu_{i})$ is that it will eventually allow for recovering conjecture $\mu_{i}$ at those histories consistent $\textup{F}_{-i,k-1}^{0}(d_{-i|i},\,\cdot\,)$. For each history $h\in H_{i}^{0}(\mu_{i})$ we are going to find a sequence of finitely-supported beliefs $
(\hat{\mu}_{i}^{m}(h))_{m\in\mathds{N}}\subseteq\Delta(S_{-i}(h)\times d_{i,1}\times T_{-i}(d_{-i|i}))$ that converges to $\mu_{i}(h)$ and satisfies certain desirable properties described below. For each $m\in\mathds{N}$ define probability measure $\mu^{m}(h^{0})$ on $\Delta(S_{-i}\times d_{i,1}\times T_{-i}(d_{-i|i}))$ by setting:
\[
\mu_{i}^{m}(h^{0})=\left(1-\frac{1}{m}\right)\cdot\mu_{i}(h^{0})+\left(\frac{1}{m}\right)\cdot\sum_{h\in H_{i}^{0}(\mu_{i})}\left(
\frac{1}{|H_{i}^{0}(\mu_{i})|}\right)\cdot\mu_{i}(h).
\]
Clearly, $\mu_{i}^{m}(h^{0})$ is well-defined and finitely-supported. Now, for each $m\in\mathds{N}$ denote:
\[
H_{i}^{m}=\left\{h\in H_{i}\cup\{h^{0}\}\left|\left(\textup{marg}_{S_{-i}}\mu_{i}^{m}(h^{0})\right)[S_{-i}(h)]>0\right.\right\},
\]
and for any $h\in H_{i}^{m}$ define belief $\mu_{i}^{m}(h)\in\Delta(S_{-i}\times d_{i,1}\times T_{-i}(d_{-i|i}))$ by extending $\mu_{i}^{m}(h^{0})$ via conditional probability. Finally, let conjecture $\mu_{i}^{m}\in\Delta^{H_{i}\cup\{h^{0}\}}(S_{-i}\times d_{i,1}\times T_{-i}(d_{-i|i}))$ and type $\hat{t}_{i}^{m}=(\theta_{i}(t_{i}),\pi_{i}^{m})$ be respectively defined as follows:
\begin{align*}
\mu_{i}^{m}(h):=\left\{
\begin{tabular}{l l}
$\mu_{i}^{m}(h)$&$\textup{for every  }h\in H_{i}^{m}$,\\
$\mu_{i}(h)$&$\textup{otherwise,}$
\end{tabular}
\right.\textup{\hspace{0.05cm}and\hspace{0.2cm}}\pi_{i}^{m}:=\tau_{i}^{-1}\left(\textup{marg}_{(d_{i,1})_{0}\times T_{-i}(d_{-i|i})}\mu_{i}^{m}(h^{0})\right).
\end{align*}
We claim now that the following four properties are satisfied:

\begin{itemize}

\item[$(a)$] $(\hat{t}_{i}^{m})_{m\in\mathds{N}}$ is a sequence of finite types that converges to $t_{i}$. For finiteness simply notice that $\mu_{i}^{m}(h^{0})$ is finitely-supported for every $m\in\mathds{N}$. Convergence follows from $(\mu_{i}^{m}(h^{0}))_{m\in\mathds{N}}$ converging to $\mu_{i}(h^{0})$ and marginalization being continuous.

\item[$(b)$] For any $m\in\mathds{N}$ the marginal on $S_{-i}$ of $\mu_{i}^{m}(h^{0})$ assigns positive probability to $S_{-i}(h)$ for every history $h\in\textup{H}_{i}(\textup{F}_{-i,k-1}^{0},d_{i})$. Fix such $h$. If $h\in H_{i}^{0}(\mu_{i})$, then:
\begin{align*}
(\textup{marg}_{S_{-i}}\mu_{i}^{m}(h^{0}))\left[S_{-i}(h)\right]&\geq \left(\frac{1}{m}\right)\cdot\left(\frac{1}{|H_{i}^{0}(\mu_{i})|}\right)\cdot(\textup{marg}_{S_{-i}}\mu_{i}(h))\left[S_{-i}(h)\right]\\[1ex]
&=\left(\frac{1}{m}\right)\cdot\left(\frac{1}{|H_{i}^{0}(\mu_{i})|}\right)>0.
\end{align*}
If $h\notin H_{i}^{0}(\mu_{i})$, then there must exist some history $h'\prec h$ such that $h'\in H_{i}^{0}(\mu_{i})$ and $(\textup{marg}_{S_{-i}}\mu_{i}^{m}(h'))\left[S_{-i}(h)\right]>0$, and thus:
\begin{align*}
(\textup{marg}_{S_{-i}}\mu_{i}^{m}(h^{0}))\left[S_{-i}(h)\right]&\geq\left(\frac{1}{m}\right)\cdot\left(\frac{1}{|H_{i}^{0}(\mu_{i})|}\right)\cdot(\textup{marg}_{S_{-i}}\mu_{i}(h'))\left[S_{-i}(h)\right]>0.
\end{align*}

\item[$(c)$] For any $m\in\mathds{N}$, $\mu_{i}^{m}$ is a well-defined conditional probability system consistent with $d_{i}$ and $\hat{t}_{i}^{m}$. Consistency with $d_{i}$ and $t_{i}^{m}$ holds by construction; thus, all we need to check is that $\mu_{i}^{m}$ does not violate the chain rule. This is easy to see by simply noticing that for any pair of different histories $h,h'\in H_{i}^{0}(\mu_{i})$ the marginal on $S_{-i}$ of $\mu_{i}(h)$ puts zero probability on $S_{-i}(h')$. This implies that for any $h,h'\in H_{i}^{m}$ there are no inconsistency issues arising from belief update.\footnote{Because if $h,h'\in H_{i}^{m}$ then there exist $\bar{h},\bar{h'}\in H_{i}^{0}(\mu_{i})$ such that $\mu_{i}(\bar{h})$ induces $\mu_{i}^{m}(h)$ and $\mu_{i}(\bar{h}')$ induces $\mu_{i}^{m}(h')$. Since the marginal on $S_{-i}$ of $\mu_{i}(\bar{h})$ (resp. $\mu_{i}(\bar{h}')$) puts zero probability on $S_{-i}(\bar{h'})$ (resp. $S_{-i}(\bar{h})$), it follows that the marginal on $S_{-i}$ of $\mu_{i}^{m}(h)$ (resp. $\mu_{i}^{m}(h')$) puts zero probability on $S_{-i}(h')$ (resp. $S_{-i}(h)$).} Clearly, there are no problems either for any pair $h,h'\notin H_{i}^{m}$ (due to $\mu_{i}$ being a conditional probability system). Since for any $h\in H_{i}^{m}$ the marginal on $S_{-i}$ of $\mu_{i}(h)$ puts zero probability on $S_{-i}(h')$ for any $h'\notin H_{i}^{m}$, it follows that pairs $h\in H_{i}^{m}$ and $h'\notin H_{i}^{m}$ are not problematic either.

In addition it holds by construction that:
\[
\textup{supp }\mu_{i}^{m}(h^{0})\subseteq\textup{Graph}\left(\textup{F}_{-i,k-1}^{0}(d_{-i|i},\,\cdot\,)\right).
\]
Thus we have that if $h\in\textup{H}_{i}(\textup{F}_{-i,k-1}^{0},d_{-i|i})$ then:
\[
\textup{supp }\mu_{i}^{m}(h)\subseteq\textup{supp }\mu_{i}^{m}(h^{0})\subseteq\textup{Graph}\left(\textup{F}_{-i,k-1}^{0}(d_{-i|i},\,\cdot\,)\right).
\]
Hence belief in appropriate behavior is maintained whenever possible to do so. If $h\notin\textup{H}_{i}(\textup{F}_{-i,k-1}^{0},d_{-i|i})$ then:
\[
\textup{supp }\mu_{i}^{m}(h)=\textup{supp }\mu_{i}(h),
\]
what guarantees belief in appropriate continuation behavior.

\item[$(d)$] There exists some $m_{0}\in\mathds{N}$ such that $r_{i}(\theta_{i}(t_{i}),\mu_{i}^{m})=[s_{i}]$ for any $m\geq m_{0}$. This follows trivially from finiteness of $S_{i}$, continuity of conditional expected payoffs and the fact that $(\mu_{i}^{m})_{m\in\mathds{N}}$ converges to $\mu_{i}$. 

\end{itemize}
This way, we conclude that $s_{i}\in\textup{F}_{i,k}^{0}(d_{i},t_{i}^{m})$ for every $m\geq m_{0}$. Hence, if for each $n\in\mathds{N}$ we relabel as $t_{i}^{n}:=\hat{t}_{i}^{n+m_{0}}$, $(t_{i}^{n})_{n\in\mathds{N}}$ is the sequence we are looking for.
\end{proof}


\subsection{Proof of Lemma \ref{lemma:thirdperturbation2}}
\label{subsection:touniqueness}


\subsubsection{Second perturbation of payoff structures}
\label{subsubsection:perturbation2}

Fix some standard payoff structure $\Upsilon$. For each $k\in\mathds{N}\cup\{0\}$, each player $i$ and each let $d_{i}^{k}$ denote the subjective payoff structure constructed recursively by setting:
\begin{itemize}

\item[$(0)$] $d_{i}^{0}(\Theta):=\Upsilon^{*}$, where $\Upsilon^{*}$ is some arbitrary standard payoff structure satisfying richness.

\item[$(1)$] $d_{i}^{1}(\Upsilon)$ is defined by setting $d_{i,1}^{1}(\Upsilon):=\Theta$ and $d_{-i|i}^{1}:=d_{-i}^{0}(\Upsilon)$.

\item[$(k)$] $d_{i}^{k}(\Upsilon)$ is defined by setting $d_{i,1}^{k}(\Upsilon):=\Theta$ and $d_{-i|i}^{k}(\Upsilon):=d_{-i}^{k-1}(\Upsilon)$

\end{itemize}
Obviously, each $d_{i}^{k}(\Upsilon)$ is a subjective payoff structure satisfying higher-order richness, and sequence $(d_{i}^{k}(\Upsilon))_{k\in\mathds{N}}$ converges to $\Upsilon$.


\subsubsection{From strictness to uniqueness: Initial step}
\label{subsubsection:tounique1}

\begin{lemma}
\label{lemma:thirdperturbation1}
Let $\Gamma$ be an extensive form and $\Upsilon$, a standard payoff structure. Then, for any player $i$, any finite type $t_{i}\in T_{i}(\Theta)$, any strategy $s_{i}\in\textup{F}_{i,1}^{0}(\Upsilon,t_{i})$ such that $t_{i}\in T_{i}(\Upsilon,s_{i})$, there exists a finite type $t_{i}^{1}$ such that:
\begin{itemize}

\item[$(i)$] $t_{i}^{1}\in T_{i}(d_{i}^{1}(\Upsilon))$.

\item[$(ii)$] $\theta_{i}(t_{i}^{1})=\theta_{i}(t_{i})$.

\item[$(iii)$] $\textup{F}_{i,2}(d_{i}^{1}(\Upsilon),t_{i}^{1})\subseteq [s_{i}|\Upsilon]_{0}$.

\end{itemize}
\end{lemma}
\begin{proof}
For notational convenience, set $d_{i}^{0}:=d_{i}^{0}(\Theta)$ and $d_{i}^{1}:=d_{i}^{1}(\Theta)$. We are going to show first that for any player $i$, any finite type $t_{i}$ and any strategy $s_{i}$ there exists a finite type $t_{i}^{0}$ such that:
\begin{itemize}

\item[$(i)$] $\textup{F}_{i,1}(d_{i}^{0},t_{i}^{0})=\{s_{i}\}$.

\item[$(ii)$] $\theta_{i}(t_{i}^{0})$ and $\theta_{i}(t_{i})$ are payoff-equivalent for player $i$'s opponents.

\end{itemize}
To see it fix player $i$, finite type $t_{i}$ and strategy $s_{i}$ and pick player $i$'s payoff type $\theta_{i}(t_{i},s_{i})$ that makes $s_{i}$ conditionally dominant and is payoff-equivalent to $\theta_{i}(t_{i})$ for player $i$'s opponents.\footnote{Richness of $\Upsilon^{*}$ allows for such a selection.} Then, set $t_{i}^{0}:=(\theta_{i}(t_{i},s_{i}),\pi_{i}(t_{i}))$. Obviously, $t_{i}^{0}$ is finite. Pick now arbitrary conjecture $\mu_{i}$ consistent with subjective model $(d_{i}^{0},t_{i}^{0})$. By definition, for any history $h\in H_{i}\cup\{h^{0}\}$ and any strategy $s_{i}'$ such that $s_{i}'(h)\neq s_{i}(h)$:
\[
U_{i}(s_{i},\mu_{i}|\theta_{i}(t_{i}^{0}),h)>U_{i}(s'_{i},\mu_{i}|\theta_{i}(t_{i}^{0}),h).
\]
Hence, $r_{i}(\theta_{i}(t_{i}^{0}),\mu_{i})=\{s_{i}\}$ and thus, we conclude that $\textup{F}_{i,1}(d_{i}^{0},t_{i}^{0})=\{s_{i}\}$. In addition, as said above, $\theta_{i}(t_{i}^{0})$ is by construction payoff-equivalent to $\theta_{i}(t_{i})$ for player $i$'s opponents.

Now, to prove the claim of the lemma fix player $i$, finite type $t_{i}\in T_{i}(\Theta)$ and strategy $s_{i}$ such that $s_{i}\in\textup{F}_{i,1}^{0}(\Upsilon,t_{i})$ and $t_{i}\in T_{i,1}(\Upsilon,s_{i})$. We know then that there exists some conjecture $\bar{\mu}_{i}$ consistent with $\Theta$ and $t_{i}$ and such that: $(a)$ $r_{i}(\theta_{i}(t_{i}),\bar{\mu}_{i})=\{\bar{s}_{i}\}$ and $(b)$ $(\textup{marg}_{S_{-i}}\bar{\mu}_{i}(h^{0}))[S_{-i}(h)]>0$ for every history $h\in\textup{H}_{i}(\textup{F}_{-i,0}^{0},\Upsilon)$. Remember now that we just proved above that for any pair $(s_{-i},t_{-i})$ in the support of the marginal of $\bar{\mu}_{i}(h^{0})$ on $S_{-i}\times T_{-i}(d_{-i}^{0})$ there exists some finite type $t_{-i}^{0}(s_{-i},t_{-i})$ such that for any player $j\neq i,$
\begin{itemize}

\item $\textup{F}_{j,1}(d_{j}^{0},t_{j}^{0}(s_{j},t_{j}))=\{s_{j}\}$.

\item $\theta_{j}(t_{j}^{0}(s_{j},t_{j}))$ is payoff-equivalent to $\theta_{j}(t_{j})$ for player $j$'s opponents.

\end{itemize}
Define then $\bar{\mu}_{i}^{1}$ as follows:
\[
\bar{\mu}_{i}^{1}[E]:=\bar{\mu}_{i}(h^{0})\left[\left\{(s_{-i},\theta_{0},t_{-i})\in S_{-i}\times\Theta_{0}\times T_{-i}(\Theta)\left|(s_{-i},\theta_{0},t_{-i}^{0}(s_{-i},t_{-i}))\in E\right.\right\}\right],
\]
for any measurable $E\subseteq S_{-i}\times (d_{i,1}^{1})_{0}\times T_{-i}(d_{-i}^{0})$. Since $S_{-i}$ and $t_{i}$ are finite it is immediate that $\bar{\mu}_{i}^{1}$ is a well-defined probability measure in $\Delta(S_{-i}\times (d_{i,1}^{1})_{0}\times T_{-i}(d_{-i}^{0}))$. Define next belief hierarchy:
\[
\pi_{i}^{1}:=\tau_{i}^{-1}\left(\textup{marg}_{(d_{i,1}^{1})_{0}\times T_{-i}(d_{-i}^{0})}\bar{\mu}_{i}^{1}\right),
\]
and set $t_{i}^{1}:=(\theta_{i}(t_{i}),\pi_{i}^{1})$. Clearly, $t_{i}^{1}$ is finite and satisfies that $\theta_{i}(t_{i}^{1})=\theta_{i}(t_{i})$. Note in addition that the supports of the marginals on $\Theta_{0}$ of $\bar{\mu}_{i}(h^{0})$ and $\bar{\mu}_{i}^{1}$ coincide, and thus, that the support of the marginal on $\Theta_{0}$ of $\tau_{i}(\pi_{i}(t_{i}^{1}))$ is contained in $(d_{i,1}^{1})_{0}$. It follows then that $t_{i}^{1}\in T_{i}(d_{i}^{1})$. Let's finish verifying that:
\[
\textup{F}_{i,2}(d_{i}^{1},t_{i}^{1})\subseteq [\bar{s}_{i}|\Upsilon]_{0}.
\]
To see it, pick arbitrary $s_{i}\in\textup{F}_{i,2}(d_{i}^{1},t_{i}^{1})$ and conjecture $\mu_{i}$ that justifies the inclusion of $s_{i}$ in $\textup{F}_{i,2}(d_{i}^{1},t_{i}^{1})$. Now, for each pair $(s_{-i},\theta_{-i})$ denote:
\[
T_{-i}^{0}(s_{-i},\theta_{-i}):=\left\{t_{-i}^{0}(s_{-i},t_{-i})\left|t_{-i}\in T_{-i}(\Theta)\textup{  and  }\theta_{-i}(t_{-i}^{0}(s_{-i},t_{-i}))=\theta_{-i}\right.\right\},
\]
and notice that for any $(s_{-i},\theta_{-i},\theta_{0})$ we have that:\footnote{We provide further clarification of the following development. The first two equalities follow from the construction of $t_{i}^{1}$, from $\mu_{i}$ initially assigning probability $1$ to the graph of $\textup{F}_{-i,1}(d_{-i}^{0},\,\cdot\,)$ and the fact that $\textup{F}_{-i,1}(d_{-i}^{0},t_{-i}^{0}(s_{-i},t_{-i}))=\{s_{-i}\}$. The third one follows from the construction of $\bar{\mu}_{i}^{1}$. The fourth is obvious. The fifth equality follows, again, from the construction of $\bar{\mu}^{1}$. The last one follows from the fact that if $t_{-i}^{0}(s_{-i}',t_{-i})=t_{-i}^{0}(s_{-i},t_{-i})$ then $s_{-i}'=s_{-i}$, because $\textup{F}_{-i,1}(d_{-i}^{0},t_{-i}^{0}(s_{-i},t_{-i}))=\{s_{-i}\}$.}
\begin{align*}
\mu_{i}(h^{0})&\left[\{(s_{-i},\theta_{-i},\theta_{0})\}\times\Pi_{-i}(\Theta)\right]=\\[2ex]
&=\mu_{i}(h^{0})\left[\{(s_{-i},\theta_{0})\}\times T_{-i}^{0}(s_{-i},\theta_{-i})\right]\\[2ex]
&=\mu_{i}(h^{0})\left[S_{-i}\times\{\theta_{0}\}\times T_{-i}^{0}(s_{-i},\theta_{-i})\right]\\[2ex]
&=\bar{\mu}_{i}^{1}\left[S_{-i}\times\{\theta_{0}\}\times T_{-i}^{0}(s_{-i},\theta_{-i})\right]\\[2ex]
&=\sum_{s_{-i}'\in S_{-i}}\bar{\mu}_{i}^{1}\left[\{(s_{-i}',,\theta_{0})\}\times T_{-i}^{0}(s_{-i},\theta_{-i})\right]\\[2ex]
&=\sum_{s_{-i}'\in S_{-i}}\bar{\mu}_{i}(h^{0})\left[\{(s_{-i}',\theta_{0})\}\times\left\{t_{-i}\in T_{-i}(\Theta)\left|t_{-i}^{0}(s_{-i}',t_{-i})\in T_{-i}^{0}(s_{-i},\theta_{-i})\right.\right\}\right]\\[2ex]
&=\bar{\mu}_{i}(h^{0})\left[\{(s_{-i},\theta_{0})\}\times\left\{t_{-i}\in T_{-i}(\Theta)\left|t_{-i}^{0}(s_{-i},t_{-i})\in T_{-i}^{0}(s_{-i},\theta_{-i})\right.\right\}\right].
\end{align*}
Now, notice two things. First, we know that for any $(s_{-i},t_{-i})$ in the support of the marginal of $\bar{\mu}_{i}(h^{0})$ on $S_{-i}\times T_{-i}(\Theta)$ $\theta_{-i}(t_{-i})$ and $\theta_{-i}(t_{-i}^{0}(s_{-i},t_{-i}))$ are payoff-equivalent for player $i$; thus, it follows from the above that $\mu_{i}(h^{0})$ and $\bar{\mu}_{i}(h^{0})$ induce the same conditional expected payoffs at $h^{0}$. Second, the marginal of $\bar{\mu}_{i}(h^{0})$ on $S_{-i}$ assigns positive probability to $S_{-i}(h)$ for every $h\in H_{i}\cup\{h^{0}\}$; in consequence, it follows from the above that $\mu_{i}$ and $\bar{\mu}_{i}$ induce the same conditional expected payoffs at $h$ for every history $h\in H_{i}\cup\{h^{0}\}$. Thus, we conclude that $r_{i}(\theta_{i}(t_{i}^{1}),\mu_{i})=\{\bar{s}_{i}\}$ and, thus, that $s_{i}\in [\bar{s}_{i}|\Upsilon]_{0}$.
\end{proof}

\subsubsection{From strictness to uniqueness: Inductive step}
\label{subsubsection:tounique2}

\begin{lemma}
\label{lemma:thirdperturbation2}
Let $\Gamma$ be an extensive form. Then for any $k\geq 1$, any player $i$, any standard payoff structure $\Upsilon$ such that $\textup{H}_{j}(\textup{F}_{i,k+1},\Upsilon)\subseteq\textup{H}_{j}(\textup{F}_{i,k+1}^{0},\Upsilon)$ for every $j\neq i$, any finite $t_{i}\in T_{i}(\Theta)$ and any strategy $s_{i}\in\textup{F}_{i,k}^{0}(\Upsilon,t_{i})$ such that $t_{i}\in T_{i}(\Upsilon,s_{i})$ there exists a finite type $t_{i}^{k}$ such that:

\begin{itemize}

\item[$(i)$] $t_{i}^{k}\in T_{i}(d_{i}^{k}(\Upsilon))$.

\item[$(ii)$] $\theta_{i}(t_{i}^{k})=\theta_{i}(t_{i})$ and $\pi_{i,k}(t_{i}^{k})=\pi_{i,k}(t_{i})$.

\item[$(iii)$] $\textup{F}_{i,k+1}(d_{i}^{k}(\Upsilon),t_{i}^{k})\subseteq [s_{i}|\Upsilon]_{k-1}$.

\end{itemize}
\end{lemma}
\begin{proof}
We proceed by induction on $k$. The initial case ($k=0$) is covered by Lemma \ref{lemma:thirdperturbation1} (the second part of property $(ii)$ holds by vacuity) so we can focus on the proof of the inductive step. Suppose that $k\geq 0$ is such that the claims hold. We verify next that they also holds for $k+1$. Fix first standard payoff structure $\Upsilon$ such that $\textup{H}_{j}(\textup{F}_{i,k+1},\Upsilon)\subseteq\textup{H}_{j}(\textup{F}_{i,k+1}^{0},\Upsilon)$ for every $j\neq i$ and set, for notational convenience, $d_{i}^{k+1}:=d_{i}^{k+1}(\Upsilon)$; let us proceed now in two steps:

\vspace{0.3cm}

\noindent\textit{1. Parts $(i)$ and $(ii)$: Construction of the required finite type}

\vspace{0.3cm}

\noindent Fix player $i$, finite $t_{i}\in T_{i}(\Theta)$ and strategy $\bar{s}_{i}$ such that $\bar{s}_{i}\in\textup{F}_{i}^{0}(\Upsilon,t_{i})$ and $t_{i}\in T_{i}(\Upsilon,\bar{s}_{i})$. Pick conjecture $\bar{\mu}_{i}$ that justifies the inclusion of $\bar{s}_{i}$ in $\textup{F}_{i,k+1}^{0}(\Upsilon,t_{i})$ and the inclusion of $t_{i}$ in $T_{i}(\Upsilon,\bar{s}_{i})$. Now, we know from the induction hypothesis that for any pair $(s_{-i},t_{-i})$ in the support of the marginal of $\bar{\mu}_{i}(h^{0})$ on $S_{-i}\times T_{-i}(\Theta)$ there exists some finite type $t_{-i}^{k}(s_{-i},t_{-i})$ such that for every player $j\neq i$,
\begin{itemize}

\item $t_{j}^{k}\in T_{j}(d_{j}^{k})=T_{j}(d_{j|i}^{k+1})$.

\item $\theta_{j}^{k}=\theta_{j}(t_{j})$ and $\pi_{j,k}^{k}=\pi_{j,k}(t_{j})$ for $\theta_{j}^{k}=\theta_{j}(t_{j}^{k}(s_{j},t_{j}))$ and $\pi_{j}^{k}=\pi_{j}(t_{j}^{k}(s_{j},t_{j}))$.

\item $\textup{F}_{j,k+1}(d_{j|i}^{k+1},t_{j}^{k}(s_{j},t_{j}))=\textup{F}_{j,k+1}(d_{j}^{k},t_{j}^{k}(s_{j},t_{j}))\subseteq [s_{j}|\Upsilon]_{k-1}$.

\end{itemize}
Define then probability measure $\bar{\mu}_{i}^{k+1}\in\Delta(S_{-i}\times (d_{i,1}^{k+1})_{0}\times T_{-i}(d_{-i}^{k}))$ by setting:
\[
\bar{\mu}_{i}^{k+1}[E]:=\bar{\mu}_{i}(h^{0})\left[\left\{(s_{-i},\theta_{0},t_{-i})\in S_{-i}\times (d_{i,1}^{k+1})_{0}\times T_{-i}(\Theta)\left|(s_{-i},\theta_{0},t_{-i}^{k}(s_{-i},t_{-i}))\in E\right.\right\}\right],
\]
for any measurable $E\subseteq S_{-i}\times (d_{i,1}^{k+1})_{0}\times T_{-i}(d_{-i}^{k})$. Finiteness of $t_{i}$ guarantees that $\bar{\mu}_{i}^{k+1}$ is well-defined. Obviously, $\bar{\mu}_{i}^{k+1}$ does not necessarily induce a whole conditional probability system via conditional probability (its marginal on $S_{-i}$ may not have full-support); however, if we set $\bar{\mu}_{i}^{k+1}(h^{0}):=\bar{\mu}_{i}^{k+1}$ conditional probability enables to obtain a family of beliefs $\{\bar{\mu}_{i}^{k+1}(h)|h\in\textup{H}_{i}(\textup{F}_{-i}^{0},\Upsilon)\}$. Notice that for any measurable $E\subseteq S_{-i}\times\Theta_{0}\times \Theta_{-i}$ it holds that:
\[
\bar{\mu}_{i}^{k+1}(h)[E\times\Pi_{-i}(d_{-i}^{k})]=\bar{\mu}_{i}(h)[E\times\Pi_{-i}(\Theta)],
\]
for every history $h\in\textup{H}_{i}(\textup{F}_{-i,k}^{0},\Upsilon)$. Thus, for any such history $h$ we can define the conditional expected payoff induced by belief $\bar{\mu}_{i}^{k+1}(h)$ and any strategy $s_{i}$ for payoff type $\theta_{i}(t_{i})$, which, with some abuse of notation, we represent by $U_{i}(s_{i},\bar{\mu}_{i}^{k+1}|\theta_{i}(t_{i}),h)$, and, clearly, satisfies that:
\begin{equation}
\label{equation:equality}
U_{i}(s_{i},\bar{\mu}_{i}^{k+1}|\theta_{i}(t_{i}),h)=U_{i}(s_{i},\bar{\mu}_{i}|\theta_{i}(t_{i}),h).
\end{equation}
Define now belief hierarchy:
\[
\pi_{i}^{k+1}:=\tau_{i}^{-1}\left(\textup{marg}_{(d_{i,1}^{k+1})_{0}\times T_{-i}(d_{-i}^{k})}\bar{\mu}_{i}^{k+1}\right),
\]
and set $t_{i}^{k+1}:=(\theta_{i}(t_{i}),\pi_{i}^{k+1})$. Clearly, $t_{i}^{k+1}$ is finite and has the same payoff type and $k^{\textup{th}}$-order beliefs as $t_{i}$. Thus, it follows that $t_{i}^{k+1}\in T_{i}(d_{i}^{k+1})$.\footnote{Finiteness is immediate. To see that lower-order beliefs are maintained, simply notice that the probability assigned by $\bar{\mu}_{i}(h^{0})$ to $(s_{-i},\theta_{0},t_{-i})$ is assigned by $\bar{\mu}_{i}^{k+1}(h^{0})$ to $(s_{-i},\theta_{0},t_{-i}^{k}(s_{-i},t_{-i}))$, being the $(k-1)^{\textup{th}}$-order beliefs of $t_{-i}^{k}(s_{-i},t_{-i})$ exactly those of $\pi_{-i}(t_{-i})$. Finally, that $t_{i}^{k+1}\in T_{i}(d_{i}^{k+1})$ follows from the fact that $\tau_{i}(\pi_{i}(t_{i}^{k+1}))$ only assigns positive probability to pairs $(\theta_{0},t_{-i}^{k}(s_{-i},t_{-i}))\in \Theta_{0}\times T_{-i}(d_{-i}^{k})$ where $(s_{-i},t_{-i})$ is in the graph of $\textup{F}_{-i,k}^{0}(\Theta,\,\cdot\,)$.}

\vspace{0.3cm}

\noindent\textit{2. Part $(iii)$: Verification of convenient behavior}

\vspace{0.3cm}

\noindent It remains to be checked that:
\[
\textup{F}_{i,k+2}(d_{i}^{k+1},t_{i}^{k+1})\subseteq [\bar{s}_{i}|\Upsilon]_{k}.
\]
To see it fix strategy $s_{i}\in\textup{F}_{i,k+2}(d_{i}^{k+1},t_{i}^{k+1})$ and conjecture $\mu_{i}$ that justifies the inclusion of $s_{i}$ in $\textup{F}_{i,k+2}(d_{i}^{k+1},t_{i}^{k+1})$. We want to verify that $s_{i}\in [\bar{s}_{i}|\Upsilon]_{k}$. The process is slightly arduous, so let us first present a road map describing the six intermediate steps required:

\begin{itemize}

\item[S1.] For any $(s_{-i},\theta_{-i},\theta_{0})\in S_{-i}\times\Theta_{0}\times\Theta_{-i}$ we have that:
\begin{align*}
\mu_{i}(h^{0})[[s_{-i}|\Upsilon]_{k-1}\times\{(\theta_{0},\theta_{-i})\}&\times\Pi_{-i}(d_{-i}^{k})]=\\[1ex]
&=\bar{\mu}_{i}^{k+1}(h^{0})[[s_{-i}|\Upsilon]_{k-1}\times\{(\theta_{0},\theta_{-i})\}\times\Pi_{-i}(\Theta)].
\end{align*}
In general, $\mu_{i}(h^{0})$ and $\bar{\mu}_{i}^{k+1}(h^{0})$ may induce different marginals on $S_{-i}\times\Theta_{0}\times\Theta_{-i}$ because the types $t_{-i}'$ to which $t_{i}^{k+1}$ assigns positive probability do not necessarily have a unique strategy in $\textup{F}_{-i,k+1}(d_{-i}^{k},t_{-i}')$. However, the marginals both beliefs induce expect $i$'s opponents to behave analogously while each opponent $j\neq i$ finds herself in a history reachable by some strategy profile in $\textup{F}_{k-1}^{0}(\Upsilon,\,\cdot\,)$.

\item[S2.] Furthermore, for any $(s_{-i},\theta_{0},\theta_{-i})\in S_{-i}\times\Theta_{0}\times\Theta_{-i}$ and any $h\in\textup{H}_{i}(\textup{F}_{k}^{0},\Upsilon)$ we have that:
\begin{align*}
\mu_{i}(h)[[s_{-i}|\Upsilon]_{k-1}\times\{(\theta_{0},\theta_{-i})\}&\times\Pi_{-i}(d_{-i}^{k})]=\\[1ex]
&=\bar{\mu}_{i}^{k+1}(h)[[s_{-i}|\Upsilon]_{k-1}\times\{(\theta_{0},\theta_{-i})\}\times\Pi_{-i}(\Theta)].
\end{align*}
That is, despite both beliefs inducing possibly different marginals on $S_{-i}\times\Theta_{0}\times\Theta_{-i}$ at different stages of the game, the observation in S1 generalizes from the initial history to every history that is consistent players having chosen some strategy profile in $\textup{F}_{k-1}^{0}(\Upsilon,\,\cdot\,)$.

\item[S3.] For any pair $(t_{-i},s_{-i})$ in the graph of $\textup{F}_{-i,k}^{0}(\Upsilon,\,\cdot\,)$ and any strategy $s'_{-i}\in [s_{-i}|\Upsilon]_{k-1}$ we have the following outcome equivalence:
\[
z((s_{-i};\bar{s}_{i})|h^{0})=z((s_{-i}';\bar{s}_{i})|h^{0}).
\]

\item[S4.] For any pair $(t_{-i},s_{-i})$ in the graph of $\textup{F}_{-i,k}^{0}(\Upsilon,\,\cdot\,)$ and any strategy $s'_{-i}\in [s_{-i}|\Upsilon]_{k-1}$ we have the following outcome equivalence:
\[
z((s_{-i};s_{i})|h^{0})=z((s_{-i}';s_{i})|h^{0}).
\]

\item[S5.] For any strategy $s'_{i}\in\{\bar{s}_{i},s_{i}\}$ and any history $h\in\textup{H}_{i}(\textup{F}_{k}^{0},\Upsilon)$:
\[
U_{i}(s'_{i},\mu_{i}|\theta_{i}(t_{i}),h)=U_{i}(s'_{i},\bar{\mu}_{i}^{k+1}|\theta_{i}(t_{i}),h).
\]
At every history consistent players having chosen a strategy profile in $\textup{F}_{-i,k}(\Upsilon,\,\cdot\,)$ player $i$'s type $\theta_{i}(t_{i})$'s expected utilities under $\mu_{i}$ and $\bar{\mu}_{i}^{k+1}$ coincide either when playing $\bar{s}_{i}$ or $s_{i}$. Notice that this is a direct consequence of the outcome equivalences in S3 and S4, and of S2, where we saw that $\mu_{i}$ and $\bar{\mu}_{i}^{k+1}$ induce beliefs that expect analogous behavior while each player $j\neq i$ finds herself in a history she finds reachable when opponents play according to $\textup{F}_{-j,k-1}^{0}(\Upsilon,\,\cdot\,)$.

\item[S6.] For any history $h\in\textup{H}_{i}(\textup{F}_{k}^{0},\Upsilon)$, $s_{i}(h)=\bar{s}_{i}(h)$. This follows from S5 and the fact that $s_{i}$ is optimal given $\mu_{i}$ and $\bar{s}_{i}$ is optimal given $\bar{\mu}_{i}$, which in turn, induces an expected payoff similar enough to that of $\bar{\mu}_{i}^{k+1}$.

\end{itemize}
Once these six steps are verified part $(iii)$ the lemma will follow immediately. Let's proceed step by step:

\begin{itemize}


\item[S1.] First, for every $S_{-i}'\subseteq S_{-i}$ and $\theta_{-i}\in\Theta_{-i}$ denote:
\[
T_{-i}^{k}(S_{-i}',\theta_{-i}):=\bigcup_{s_{-i}\in S_{-i}'}\left\{t_{-i}^{k}(s_{-i},t_{-i})\left|t_{-i}\in T_{-i}(\Theta)\textup{  and  }\theta_{-i}(t_{-i}^{k}(s_{-i},t_{-i}))=\theta_{-i}\right.\right\}.
\]
Then, notice that for any triple $(s_{-i},\theta_{0},\theta_{-i})\in S_{-i}\times\Theta_{0}\times\Theta_{-i}$ we have that:
\begin{align*}
&\mu_{i}(h^{0})\left[[s_{-i}|\Upsilon]_{k-1}\times\{(\theta_{0},\theta_{-i})\}\times\Pi_{-i}(\Theta)\right]=\\[2ex]
&=\mu_{i}(h^{0})\left[[s_{-i}|\Upsilon]_{k-1}\times\{\theta_{0}\}\times T_{-i}^{k}([s_{-i}|\Upsilon]_{k-1},\theta_{-i})\right]\\[2ex]
&=\mu_{i}(h^{0})\left[S_{-i}\times\{\theta_{0}\}\times T_{-i}^{k}([s_{-i}|\Upsilon]_{k-1},\theta_{-i})\right]\\[2ex]
&=\bar{\mu}_{i}^{k+1}\left[S_{-i}\times\{\theta_{0}\}\times T_{-i}^{k}([s_{-i}|\Upsilon]_{k-1},\theta_{-i})\right]\\[2ex]
&=\sum_{s_{-i}''\in S_{-i}}\bar{\mu}_{i}^{k+1}\left[\{(s_{-i}'',,\theta_{0})\}\times T_{-i}^{k}([s_{-i}|\Upsilon]_{k-1},\theta_{-i})\right]\\[2ex]
&=\sum_{s_{-i}''\in S_{-i}}\bar{\mu}_{i}(h^{0})\left[\{(s_{-i}'',\theta_{0})\}\times\left\{t_{-i}\in T_{-i}(d_{-i}^{k})\left|t_{-i}^{k}(s_{-i}'',t_{-i})\in T_{-i}^{k}([s_{-i}|\Upsilon]_{k-1},\theta_{-i})\right.\right\}\right]\\[2ex]
&=\sum_{s_{-i}''\in [s_{-i}|\Theta]_{k-1}}\bar{\mu}_{i}(h^{0})\left[\{(s_{-i}'',\theta_{0})\}\times\left\{t_{-i}\in T_{-i}(d_{-i}^{k})\left|t_{-i}^{k}(s_{-i}'',t_{-i})\in T_{-i}^{k}([s_{-i}|\Upsilon]_{k-1},\theta_{-i})\right.\right\}\right]\\[2ex]
&=\sum_{s_{-i}''\in [s_{-i}|\Upsilon]_{k-1}}\bar{\mu}_{i}(h^{0})\left[\{(s_{-i}'',\theta_{0},\theta_{-i})\}\times\Pi_{-i}(d_{-i}^{k})\right]\\[2ex]
&=\bar{\mu}_{i}(h^{0})\left[[s_{-i}|\Upsilon]_{k-1}\times\{(\theta_{0},\theta_{-i})\}\times\Pi_{-i}(d_{-i}^{k})\right].
\end{align*}


\item[S2.] Fix history $h\in\textup{H}_{i}(\textup{F}_{k}^{0},\Upsilon)$ and $s_{-i}\in S_{-i}(h)$, and pick some $s_{-i}'\in [s_{-i}|\Upsilon]_{k-1}$. It is easy to see that $s_{-i}'\in S_{-i}(h)$: for any player $j\neq i$ and any history $h'\in H_{j}\cup\{h^{0}\}$ that precedes or coincides with $h$ we have that $h'\in\textup{H}_{j}(\textup{F}_{k}^{0},\Upsilon)$ and thus, it holds that $s'_{j}(h')=s_{j}(h')$. Hence, $[s_{-i}|\Upsilon]_{k-1}\subseteq S_{-i}(h)$ and, in consequence, there must exist some $s_{-i}^{1},\dots, s_{-i}^{M}\in S_{-i}(h)$ such that the family $\{[s_{-i}^{1}|\Upsilon]_{k-1},\dots,[s_{-i}^{M}|\Upsilon]_{k-1}\}$ is a partition of $S_{-i}(h)$. Then, it follows from S1 that:
\begin{align*}
\mu_{i}(h^{0})[S_{-i}(h)\times\Theta_{0}\times T_{-i}(d_{-i}^{k})]&=\sum_{m=1}^{M}\mu_{i}(h^{0})[[s_{-i}^{m}|\Upsilon]_{k-1}\times\Theta_{0}\times T_{-i}(d_{-i}^{k})]\\
&=\sum_{m=1}^{M}\bar{\mu}_{i}^{k+1}(h^{0})[[s_{-i}^{m}|\Upsilon]_{k-1}\times\Theta_{0}\times T_{-i}(d_{-i}^{k})]\\
&=\bar{\mu}_{i}^{k+1}(h^{0})[S_{-i}(h)\times\Theta_{0}\times T_{-i}(d_{-i}^{k})].
\end{align*}
Remember now that $\bar{\mu}_{i}^{k+1}(h^{0})$ puts positive probability on $S_{-i}(h)\times\Theta_{0}\times T_{-i}(d_{-i}^{k})$ and thus, it follows that for any $(s_{-i},\theta_{0},\theta_{-i})$:
\begin{align*}
\bar{\mu}_{i}^{k+1}(h)[[s_{-i}|\Upsilon]_{k-1}\times&\{(\theta_{0},\theta_{-i})\}\times \Pi_{-i}(d_{-i}^{k})]=\\[1ex]
&=\frac{\bar{\mu}_{i}^{k+1}(h^{0})[(S_{-i}(h)\cap [s_{-i}|\Upsilon]_{k-1})\times\{(\theta_{0},\theta_{-i})\}\times\Pi_{-i}(d_{-i}^{k})]}{\bar{\mu}_{i}^{k+1}(h^{0})[S_{-i}(h)\times\Theta_{0}\times T_{-i}(d_{-i}^{k})]}\\[1ex]
&=\frac{\bar{\mu}_{i}(h^{0})[(S_{-i}(h)\cap [s_{-i}|\Upsilon]_{k-1})\times\{(\theta_{0},\theta_{-i})\}\times\Pi_{-i}(\Theta)]}{\bar{\mu}_{i}(h^{0})[S_{-i}(h)\times\Theta_{0}\times T_{-i}(\Theta)]}\\[1ex]
&=\bar{\mu}_{i}(h)[[s_{-i}|\Upsilon]_{k-1}\times\{(\theta_{0},\theta_{-i})\}\times\Pi_{-i}(\Theta)].
\end{align*}


\item[S3.] Fix pair $(s_{-i},t_{-i})$ in the graph of $\textup{F}_{-i,k}^{0}(\Upsilon,\,\cdot\,)$ and $s'_{-i}\in [s_{-i}|\Upsilon]_{k-1}$. Next, pick player $j\neq i$ such that there exists some history $h\in H_{j}$ that precedes $z((s_{-i};\bar{s}_{i})|h^{0})$ and is reached by both $s_{j}$ and $s_{j}'$.\footnote{\label{footnote2} I.e., such that $h\prec z((s_{-i};\bar{s}_{i})|h^{0})$ and $h\in H_{j}(s_{j})\cap H_{j}(s_{j}')$. Obviously, if there exists some $h$ satisfying the first requirement, there exists some history that satisfies both. We do not have to worry about players $j$ that lack such history $h$: they do not make any choice along the path leading to $z((s_{-i};\bar{s}_{i})|h^{0})$.} To see that $h\in\textup{H}_{i}(\textup{F}_{j,k-1}^{0},\Upsilon)$ simply notice that: $(a)$ $s_{\ell}\in\textup{F}_{\ell,k}^{0}(\Upsilon,t_{\ell})$ for every player $\ell\neq i,j$ and $(b)$ $\bar{s}_{i}\in\textup{F}_{i,k+1}^{0}(\Upsilon,t_{i})$. Thus, we have that:
\[
h\in\textup{H}_{j}(\textup{F}_{i,k+1}^{0},\Upsilon)\cap\bigcap_{\ell\neq i,j}\textup{H}_{j}(\textup{F}_{\ell,k}^{0},\Upsilon)\subseteq\bigcap_{\ell\neq j}\textup{H}_{j}(\textup{F}_{\ell,k-1}^{0},\Upsilon)=\textup{H}_{j}(\textup{F}_{-j,k-1}^{0},\Upsilon)
\]
Then, we know by definition of $[s_{j}|\Upsilon]_{k-1}$, that $s_{j}(h)=s_{j}'(h)$. This lets us conclude that $z((s_{-i};\bar{s}_{i})|h^{0})=z((s_{-i}';\bar{s}_{i})|h^{0})$ for every $s_{-i}'\in [s_{-i}|\Upsilon]_{k-1}$.


\item[S4.] Fix pair $(s_{-i},t_{-i})$ in the graph of $\textup{F}_{-i,k}^{0}(\Upsilon,\,\cdot\,)$ and $s'_{-i}\in [s_{-i}|\Upsilon]_{k-1}$. Next, pick player $j\neq i$ such that there exists some history $h\in H_{j}$ that precedes $z((s_{-i};s_{i})|h^{0})$ and is reached by both $s_{j}$ and $s_{j}'$.\footnote{See Footnote \ref{footnote2}; the same logic applies here, mutatis mutandi.} We verify first that $h\in\textup{H}_{j}(\textup{F}_{-j,k-1}^{0},\Upsilon)$. To begin, remember that we have that, for player $\ell\neq i,j$, $s_{\ell}\in\textup{F}_{\ell,k}^{0}(\Upsilon,t_{\ell})$. Thus:
\[
h\in\bigcap_{\ell\neq i,j}\textup{H}_{j}(\textup{F}_{\ell,k},\Upsilon).\subseteq\bigcap_{\ell\neq i,j}\textup{H}_{j}(\textup{F}_{\ell,k-1},\Upsilon).
\]
Next, since $s_{i}\in\textup{F}_{i,k+2}(d_{i}^{k+1},t_{i}^{k+1})$, in particular, we know that $s_{i}\in\textup{F}_{i,k+1}(d_{i}^{k+1},t_{i}^{k+1})$. In consequence, since $d_{i,k+1}^{k+1}=d_{i,k+1}$ and $\textup{F}_{i,k+1}$ only depends on the $(k+1)^{\textup{th}}$-order specification of the directory and the type, there must exist some finite type $t_{i}'\in T_{i}(\Theta)$ such that $s_{i}\in\textup{F}_{i,k+1}(\Upsilon,t_{i}')$.\footnote{The possibility of finding \textit{finite} $t_{i}$ follows from $t_{i}^{k+1}$ itself being finite.} Thus, we have that:
\[
h\in\textup{H}_{j}(\textup{F}_{i,k-1},d_{i}^{k+1})\subseteq\textup{H}_{j}(\textup{F}_{i,k-1},\Upsilon)\subseteq\textup{H}_{j}(\textup{F}_{i,k-1}^{0},\Upsilon)
\]
the last inclusion being one of the conditions on $\Upsilon$ in the statement of the lemma. Hence we conclude that:
\[
h\in\bigcap_{\ell\neq j}\textup{H}_{j}(\textup{F}_{\ell,k-1}^{0},\Upsilon)=\textup{H}_{j}(\textup{F}_{-j,k-1}^{0},\Upsilon),
\]
and thus, by definition of $[s_{j}|\Upsilon]_{k-1}$, that $s_{j}(h)=s_{j}'(h)$. Obviously, the latter implies that $z((s_{-i};s_{i})|h^{0})=z((s_{-i}';s_{i})|h^{0})$ for every $s_{-i}'\in [s_{-i}|\Upsilon]_{k-1}$.


\item[S5.] Fix history $h\in\textup{H}_{i}(\textup{F}_{k}^{0},\Upsilon)$. Next pick strategies $s_{-i}^{1},\dots, s_{-i}^{M}\in S_{-i}(h)$ such that family $\{[s_{-i}^{1}|\Upsilon]_{k-1},\dots [s_{-i}^{M}|\Upsilon]_{k-1}\}$ is a partition of $S_{-i}(h)$. Based on the latter, set $\widetilde{S_{i}}(h):=\{[s_{-i}^{1}|\Upsilon]_{k-1},\dots,[s_{-i}^{M}|\Upsilon]_{k-1}\}$ and define the two measures on $\widetilde{S_{i}}(h)\times\Theta_{0}\times\Theta_{-i}$ induced by setting, for any $(\tilde{s}_{-i},\theta_{0},\theta_{-i})=([s_{-i}|\Upsilon]_{k-1},\theta_{0},\theta_{-i})\in\widetilde{S}_{-i}(h)\times\Theta_{0}\times\Theta_{-i}$,
\begin{align*}
\widetilde{\mu}_{i}^{k+1}[(\tilde{s}_{-i},\theta_{0},\theta_{-i})]&:=(\textup{marg}_{S_{-i}\times\Theta_{0}\times\Theta_{-i}}\bar{\mu}_{i}^{k+1}(h))[\tilde{s}_{-i}\times\{\theta_{0},\theta_{-i}\}],\\[1ex]
\widetilde{\mu}_{i}[(\tilde{s}_{-i},\theta_{0},\theta_{-i})]&:=(\textup{marg}_{S_{-i}\times\Theta_{0}\times\Theta_{-i}}\mu_{i}(h))[\tilde{s}_{-i}\times\{(\theta_{0},\theta_{-i})\}].
\end{align*}
Notice that we know from S2 that $\widetilde{\mu}_{i}=\widetilde{\mu}_{i}^{k+1}$. It follows then that for any strategy $s'_{i}\in\{\bar{s}_{i},s_{i}\}$:
\begin{align*}
U_{i}&(\mu_{i},s'_{i}|\theta_{i}(t_{i}),h)=\\[1ex]
&=\int_{S_{-i}(h)\times\Theta_{0}\times\Theta_{-i}}u_{i}(z(s_{-i};s'_{i}),((\theta_{0},\theta_{i});\theta_{i}(t_{i})))\textup{d}(\textup{marg}_{S_{-i}\times\Theta_{0}\times\Theta_{-i}}\mu_{i}(h))\\
&=\int_{\widetilde{S}_{-i}(h)\times\Theta_{0}\times\Theta_{-i}}u_{i}(z(\tilde{s}_{-i};s'_{i}),((\theta_{0},\theta_{-i});\theta_{i}(t_{i})))\textup{d}\widetilde{\mu}_{i}\\
&=\int_{\widetilde{S}_{-i}(h)\times\Theta_{0}\times\Theta_{-i}}u_{i}(z(\tilde{s}_{-i};s'_{i}),((\theta_{0},\theta_{-i});\theta_{i}(t_{i})))\textup{d}\widetilde{\mu}_{i}^{k+1}\\
&=\int_{S_{-i}(h)\times\Theta_{0}\times\Theta_{-i}}u_{i}(z(s_{-i};s'_{i}),((\theta_{0},\theta_{-i});\theta_{i}(t_{i})))\textup{d}(\textup{marg}_{S_{-i}\times\Theta_{0}\times\Theta_{-i}}\bar{\mu}_{i}^{k+1}(h))\\[1ex]
&=U_{i}(\bar{\mu}_{i}^{k+1},s'_{i}|\theta_{i}(t_{i}),h),
\end{align*}
being the second and fourth equalities consequences of the outcome equivalences in S4 and S5 (the abuse of notation $z(\tilde{s}_{-i};s_{i}')=z(s_{-i};s_{i}')$ is is innocuous due to S4).


\item[S6.] Proceed by contradiction and suppose that there exists some history $h\in\textup{H}_{i}(\textup{F}_{k}^{0},\Upsilon)$ such that $s_{i}(h)\neq\bar{s}_{i}(h)$. Pick then strategy $\hat{s}_{i}$ that maximizes $U_{i}(\,\cdot\,,\bar{\mu}_{i}|\theta_{i}(t_{i}),h')$ at every history $h'\succ h$ and define new strategy $s_{i}^{0}$ as follows:
\[
s_{i}^{0}(h'):=\left\{
\begin{tabular}{l l}
$\hat{s}_{i}(h')$&$\textup{if  }h'\succ h$,\\
$s_{i}(h')$&$h'=h$,\\
$\bar{s}_{i}(h')$&$\textup{otherwise.}$
\end{tabular}
\right.
\]
Let's check next that $s_{i}^{0}\in r_{i}(\theta_{i}(t_{i}),\bar{\mu}_{i})$. To do it conceive $H_{i}$ as the disjoint union of the following four components:
\[
\{h'\in H_{i}\cup\{h^{0}\}|h'\succ h\}\cup\{h\}\cup\{h'\in H_{i}\cup\{h^{0}\}|h'\prec h\}\cup\{h'\in H_{i}\cup\{h^{0}\}|h\npreceq h'\textup{  and  }h'\npreceq h\}.
\]
We distinguish then four cases:
\begin{itemize}

\item[$\bullet$] $h'\succ h$. The construction of $s_{i}^{0}$ ensures that $U_{i}(s_{i}^{0},\bar{\mu}_{i}|\theta_{i}(t_{i}),h')=U_{i}(\hat{s}_{i},\bar{\mu}_{i}|\theta_{i}(t_{i}),h')$; thus $s_{i}^{0}$ must be a maximizer of $U_{i}(\,\cdot\,,\bar{\mu}_{i}|\theta_{i}(t_{i}),h')$.

\item[$\bullet$] $h'=h$. The above, together with $s_{i}^{0}(h')=s_{i}(h')$, implies that:\footnote{For the first inequality notice that $s_{i}^{0}$ and $s_{i}$ are identical at $h'=h$, and that $s_{i}^{0}$ is a maximizer of $U_{i}(\,\cdot\,,\bar{\mu}_{i}|\theta_{i}(t_{i}),h')$ for every $h''\succ h'$. The equality follows from \eqref{equation:equality} above.}
\[
U_{i}(s_{i}^{0},\bar{\mu}_{i}|\theta_{i}(t_{i}),h')\geq U_{i}(s_{i},\bar{\mu}_{i}|\theta_{i}(t_{i}),h')\\[1ex]
=U_{i}(s_{i},\bar{\mu}_{i}^{k+1}|\theta_{i}(t_{i}),h').
\]
Now, we also have that:\footnote{The first two equalities follow from S5 and the last, from equation \eqref{equation:equality} above. The inequality is a consequence of $s_{i}$ being a best response to $\mu_{i}$ given $\theta_{i}(t_{i})$.}
\begin{align*}
U_{i}(s_{i},\bar{\mu}_{i}^{k+1}|\theta_{i}(t_{i}),h')&=U_{i}(s_{i},\mu_{i}|\theta_{i}(t_{i}),h')\\
&\geq U_{i}(\bar{s}_{i},\mu_{i}|\theta_{i}(t_{i}),h')\\
&=U_{i}(\bar{s}_{i},\bar{\mu}_{i}^{k+1}|\theta_{i}(t_{i}),h')\\
&=U_{i}(\bar{s}_{i},\bar{\mu}_{i}|\theta_{i}(t_{i}),h'),
\end{align*}
and remember that $\bar{s}_{i}$ is a maximizer of $U_{i}(\,\cdot\,,\bar{\mu}_{i}|\theta_{i}(t_{i}),h')$. Thus, it must necessarily hold that $U_{i}(s_{i}^{0},\bar{\mu}_{i}|\theta_{i}(t_{i}),h')=U_{i}(\bar{s}_{i},\bar{\mu}_{i}|\theta_{i}(t_{i}),h')$. Hence, we conclude that $s_{i}^{0}$ is a maximizer of $U_{i}(\,\cdot\,,\bar{\mu}_{i}|\theta_{i}(t_{i}),h')$.

\item[$\bullet$] $h'\prec h$. Notice that the construction of $s_{i}^{0}$ on the one hand, and the facts that $U_{i}(s_{i}^{0},\bar{\mu}_{i}|\theta_{i}(t_{i}),h)=U_{i}(\bar{s}_{i},\bar{\mu}_{i}|\theta_{i}(t_{i}),h)$ and $s_{i}^{0}$ and $\bar{s}_{i}$ are equivalent at any history preceding $h$, on the other, imply that $U_{i}(s_{i}^{0},\bar{\mu}_{i}|\theta_{i}(t_{i}),h')=U_{i}(\bar{s}_{i},\bar{\mu}_{i}|\theta_{i}(t_{i}),h')$. Thus, since $\bar{s}_{i}$ maximizes $U_{i}(\,\cdot\,,\bar{\mu}_{i}|\theta_{i}(t_{i}),h')$ it follows that $s_{i}^{0}$ must maximize $U_{i}(\,\cdot\,,\bar{\mu}_{i}|\theta_{i}(t_{i}),h')$ as well.

\item[$\bullet$] For any other $h'$, clearly, we have that:
\[
U_{i}(s_{i}^{0},\bar{\mu}_{i}|\theta_{i}(t_{i}),h')= U_{i}(\bar{s}_{i},\bar{\mu}_{i}|\theta_{i}(t_{i}),h').
\]

\end{itemize}
We have then reached the following contradictory conclusions: $s_{i}^{0}\notin[\bar{s}_{i}|\Upsilon]_{k}$ and $s_{i}^{0}\in r_{i}(\theta_{i}(t_{i}),\bar{\mu}_{i})=[\bar{s}_{i}|\Upsilon]_{k}$.

\end{itemize}
Hence, if $s_{i}$ is a best response to $\mu_{i}$ for $\theta_{i}(t_{i})$ it must hold that $s_{i}(h)=\bar{s}_{i}(h)$ for every $h\in\textup{H}_{i}(\textup{F}_{k}^{0},\Upsilon)$, that is,
\[
\textup{F}_{i,k+2}(d_{i}^{k+1},t_{i}^{k+1})\subseteq [\bar{s}_{i}|\Upsilon]_{k},
\]
and in consequence, the proof is finally complete.
\end{proof}


\subsection{Proofs of the auxiliary results}
\label{subsection:Bauxiliary}


\subsubsection{Lemma \ref{lemma:basiclemma}: Availability of finite conjectures}
\label{subsubsection:fintieconjectures}

\begin{lemma}
\label{lemma:basiclemma}
Let $\Gamma$ be an extensive form and $\Upsilon$, a standard payoff structure. Then for any $k\geq 0$, any player $i$, any finite type $t_{i}\in T_{i}(\Theta)$ and every strategy $s_{i}\in\textup{F}_{i,k}(\Upsilon,t_{i})$ there exists a finite conjecture $\mu_{i}$ that justifies the inclusion of $s_{i}$ in $\textup{F}_{i,k}(\Upsilon,t_{i})$.
\end{lemma}
\begin{proof}
Fix $k\geq 0$, player $i$, finite type $T_{i}(\Theta)$ and strategy $s_{i}\in\textup{F}_{i,k+1}(\Upsilon,t_{i})$ and pick conjecture $\mu_{i}$ that justifies the inclusion of $s_{i}$ in $\textup{F}_{i,k+1}(\Upsilon,t_{i})$. For convenience, let us denote $X_{i}:=\Theta_{0}\times T_{-i}(\Theta)$. Since $X_{i}$ is separable, we ca pick 
countable $(x_{i}^{m})_{m\in\mathds{N}}$ where $\{x_{i}^{m}\}_{m\in\mathds{N}}$ is dense in $X_{i}$. 

Now, for each $m,n\in \mathds{N}$ let $B(x_{i}^{m},1/n)$ denote the ball of radius $1/n$ around $x_{i}^{m}$. We know because of the upper hemicontinuity of $\textup{F}_{-i,\ell}(\Upsilon,\,\cdot\,)$ for each $\ell=0,1,\dots,k$ that for each $n,m\in \mathds{N}$ there exists some open set $W_{i}^{\ell,n,m}\subseteq B(x_{i}^{m}, 1/n)$ such that $\textup{F}_{-i,\ell}(\Upsilon,t_{-i})\subseteq \textup{F}_{-i,\ell}(\Upsilon,t_{-i}^{m})$ for every $t_{-i}$ is in the projection of $W_{i}^{\ell,n,m}$ on $T_{-i}(\Theta)$, $t_{-i}^{m}$ being the projection on $T_{-i}(\Theta)$ of $x_{i}^{m}$. Furthermore, for each $\ell=0,1,\dots,k$ and $n \in \mathds{N}$, the family $\{W_{i}^{\ell,n,m}\}_{m\in\mathds{N}}$ is an open cover of $X_{i}$, which, due to $X_{i}$ being compact, we can assume as finite: $\{W_{i}^{\ell,n,m}\}_{m=1}^{M_{\ell,n}}$. 

Now, for each $\ell=0,1,\dots,k$ and $n \in \mathds{N}$ set:
\[
V_{i}^{\ell,n,m}:=W_{i}^{\ell,n,m}\setminus
\bigcup_{r=1}^{m-1}W_{i}^{\ell,n,r}.
\]
Notice that for each $\ell=0,1,\dots,k$ and $n \in \mathds{N}$ family $\{V_{i}^{\ell,n,m}\}_{m=1}^{M_{\ell,n}}$ is a partition of $X_{i}$ consisting of measurable sets, contained in a ball of radius $1/n$. Since the set of finite types is dense in $X_{i}$ for each $\ell=0,1,\dots,k$ and $n \in \mathds{N}$,
there exists some list $(y_{i}^{\ell,n,m})_{m=1}^{M_{\ell,n}}$ such that, for each $m=1,\dots,M_{\ell,n}$, the projection on $T_{-i}(\Theta)$ of $y_{i}^{\ell,n,m}$ is finite and $y_{i}^{\ell,n,m}\in V_{i}^{\ell,n,m}$.

We turn now back to $\mu_{i}$. Let us denote $\textup{H}_{i,k}(\mu_{i}):=\textup{H}_{i}(\mu_{i})\cap\textup{H}_{i}(\textup{F}_{-i,k},\Theta)$ and then, define recursive, for each $\ell=0,1,\dots,k-1$, set:
\[
\textup{H}_{i,\ell}(\mu_{i}):=\left(\textup{H}_{i}(\mu_{i})\cap\textup{H}_{i}(\textup{F}_{-i,\ell},\Upsilon)\right)\setminus\textup{H}_{i,k}(\mu_{i}).
\] 
Next, we will construct a conditional probability system $\mu_{i}^{n}$ for each $n\in\mathds{N}$. First, for each each $\ell=0,1,\dots,k$ and each $h\in\textup{H}_{i,\ell}(\mu_{i})\setminus\{h^{0}\}$ set: 
\[
\mu_{i}^{n}(h)[(s_{-i},y_{i}^{\ell,m})]:=\mu_{i}(h)[\{s_{-i}\}\times V_{i}^{\ell,n,m}],
\]
for every $s_{-i}\in S_{-i}$ and every $m=1,\dots,M_{\ell,n}$. Second, set $\mu_{i}^{n}(h^{0}):=\mu_{i}(h^{0})$. Finally, for each $h\notin\textup{H}_{i}(\mu_{i})$ define $\mu_{i}^{n}(h)$ using the chain rule. Notice that the marginals on $S_{-i}$ of $\mu_{i}(h)$ and each $\mu_{i}^{n}(h)$ coincide for every history $h$, and this guarantees that $\mu_{i}^{n}$ is (or, has been) well-defined.

Now, notice also that, for each $\ell=0,1,\dots,k$, $\mu_{i}^{n}(h)$ assigns probability one to the graph of $\textup{F}_{-i,\ell}(\Upsilon,\,\cdot\,)$. Obviously, every $\mu_{i}^{m}$ is consistent with type $t_{i}$, and sequence $(\mu_{i}^{n})_{n\in\mathds{N}}$ converges to $\mu_{i}$. Thus, the upper hemicontinuity of $r_{i}(\theta_{i}(t_{i}),\,\cdot\,)$ ensures the existence of some $N\in\mathds{N}$ such that $s_{i}\in r_{i}(\theta_{i}(t_{i}),\mu_{i}^{n})$ for every $n\geq N$. Hence, every $\mu_{i}^{n}$ where $n\geq N$ is a \textit{finite} conjecture that justifies the inclusion of $s_{i}$ in $\textup{F}_{i,k+1}(\Upsilon,t_{i})$
\end{proof}


\subsubsection{Lemma \ref{lemma:firstperturbation1}: A partial upper hemicontinuity result}
\label{subsubsection:partialUHC}

\begin{lemma}
\label{lemma:firstperturbation1}
Let $\Gamma$ be an extensive form and $\Upsilon$, a standard payoff structure. Then for any $k\geq 0$, any player $i$ there exists some $n_{k}^{i}\in\mathds{N}$ such that the following three hold:
\begin{itemize}

\item[$(i)$] For every type $t_{i}\in T_{i}(\Theta)$, $\textup{F}_{i,k}(\Upsilon,t_{i})\subseteq\textup{F}_{i,k}(\Upsilon^{n},t_{i})\subseteq\textup{F}_{i,k}(\Upsilon^{n+\ell},t_{i})$ for every $n\geq n_{k}^{i}$ and every $\ell\in\mathds{N}$.

\item[$(ii)$] For every type $t_{i}\in T_{i}(\Theta)$, $\bigcap_{n\geq n_{k}^{i}}\textup{F}_{i,k}(\Upsilon^{n},t_{i})\subseteq\textup{F}_{i,k}(\Upsilon,t_{i})$.

\item[$(iii)$] For every player $j\neq i$, $\textup{H}_{i}(\textup{F}_{j,k},\Upsilon^{n})=\textup{H}_{i}(\textup{F}_{j,k},\Upsilon)$ for every $n\geq n_{k}^{i}$.
\end{itemize}
\end{lemma}
\begin{proof}
We proceed by induction on $k$. The claims hold trivially for the initial case ($k=0$) so we can focus on the proof of the inductive step. Suppose that the claims hold for $k$; we verify now that they also hold for $k+1$:

\vspace{0.3cm}

\noindent\textsc{Claim $(i)$.} Fix player $i$ and set $n_{k+1}^{i,1}:=\textup{max}\{n_{\ell}^{j}|\ell=0,\dots,k\textup{  and  }j\in I\}$. We know from part $(i)$ of the induction hypothesis that for every $n\geq n_{k+1}^{i,1}$ and every $\ell\in\mathds{N}$,
\[
\textup{F}_{i,k}(\Upsilon,t_{i})\subseteq\textup{F}_{i,k}(\Upsilon^{n},t_{i})\subseteq\textup{F}_{i,k}(\Upsilon^{n+\ell},t_{i}),
\]
for every type $t_{i}\in T_{i}(\Theta)$ and,
\[
\textup{Graph}\left(\textup{F}_{-i,k}(\Upsilon,\,\cdot\,)\right)\subseteq\textup{Graph}\left(\textup{F}_{-i,k}(\Upsilon^{n},\,\cdot\,)\right)\subseteq\textup{Graph}\left(\textup{F}_{-i,k}(\Upsilon^{n+\ell},\,\cdot\,)\right).
\]
We also know from part $(iii)$ that,
\[
\textup{H}_{i}(\textup{F}_{-i,k},\Upsilon^{n+\ell})=\textup{H}_{i}(\textup{F}_{-i,k},\Upsilon^{n})=\textup{H}_{i}(\textup{F}_{-i,k},\Upsilon).
\]
Fix now $n\geq n_{k+1}^{i,1}$, $\ell\in\mathds{N}$, type $t_{i}\in T_{i}(\Theta)$, strategy $s_{i}\in\textup{F}_{i,k+1}(\Upsilon,t_{i})$ and conjecture $\mu_{i}$ that justifies the inclusion of $s_{i}$ in $\textup{F}_{i,k+1}(\Upsilon,t_{i})$. It follows from the above that $s_{i}\in\textup{F}_{i,k}(\Upsilon^{n},t_{i})$ and that $\mu_{i}$ is a conjecture that justifies the inclusion of $s_{i}$ in $\textup{F}_{i,k+1}(\Upsilon^{n},t_{i})$. Fix now $s_{i}^{n}\in\textup{F}_{i,k+1}(\Upsilon^{n},t_{i})$ and conjecture $\mu_{i}^{n}$ that satisfies the inclusion of $s_{i}^{n}$ in $\textup{F}_{i,k+1}(\Upsilon^{n},t_{i})$. Again, it clearly follows from the above that $s_{i}\in\textup{F}_{i,k}(\Upsilon^{n+\ell},t_{i})$ and that $\mu_{i}^{n}$ is a conjecture that satisfies the inclusion of $s_{i}^{n}$ in $\textup{F}_{i,k+1}(\Upsilon^{n+\ell},t_{i})$ (see the inclusion property in comment (B) in Section \ref{subsubsection:perturbation1}).
\hfill$\bigstar$

\vspace{0.3cm}

\noindent\textsc{Claim $(ii)$.} Fix player $i$ and set again $n_{k+1}^{i,1}:=\textup{max}\{n_{k}^{j}|\ell=0,\dots,k\textup{  and  }j\in I\}$. Fix type $t_{i}\in T_{i}(\Theta)$ and strategy $s_{i}\in\bigcap_{n\geq n_{k+1}^{i,2}}\textup{F}_{i,k+1}(\Upsilon^{n},t_{i})$, and take convergent sequence of conjectures $(\mu_{i}^{m})_{m\geq n_{k+1}^{i,2}}$ such that, for each $m\geq n_{k+1}^{i,2}$, $\mu_{i}^{m}$ justifies the inclusion of $s_{i}$ in $\textup{F}_{i,k+1}(\Upsilon^{n_{m}},t_{i})$ (being $n_{m}\geq m$). Let $\mu_{i}$ denote the limit of this sequence. Notice that we know from part $(ii)$ of the induction hypothesis that:
\[
s_{i}\in\bigcap_{n\geq n_{k+1}^{i,2}}\textup{F}_{i,k}(\Upsilon^{n},t_{i})\subseteq\textup{F}_{i,k}(\Upsilon,t_{i}).
\]
Now, set $\ell=0,\dots, k$ and pick history $h\in\textup{H}_{i}(\textup{F}_{-i,\ell},\Upsilon)$. Notice that:\footnote{The first implication follows from parts $(i)$ and $(iii)$ of the induction hypothesis. The second is a consequence of the graph of $\textup{F}_{-i,\ell}(\Upsilon^{n_{m}},\,\cdot\,)$ being closed. The third and the fourth inclusions are obvious and the last one follows from part $(ii)$ the induction hypothesis.}
\begin{align*}
\forall m\geq n_{k+1}^{i,2},\,\mu_{i}^{m}&(h)[\Theta_{0}^{n_{m}}\times\textup{Graph}\left(\textup{F}_{-i,\ell}\left(\Upsilon^{n_{m}},\,\cdot\,\right)\right)]=1\Longrightarrow\\[2ex]
\Longrightarrow&\,\forall m\geq n_{k+1}^{i,2},\forall r\geq 0,\,\mu_{i}^{m+r}(h)\left[\Theta_{0}^{n_{m}}\times\textup{Graph}\left(\textup{F}_{-i,\ell}\left(\Upsilon^{n_{m}},\,\cdot\,\right)\right)\right]=1\\[1.5ex]
\Longrightarrow&\,\forall m\geq n_{k+1}^{i,2},\,\underset{r\rightarrow\infty}{\textup{lim}}\mu_{i}^{m+r}(h)\left[\Theta_{0}^{n_{m}}\times\textup{Graph}\left(\textup{F}_{-i,\ell}\left(\Upsilon^{n_{m}},\,\cdot\,\right)\right)\right]=1\\[1.5ex]
\Longrightarrow&\,\forall m\geq n_{k+1}^{i,2},\,\mu_{i}(h)\left[\Theta_{0}\times\textup{Graph}\left(\textup{F}_{-i,\ell}\left(\Upsilon^{n_{m}},\,\cdot\,\right)\right)\right]=1\\[1.5ex]
\Longrightarrow&\,\mu_{i}(h)\left[\Theta_{0}\times\bigcap_{m\geq n_{k+1}^{i,2}}\textup{Graph}\left(\textup{F}_{-i,\ell}\left(\Upsilon^{n_{m}},\,\cdot\,\right)\right)\right]=1\\[1.5ex]
\Longrightarrow&\,\mu_{i}(h)\left[\Theta_{0}\times\textup{Graph}\left(\textup{F}_{-i,\ell}\left(\Upsilon,\,\cdot\,\right)\right)\right]=1.
\end{align*}
Given the above, upper hemicontinuity of the best response operator ensures that $\mu_{i}$ is a conjecture that justifies the inclusion of $s_{i}$ in $\textup{F}_{i,k+1}(\Upsilon,t_{i})$.
\hfill$\bigstar$

\vspace{0.3cm}

\noindent\textsc{Claim $(iii)$.} Fix players $i$ and $j\neq i$. It follow from part $(i)$ that $\textup{H}_{i}(\textup{F}_{j,k+1},\Upsilon)\subseteq\textup{H}_{i}(\textup{F}_{j,k+1},\Upsilon^{n})$ for every $n\geq n_{k+1}^{j,1}$ (the latter has been defined in the proofs of the previous two claims). To prove the reverse inclusion simply notice that, since $H_{i}\cup\{h^{0}\}$ is finite, it follows from part $(ii)$ that there must exist some $n_{k+1}^{j,*}\in\mathds{N}$ such that:\footnote{To better see it suppose by contradiction that for any $m\in\mathds{N}$ there exists some $n_{m}\geq n$ such that $h\in\textup{H}_{i}(\textup{F}_{j,k+1},\Upsilon^{n_{m}})$. Then, we know from part $(i)$ that there exists some $\bar{m}$ such that $h\in\textup{H}_{i}(\textup{F}_{j,k+1},\Upsilon^{m})$ for every $m\geq\bar{m}$ and thus, it follows from part $(ii)$ that $h\in\textup{H}_{i}(\textup{F}_{j,k+1},\Upsilon).$}
\[
h\notin\textup{H}_{i}(\textup{F}_{j,k+1},\Upsilon)\Longrightarrow h\notin \textup{H}_{i}(\textup{F}_{j,k+1},\Upsilon^{n})
\]
for every $n\geq n_{k+1}^{j,*}$. Hence, it follows that $\textup{H}_{i}(\textup{F}_{j,k+1},\Upsilon^{n})=\textup{H}_{i}(\textup{F}_{j,k+1},\Upsilon)$ for every $n\geq n_{k+1}^{2}:=\textup{max}\{\textup{max}\{n_{k+1}^{j,1},n_{k+1}^{j,*}\}|j\in I\}$.
\hfill$\bigstar$

\vspace{0.3cm}

Thus, to finish the proof simply set $n_{k+1}^{i}:=\textup{max}\{n_{k+1}^{i,1},n_{k+1}^{i,2}\}$.
\end{proof}


\bibliographystyle{ifac}

\bibliography{bibliography}

\begin{thebibliography}{66}
\providecommand{\natexlab}[1]{#1}

\bibitem[\protect\citeauthoryear{Aliprantis and Border}{2007}]{aliprantis-07}
Aliprantis, Charambolos~D. and Kim~C. Border (2007).
\newblock \emph{\href{https://www.springer.com/gp/book/9783540295860}{Infinite
  dimensional analysis: A hitchhiker's guide}}.
\newblock Springer, Berlin, 3rd edn.

\bibitem[\protect\citeauthoryear{Angeletos, Hellwig and
  Pavan}{2006}]{angeletos-06}
Angeletos, George-Marios, Christian Hellwig and Alessandro Pavan (2006).
\newblock \href{https://economics.mit.edu/files/338}{``Signaling in a global
  game: Coordination and policy traps''}.
\newblock \emph{Journal of Political Economy}, \textbf{114}, 452--484.

\bibitem[\protect\citeauthoryear{Angeletos, Hellwig and
  Pavan}{2007}]{angeletos-07}
Angeletos, George-Marios, Christian Hellwig and Alessandro Pavan (2007).
\newblock \href{https://economics.mit.edu/files/343}{``Dynamic global games of
  regime change: Learning, multiplicity and timing of attacks''}.
\newblock \emph{Econometrica}, \textbf{75}, 711--756.

\bibitem[\protect\citeauthoryear{Baliga and Sj\"{o}str\"{o}m}{2012}]{baliga-12}
Baliga, Sandeep and Tjomas Sj\"{o}str\"{o}m (2012).
\newblock
  \href{https://www.aeaweb.org/articles?id=10.1257/aer.102.6.2897}{``The
  strategy of manipulating conflict''}.
\newblock \emph{American Economic Review}, \textbf{102}, 2897--2922.

\bibitem[\protect\citeauthoryear{Battigalli}{1997}]{battigalli-97}
Battigalli, Pierpaolo (1997).
\newblock
  \href{https://www.dropbox.com/s/bq12psrek67v4r9/battigalli-97.pdf}{``{O}n
  rationalizability on extensive-form games''}.
\newblock \emph{Journal of Economic Theory}, \textbf{74}, 40--61.

\bibitem[\protect\citeauthoryear{Battigalli and {De
  Vito}}{2018}]{battigalli-18}
Battigalli, Pierpaolo and Nicodemo {De Vito} (2018).
\newblock
  \href{http://didattica.unibocconi.it/mypage/upload/48808_20190611_104553_BATTIDEVITO2018_WP-629.PDF}{``Beliefs,
  plans, and perceived intentions in dynamic games''}.
\newblock IGIER Working Paper Series \#629.

\bibitem[\protect\citeauthoryear{Battigalli and
  Siniscalchi}{2002}]{battigalli-02}
Battigalli, Pierpaolo and Marciano Siniscalchi (2002).
\newblock
  \href{https://www.dropbox.com/s/5efnmlndl9x8drn/battigalli-siniscalchi-02.pdf}{``{S}trong
  belief and forward induction reasoning''}.
\newblock \emph{Journal of Economic Theory}, \textbf{106}, 356--391.

\bibitem[\protect\citeauthoryear{Battigalli and
  Siniscalchi}{2003}]{battigalli-03}
Battigalli, Pierpaolo and Marciano Siniscalchi (2003).
\newblock
  \href{http://faculty.wcas.northwestern.edu/~msi661/bs03iig.pdf}{``Rationalization
  and incomplete information''}.
\newblock \emph{The B.E. Journal of Theoretical Economics}, \textbf{3}, 1--46.

\bibitem[\protect\citeauthoryear{Battigalli and
  Siniscalchi}{2007}]{battigalli-07}
Battigalli, Pierpaolo and Marciano Siniscalchi (2007).
\newblock
  \href{https://www.dropbox.com/s/ajbo7qsv5onv6zy/battigalli-siniscalchi-07.pdf}{``{I}nteractive
  epistemology in games with payoff uncertainty''}.
\newblock \emph{Research in Economics}, \textbf{61}, 165--184.

\bibitem[\protect\citeauthoryear{Brandenburger and
  Dekel}{1993}]{brandenburger-93}
Brandenburger, Adam and Eddie Dekel (1993).
\newblock
  \href{https://www.dropbox.com/s/utawopb352qc49w/brandenburger-dekel-93.pdf}{``{H}ierarchies
  of beliefs and common knowledge''}.
\newblock \emph{Journal of Economic Theory}, \textbf{59}, 189--198.

\bibitem[\protect\citeauthoryear{Carlsson and van Damme}{1993}]{carlsson-93}
Carlsson, Hans and Eric van Damme (1993).
\newblock \href{https://www.jstor.org/stable/2951491}{``Global games and
  equilibrium selection''}.
\newblock \emph{Econometrica}, \textbf{61}, 989--1018.

\bibitem[\protect\citeauthoryear{Catonini}{2019}]{catonini-19}
Catonini, Emiliano (2019).
\newblock
  \href{https://www.sciencedirect.com/science/article/pii/S089982561830201X}{``Rationalizability
  and epistemic priority orderings''}.
\newblock \emph{Games and Economic Behavior}, \textbf{114}, 101--117.

\bibitem[\protect\citeauthoryear{Catonini}{2020}]{catonini-20}
Catonini, Emiliano (2020).
\newblock
  \href{https://www.sciencedirect.com/science/article/pii/S0899825620300051}{``On
  non-monotonic strategic reasoning''}.
\newblock \emph{Games and Economic Behavior}, \textbf{120}, 209--224.

\bibitem[\protect\citeauthoryear{Chen and Micali}{2013}]{chen-13}
Chen, Jing and Silvio Micali (2013).
\newblock
  \href{http://onlinelibrary.wiley.com/doi/10.3982/TE942/abstract}{``The order
  independence of iterated dominance in extensive games''}.
\newblock \emph{Theoretical Economics}, \textbf{8}, 25--63.

\bibitem[\protect\citeauthoryear{Chen}{2012}]{chen-12}
Chen, Yi-Chun (2012).
\newblock
  \href{http://www.sciencedirect.com/science/article/pii/S0899825612000267}{``A
  structure theorem for rationalizability in the normal form of dynamic
  games''}.
\newblock \emph{Games and Economic Behavior}, \textbf{75}, 587--597.

\bibitem[\protect\citeauthoryear{Chen, {Di Tillio}, Faingold and
  Xiong}{2010}]{chen-10}
Chen, Yi-Chun, Alfredo {Di Tillio}, Eduardo Faingold and Siyang Xiong (2010).
\newblock
  \href{https://drive.google.com/file/d/0B95AwlwATu7wellreHpvRm91YW8/edit?usp=sharing}{``Uniform
  topologies on types''}.
\newblock \emph{Theoretical Economics}, \textbf{5}, 445--478.

\bibitem[\protect\citeauthoryear{Chen, Mueller-Frank and Pai}{2020}]{chen-20}
Chen, Yi-Chun, Manuel Mueller-Frank and Mallesh~M. Pai (2020).
\newblock
  \href{http://yichun.weebly.com/uploads/3/6/2/9/3629507/continuousimplementationsimple.pdf}{``Continuous
  implementation with direct revelation mechanisms''}.
\newblock Mimeo.

\bibitem[\protect\citeauthoryear{Chen, Tillio, Faingold and
  Xiong}{2017}]{chen-17}
Chen, Yi-Chun, Alfredo~Di Tillio, Eduardo Faingold and Siyang Xiong (2017).
\newblock
  \href{http://yichun.weebly.com/uploads/3/6/2/9/3629507/2016_08_10_impact.pdf}{``Characterizing
  the stratrategic impact of misspecified beliefs''}.
\newblock \emph{The Review of Economic Studies}, \textbf{84}, 1424--1471.

\bibitem[\protect\citeauthoryear{Cho and Kassa}{2017}]{cho-17}
Cho, In-Koo and Kenneth Kassa (2017).
\newblock
  \href{https://research.stlouisfed.org/publications/review/2017/07/05/model-averaging-and-persistent-disagreement/}{``Model
  averaging and persistent disagreement''}.
\newblock \emph{Federal Reserve Bank of St. Louis Review,}, \textbf{99},
  279--294.

\bibitem[\protect\citeauthoryear{Chung and Ely}{2007}]{chung-07}
Chung, Kim-Sau and Jeffrey~C. Ely (2007).
\newblock \href{https://www.jstor.org/stable/4626147?seq=1}{``Foundations of
  dominant-strategy mechanisms''}.
\newblock \emph{Review of Economic Studies}, \textbf{74}, 447--476.

\bibitem[\protect\citeauthoryear{Dekel, Fudenberg and Morris}{2006}]{dekel-06}
Dekel, Eddie, Drew Fudenberg and Stephen Morris (2006).
\newblock
  \href{http://www.princeton.edu/~smorris/pdfs/Morris-TopologiesOnTypes.pdf}{``Topologies
  on types''}.
\newblock \emph{Theoretical Economics}, \textbf{1}, 275--309.

\bibitem[\protect\citeauthoryear{Dekel, Fudenberg and Morris}{2007}]{dekel-07}
Dekel, Eddie, Drew Fudenberg and Stephen Morris (2007).
\newblock
  \href{http://econtheory.org/ojs/index.php/te/article/viewFile/20070015/1059}{``Interim
  correlated rationalizability''}.
\newblock \emph{Theoretical Economics}, \textbf{2}, 15--40.

\bibitem[\protect\citeauthoryear{Dekel, Lipman and Rustichini}{1998}]{dekel-98}
Dekel, Eddie, Barton~L. Lipman and Aldo Rustichini (1998).
\newblock
  \href{https://www.jstor.org/stable/2998545?origin=crossref\&seq=1}{``Standard
  state-space models preclude unawareness''}.
\newblock \emph{Econometrica}, \textbf{66}, 159--173.

\bibitem[\protect\citeauthoryear{Dye}{1985}]{dye-85}
Dye, Ronal~A. (1985).
\newblock \href{https://www.jstor.org/stable/2490910?seq=1}{``Disclosure of
  nonproprietary information''}.
\newblock \emph{Journal of Accounting Research}, \textbf{23}, 123--145.

\bibitem[\protect\citeauthoryear{Ely and P\c{e}ski}{2011}]{ely-11}
Ely, Jeffrey and Marcin P\c{e}ski (2011).
\newblock \href{https://www.jstor.org/stable/23015835?seq=1}{``Critical
  types''}.
\newblock \emph{Review of Economic Studies}, \textbf{78}, 907--937.

\bibitem[\protect\citeauthoryear{Esponda and Pouzo}{2016}]{esponda-16}
Esponda, Ignacio and Demian Pouzo (2016).
\newblock
  \href{https://onlinelibrary.wiley.com/doi/abs/10.3982/ECTA12609}{``Berk-Nash
  equilibrium: A framework for modeling agents with misspecified models''}.
\newblock \emph{Econometrica}, \textbf{84}, 1093--1130.

\bibitem[\protect\citeauthoryear{Fagin and Halpern}{1988}]{fagin1988belief}
Fagin, Ronald and Joseph~Y Halpern (1988).
\newblock Belief, awareness, and limited reasoning.
\newblock \emph{Artificial intelligence}, \textbf{34}(1), 39--76.

\bibitem[\protect\citeauthoryear{Germano, Weinstein and
  Zuazo-Garin}{2020}]{germano-20}
Germano, Fabrizio, Jonathan Weinstein and Peio Zuazo-Garin (2020).
\newblock
  \href{https://econtheory.org/ojs/index.php/te/article/view/20200089}{``Uncertain
  rationality, depth of reasoning and robustness in games with incomplete
  information''}.
\newblock \emph{Theoretical Economics}, \textbf{15}, 89--122.

\bibitem[\protect\citeauthoryear{Guarino}{2020}]{guarino-20}
Guarino, Pierfrancesco (2020).
\newblock
  \href{https://www.sciencedirect.com/science/article/pii/S0899825619301526}{``An
  epistemic analysis of dynamic games with unawareness''}.
\newblock \emph{Games and Economic Behavior}, \textbf{120}, 257--288.

\bibitem[\protect\citeauthoryear{Halpern and
  Piermont}{2019}]{halpern2019partial}
Halpern, Joseph~Y and Evan Piermont (2019).
\newblock Partial awareness.
\newblock In \emph{Proceedings of the AAAI Conference on Artificial
  Intelligence}, vol.~33, pp. 2851--2858.

\bibitem[\protect\citeauthoryear{Halpern and R\^{e}go}{2009}]{halpern-09}
Halpern, Joseph~Y. and Leandro~C. R\^{e}go (2009).
\newblock
  \href{https://www.sciencedirect.com/science/article/abs/pii/S089982560900027X}{``Reasoning
  about knowledge of unawareness''}.
\newblock \emph{Games and Economic Behavior}, \textbf{67}, 503--525.

\bibitem[\protect\citeauthoryear{Han and Kyle}{2017}]{han-17}
Han, Jungsuk and Albert~S. Kyle (2017).
\newblock
  \href{https://pubsonline.informs.org/doi/abs/10.1287/mnsc.2017.2759?fbclid=IwAR2-n-mOX5BiF8iKpYdEJVR22r1wxOWO3cNjJuht4x-9yiVunGvTBosKJVU\&journalCode=mnsc}{``Speculative
  equilibrium with differences in higher-order beliefs}.
\newblock \emph{Management Science}, \textbf{69}, 4317--4332.

\bibitem[\protect\citeauthoryear{Hansen and Sargent}{2001}]{hansen-01}
Hansen, Lars~Peter and Thomas~J. Sargent (2001).
\newblock \href{https://www.jstor.org/stable/2677734?seq=1}{``Robust control
  and model uncertainty''}.
\newblock \emph{American Economic Review}, \textbf{91}, 60--66.

\bibitem[\protect\citeauthoryear{Harsanyi}{1967--1968}]{harsanyi-67}
Harsanyi, John~C. (1967--1968).
\newblock \href{http://www.dklevine.com/archive/refs41175.pdf}{``Games with
  incomplete information played by `{B}ayesian' players, {I}--{III}''}.
\newblock \emph{Management Science}, \textbf{14}, 159--182, 320--334, 486--502.

\bibitem[\protect\citeauthoryear{Heifetz and Kets}{2018}]{heifetz-18a}
Heifetz, Aviad and Willemien Kets (2018).
\newblock
  \href{https://onlinelibrary.wiley.com/doi/abs/10.3982/TE2098}{``Robust
  multiplicity with a grain of naivet\'{e}}.
\newblock \emph{Theoretical Economics}, \textbf{13}, 415--465.

\bibitem[\protect\citeauthoryear{Heifetz, Meier and
  Schipper}{2006}]{heifetz-06}
Heifetz, Aviad, Martin Meier and Burkhard Schipper (2006).
\newblock
  \href{https://econpapers.repec.org/article/eeejetheo/v_3a130_3ay_3a2006_3ai_3a1_3ap_3a78-94.htm}{``Interactive
  unawareness''}.
\newblock \emph{Journal of Economic Theory}, \textbf{130}, 78--94.

\bibitem[\protect\citeauthoryear{Heifetz, Meier and
  Schipper}{2008}]{heifetz-08}
Heifetz, Aviad, Martin Meier and Burkhard Schipper (2008).
\newblock
  \href{https://www.sciencedirect.com/science/article/abs/pii/S0899825607001030}{``A
  canonical model for interactive unawareness''}.
\newblock \emph{Games and Economic Behavior}, \textbf{62}, 304--324.

\bibitem[\protect\citeauthoryear{Heifetz, Meier and
  Schipper}{2014}]{heifetz-14}
Heifetz, Aviad, Martin Meier and Burkhard Schipper (2014).
\newblock
  \href{https://econpapers.repec.org/article/eeegamebe/v_3a81_3ay_3a2013_3ai_3ac_3ap_3a50-68.htm}{``Dynamic
  unawareness and rationalizable behavior''}.
\newblock \emph{Games and Economic Behavior}, \textbf{81}, 50--68.

\bibitem[\protect\citeauthoryear{Heifetz and Perea}{2015}]{heifetz-15}
Heifetz, Aviad and Andr\'{e}s Perea (2015).
\newblock \href{http://www.epicenter.name/Perea/Papers/FIBI-paths.pdf}{``On the
  outcome equivalence of backward induction and extensive form
  rationalizability''}.
\newblock \emph{International Journal of Game Theory}, \textbf{44}, 37--59.

\bibitem[\protect\citeauthoryear{Inostroza and Pavan}{2018}]{inostroza-18}
Inostroza, Nicolas and Alessandro Pavan (2018).
\newblock
  \href{https://cpb-us-e1.wpmucdn.com/sites.northwestern.edu/dist/4/2552/files/2018/10/persuasion-GG-April-4-2018-1e1h2rz.pdf}{``''Persuasion
  in global games with application to stress testing''}.
\newblock Mimeo.

\bibitem[\protect\citeauthoryear{Jimenez-Gomez}{2019}]{jimenez-gomez-19}
Jimenez-Gomez, David (2019).
\newblock
  \href{https://papers.ssrn.com/sol3/papers.cfm?abstract_id=3216040}{``False
  consensus in games: Embedding level-$k$ models into games of incomplete
  information''}.
\newblock Mimeo.

\bibitem[\protect\citeauthoryear{Kennan}{2001}]{kennan-01}
Kennan, John (2001).
\newblock \href{https://www.jstor.org/stable/2695907?seq=1}{``Repeated
  bargaining with persistent private information''}.
\newblock \emph{Review of Economic Studies}, \textbf{68}, 719--755.

\bibitem[\protect\citeauthoryear{Mertens and Zamir}{1985}]{mertens-85}
Mertens, Jean-Fran\c{c}ois and Shmuel Zamir (1985).
\newblock
  \href{http://www.ma.huji.ac.il/~zamir/papers/22_IJGT85.pdf}{``Formulation of
  {B}ayesian analysis for games with incomplete information''}.
\newblock \emph{International Journal of Game Theory}, \textbf{14}, 1--29.

\bibitem[\protect\citeauthoryear{Modica and
  Rustichini}{1994}]{modica1994awareness}
Modica, Salvatore and Aldo Rustichini (1994).
\newblock Awareness and partitional information structures.
\newblock \emph{Theory and decision}, \textbf{37}(1), 107--124.

\bibitem[\protect\citeauthoryear{Morris, Shin and Yildiz}{2016}]{morris-16}
Morris, Stepehen, Hyung~Song Shin and Muhamet Yildiz (2016).
\newblock \href{https://economics.mit.edu/files/11434}{``Common belief
  foundations of global games''}.
\newblock \emph{Journal of Economic Theory}, \textbf{163}, 826--848.

\bibitem[\protect\citeauthoryear{Morris and Shin}{1998}]{morris-98}
Morris, Stephen and Hyun~Song Shin (1998).
\newblock \href{http://www.sfu.ca/~kkasa/morris%26shin.pdf}{``Unique
  Equilibrium in a Model of Self-Fulfilling Currency Attacks''}.
\newblock \emph{American Economic Review}, \textbf{88}, 587--597.

\bibitem[\protect\citeauthoryear{Murayama}{2020}]{murayama-20}
Murayama, Kota (2020).
\newblock
  \href{https://link.springer.com/article/10.1007/s42973-019-00005-y}{``Robust
  predictions under finite depth of reasoning''}.
\newblock \emph{The Japanese Economic Review}, \textbf{71}, 59--84.

\bibitem[\protect\citeauthoryear{Oury and Tercieux}{2012}]{oury-12}
Oury, Marion and Olivier Tercieux (2012).
\newblock
  \href{https://onlinelibrary.wiley.com/doi/abs/10.3982/ECTA8577}{``Continuous
  implementation''}.
\newblock \emph{Econometrica}, \textbf{80}, 1605--1637.

\bibitem[\protect\citeauthoryear{Pearce}{1984}]{pearce-84}
Pearce, David~G. (1984).
\newblock
  \href{https://www.dropbox.com/s/tzfhnlr9ix33xz1/pearce-84.pdf}{``{R}ationalizable
  strategic behavior and the problem of perfection''}.
\newblock \emph{Econometrica}, \textbf{52}, 1029--1050.

\bibitem[\protect\citeauthoryear{Penta}{2011}]{penta-11}
Penta, Antonio (2011).
\newblock
  \href{https://www.dropbox.com/s/c5fdyozgy3iteln/penta-11.pdf}{``{B}ackward
  induction reasoning in games with incomplete information''}.
\newblock Mimeo. University of Winsonsin-Madison.

\bibitem[\protect\citeauthoryear{Penta}{2012}]{penta-12}
Penta, Antonio (2012).
\newblock
  \href{https://www.dropbox.com/s/sjdp5j8g32jw35b/penta-12.pdf}{``{H}igher
  order uncertainty and information: Static and dynamic games''}.
\newblock \emph{Econometrica}, \textbf{80}, 631--660.

\bibitem[\protect\citeauthoryear{Penta}{2015}]{penta-15}
Penta, Antonio (2015).
\newblock
  \href{https://www.sciencedirect.com/science/article/pii/S0022053115001829}{``Robust
  dynamic implementation''}.
\newblock \emph{Journal of Economic Theory}, \textbf{160}, 280--316.

\bibitem[\protect\citeauthoryear{Penta and Zuazo-Garin}{2021}]{penta-21}
Penta, Antonio and Peio Zuazo-Garin (2021).
\newblock
  \href{https://www.barcelonagse.eu/sites/default/files/working_paper_pdfs/1106.pdf}{``Rationalizability,
  observability and common knolwedge}.
\newblock \emph{Review of Economic Studies}.
\newblock (forthcoming).

\bibitem[\protect\citeauthoryear{Perea}{2014}]{perea-14}
Perea, Andr\'es (2014).
\newblock
  \href{http://www.personeel.unimaas.nl/a.perea/Papers/FutureRat.pdf}{``{B}elief
  in the opponent's future rationality''}.
\newblock \emph{Games and Economic Behavior}, \textbf{81}, 235--254.

\bibitem[\protect\citeauthoryear{Perea}{2018{\natexlab{a}}}]{perea-18}
Perea, Andr\'{e}s (2018{\natexlab{a}}).
\newblock
  \href{http://www.epicenter.name/Perea/Papers/Battigalli-theorem.pdf}{``Why
  forward induction leads to the backward induction outcome: A new proof for
  Battigalli's theorem''}.
\newblock \emph{Games and Economic Behavior}, \textbf{110}, 120--138.

\bibitem[\protect\citeauthoryear{Perea}{2018{\natexlab{b}}}]{perea-18c}
Perea, Andr\'{e}s (2018{\natexlab{b}}).
\newblock
  \href{http://www.epicenter.name/Perea/Papers/CBR-Unawareness.pdf}{``Common
  belief in rationality in games with unawareness''}.
\newblock Mimeo (previous version, Epicenter Working Paper No.~13).

\bibitem[\protect\citeauthoryear{Perea}{2018{\natexlab{c}}}]{perea-18b}
Perea, Andr\'{e}s (2018{\natexlab{c}}).
\newblock
  \href{http://www.epicenter.name/Perea/Papers/Order-ind-dynamic.pdf}{``Order
  independence in dynamic games''}.
\newblock Mimeo (previous version, Epicenter Working Paper No. 8).

\bibitem[\protect\citeauthoryear{Piermont}{2019}]{piermont2019algebraic}
Piermont, Evan (2019).
\newblock \href{https://arxiv.org/abs/1910.07275}{``Algebraic semantics for
  propositional awareness logics''}.
\newblock \emph{arXiv preprint arXiv:1910.07275}.

\bibitem[\protect\citeauthoryear{Rubinstein}{1989}]{rubinstein-89}
Rubinstein, Ariel (1989).
\newblock \href{https://www.dropbox.com/s/u4tj6bdvs94hll0/rubinstein-89.pdf}{
  ``The electronic mail game: strategic behaviour under `almost common
  knowledge'''}.
\newblock \emph{American Economic Review}, \textbf{79}, 385--391.

\bibitem[\protect\citeauthoryear{{Ruiz G.}}{2018}]{ruiz-g-18}
{Ruiz G.}, David (2018).
\newblock
  \href{https://drive.google.com/file/d/1r72kEcwqxrFhn__P0AOyyWiVapWj98xC/view}{``Critical
  types in dynamic games''}.
\newblock Mimeo.

\bibitem[\protect\citeauthoryear{Sannikov and Skrzypack}{2016}]{sannikov-16}
Sannikov, Yuliy and Andrzej Skrzypack (2016).
\newblock \href{https://web.stanford.edu/~skrz/Dynamic_Trading.pdf}{``Dynamic
  trading: Price inertia and front-running''}.
\newblock Mimeo.

\bibitem[\protect\citeauthoryear{Strzalecki}{2014}]{strzalecki-14}
Strzalecki, Tomasz (2014).
\newblock
  \href{https://www.dropbox.com/s/zw51k2cq9908441/strzalecki-10.pdf}{``{D}epth
  of reasoning and higher order beliefs''}.
\newblock \emph{Journal of Economic Behavior and Organization}, \textbf{108},
  108--122.
\newblock Mimeo.

\bibitem[\protect\citeauthoryear{Suetonius}{121}]{suetonious-07}
Suetonius, {Gaius Tranquillus} (121).
\newblock
  \emph{\href{https://www.amazon.es/Twelve-Caesars-Penguin-Classics/dp/0140455167}{``The
  twelve Caesars''}}.
\newblock Peng\"{u}in Classics.

\bibitem[\protect\citeauthoryear{Weinstein and Yildiz}{2007}]{weinstein-07}
Weinstein, Jonathan and Muhamet Yildiz (2007).
\newblock
  \href{https://www.dropbox.com/s/k0mx980tkfx05ld/weinstein-yildiz-11.pdf}{``A
  structure theorem for rationalizability with application to robust
  predictions of refinements''}.
\newblock \emph{Econometrica}, \textbf{75}, 365--400.

\bibitem[\protect\citeauthoryear{Wilson}{1987}]{wilson-87}
Wilson, Robert (1987).
\newblock ``{G}ame-theoretic approaches to trading proceses''.
\newblock In Truman~F. Bewley (editor),
  \emph{\href{https://www.cambridge.org/es/academic/subjects/economics/economics-general-interest/advances-economic-theory-fifth-world-congress?format=PB\&isbn=9780521389259}{Advances
  in economic theory: Fifth world congress}}. Cambridge University Press.

\bibitem[\protect\citeauthoryear{Ziegler}{2019}]{ziegler-19}
Ziegler, Gabriel (2019).
\newblock \href{https://zieglergabriel.github.io/papers/ABID.pdf}{``Adversarial
  bilateral information design''}.
\newblock Mimeo.

\end{thebibliography}

\end{document}